%% file: main.tex
\documentclass[sigconf, nonacm]{acmart}

\newcommand\vldbdoi{XX.XX/XXX.XX}

\newcommand\vldbavailabilityurl{https://github.com/yihaoh/I-Rex}

\usepackage{graphicx}
\usepackage{balance}  
\usepackage{booktabs}
\usepackage{url}
\usepackage{listings}
\usepackage{enumitem}
\usepackage{tabularx}
\usepackage{tabu}
\usepackage{xcolor}
\usepackage{bm}
\usepackage{clrscode3e}
\usepackage[noend]{algpseudocode}
\usepackage{standalone} 
\usepackage{colortbl}
\usepackage{forest}
\usepackage{algorithm}
\usepackage{xspace}
\usepackage{url}
\usepackage{textcomp}
\usepackage{multirow}
\usepackage{booktabs}
\usepackage{hyperref}
\usepackage{color}
\usepackage{tcolorbox}
\usepackage{xpatch}
\usepackage[skip=0.2pt]{subcaption}
\usepackage[nounderscore]{syntax}
\usepackage{tikz-qtree,tikz-qtree-compat}
\usepackage{tikz}
\usepackage{cleveref}
\usetikzlibrary{arrows.meta,bending,quotes}
\usepackage{ulem}
\usepackage{array}
\usepackage{pifont}
\usepackage{xr}
\usepackage[dvipsnames]{xcolor}
\tikzset{every tree node/.style={align=center, anchor=north}}

\definecolor{dkgreen}{rgb}{0,0.6,0}
\definecolor{gray}{rgb}{0.5,0.5,0.5}
\definecolor{mauve}{rgb}{0.58,0,0.82}
\AtBeginDocument{%
  \providecommand\BibTeX{{%
    \normalfont B\kern-0.5em{\scshape i\kern-0.25em b}\kern-0.8em\TeX}}}

\setcopyright{None}
\copyrightyear{2018}
\acmYear{2018}
\acmDOI{10.1145/1122445.1122456}

\acmConference[Woodstock '18]{Woodstock '18: ACM Symposium on Neural
  Gaze Detection}{June 03--05, 2018}{Woodstock, NY}
\acmBooktitle{Woodstock '18: ACM Symposium on Neural Gaze Detection,
  June 03--05, 2018, Woodstock, NY}
\acmPrice{15.00}
\acmISBN{978-1-4503-XXXX-X/18/06}

\input{helpers/macros.tex}
\input{helpers/lstdefs.tex}

\renewcommand\footnotetextcopyrightpermission[1]{} 

\begin{document}



\title{\oursystitle: An Interactive Debugger for SQL}
\settopmatter{authorsperrow=5}

\author{Yihao Hu}
\email{yihao.hu@duke.edu}
\affiliation{
  \institution{Duke University}
}

\author{Zian Chen}
\email{zian.chen@duke.edu}
\affiliation{
  \institution{Duke University}
}

\author{Zhiming Leong}
\authornote{Work completed at Duke prior to joining respective companies/institutions.}
\email{jlzeong@stripe.com}
\affiliation{
  \institution{Stripe}
}

\author{Sharan Sokhi}
\authornotemark[1]
\email{sokhish@amazon.com}
\affiliation{
  \institution{Amazon}
}

\author{Zachary Zheng}
\authornotemark[1]
\email{zzackzheng@meta.com}
\affiliation{
  \institution{Meta}
}

\author{Alex Chao}
\authornotemark[1]
\email{alexchao@ucsd.edu}
\affiliation{
  \institution{University of California, San Diego}
}

\author{Kristin Stephens-Martinez}
\email{ksm@cs.duke.edu}
\affiliation{
  \institution{Duke University}
}

\author{Sudeepa Roy}
\email{sudeepa@cs.duke.edu}
\affiliation{
  \institution{Duke University}
}

\author{Jun Yang}
\email{junyang@cs.duke.edu}
\affiliation{
  \institution{Duke University}
}


\begin{abstract}
SQL is declarative in nature and rich in its features.
Writing semantically correct SQL queries and finding logical bugs in SQL are not easy, even for experienced programmers,
who are often used to the mindset of working with general-purpose programming languages (GPLs).
While there are many GPL debuggers, SQL debugging has received much less attention.
In this paper, we present \oursys, a SQL debugger that enables users to inspect the logical execution of SQL queries visually and interactively
to identify and potentially fix logical bugs in the queries.
\oursys\ draws analogies to the debugging paradigm of GPLs (e.g., stepping, watchpoints, etc.), making it easier for programmers to adopt.
However, unlike debugging GPLs, which involves executing the underlying program in full to the point of interest,
\oursys\ allows users to jump to arbitrary points of interest by leveraging the power of the database systems,
through selective materialization and query rewrites.
To simplify deployment, \oursys\ acts as a lightweight middleware on top of the database system; it imposes no overhead to prepare a database for debugging
and maintains no state in the database systems during debugging sessions.
We demonstrate the effectiveness of \oursys\ through performance experiments as well as a user study in an educational setting.
\end{abstract}

\maketitle


\input{sections/1-intro}

\input{sections/1-1-example}
\input{sections/2-debug-paradigm}
\input{sections/3-system}
\input{sections/4-experiments}

\input{sections/5-user-study}
\input{sections/6-related-work}
\input{sections/7-conclusion}



\clearpage
\balance
\bibliographystyle{ACM-Reference-Format}
\bibliography{main}

\clearpage
\input{sections/appendix/appendix}

\end{document}

%% file: helpers/macros.tex
\newcommand{\narrow}[1]{\ensuremath{\text{\scalebox{0.7}[1.0]{\textsf{#1}}}}}
\newcommand{\marrow}[1]{\ensuremath{\text{\scalebox{0.8}[1.0]{\textsf{#1}}}}}

\newcommand{\sql}[1]{\marrow{\texttt{#1}}}

\newcommand{\SELECT}{\sql{SELECT}}
\newcommand{\FROM}{\sql{FROM}}
\newcommand{\WHERE}{\sql{WHERE}}
\newcommand{\GROUPBY}{\sql{GROUP} \sql{BY}}
\newcommand{\HAVING}{\sql{HAVING}}
\newcommand{\WITH}{\sql{WITH}}

\newcommand{\JOIN}{\sql{JOIN}}

\newcommand{\UNION}{\sql{UNION}}
\newcommand{\EXCEPT}{\sql{EXCEPT}}
\newcommand{\INTERSECT}{\sql{INTERSECT}}

\newcommand{\LIMIT}{\sql{LIMIT}}
\newcommand{\OFFSET}{\sql{OFFSET}}

\newcommand{\mseq}{\ensuremath{\mathrel{\smash{\overset{\lower.5em\hbox{$\scriptscriptstyle\Box$}}{=}}}}}

\definecolor{black}{rgb}{0,0,0}
\definecolor{grey}{rgb}{0.8,0.8,0.8}
\definecolor{red}{rgb}{1,0,0}
\definecolor{green}{rgb}{0,1,0}
\definecolor{applegreen}{rgb}{0.55, 0.71, 0.0}
\definecolor{darkgreen}{rgb}{0,0.5,0}
\definecolor{darkpurple}{rgb}{0.5,0,0.5}
\definecolor{darkdarkpurple}{rgb}{0.3,0,0.3}
\definecolor{blue}{rgb}{0,0,1}
\definecolor{shadegreen}{rgb}{0.95,1,0.95}
\definecolor{shadeblue}{rgb}{0.95,0.95,1}
\definecolor{shadered}{rgb}{1,0.85,0.85}
\definecolor{shadegrey}{rgb}{0.85,0.85,0.85}
\definecolor{oddRowGrey}{rgb}{0.80,0.80,0.80}
\definecolor{evenRowGrey}{rgb}{0.85,0.85,0.85}

\newcommand{\cut}[1]{{}}

\newcommand{\oursys}{\textsc{I-Rex}}
\newcommand{\oursystitle}{I-Rex}

\newcommand{\mypar}[1]{\smallskip\noindent\textbf{{#1}.}}

\DeclareMathAlphabet{\mathbbold}{U}{bbold}{m}{n}

\newtheorem{Example}{Example}

\makeatletter
\newcommand*{\shifttext}[2]{%
  \settowidth{\@tempdima}{#2}%
  \makebox[0pt]{\hspace*{#1}#2}%
}
\makeatother

%% file: helpers/lstdefs.tex
\definecolor{lstpurple}{rgb}{0.5,0,0.5}
\definecolor{lstred}{rgb}{1,0,0}
\definecolor{lstreddark}{rgb}{0.7,0,0}
\definecolor{lstredl}{rgb}{0.64,0.08,0.08}
\definecolor{lstmildblue}{rgb}{0.66,0.72,0.78}
\definecolor{lstblue}{rgb}{0,0,1}
\definecolor{lstmildgreen}{rgb}{0.42,0.53,0.39}
\definecolor{lstgreen}{rgb}{0,0.5,0}
\definecolor{lstorangedark}{rgb}{0.6,0.3,0}	
\definecolor{lstorange}{rgb}{0.75,0.52,0.005}
\definecolor{lstorangelight}{rgb}{0.89,0.81,0.67}
\definecolor{lstbeige}{rgb}{0.90,0.86,0.45}

\DeclareFontShape{OT1}{cmtt}{bx}{n}{<5><6><7><8><9><10><10.95><12><14.4><17.28><20.74><24.88>cmttb10}{}

\lstdefinelanguage{smtlib2}{
  alsoletter=-,
  morekeywords={declare-const,define-fun,assert,minimize,maximize,check-sat,get-objectives,and,or,not,distinct},
  extendedchars=false,
  keywordstyle=\bfseries\color{lstpurple},
  deletekeywords={Int,Bool},
  keywords=[2]{Int,Bool},
  keywordstyle=[2]\color{lstblue},
}

\lstdefinestyle{psql}
{
tabsize=2,
basicstyle=\scriptsize\upshape\ttfamily,
language=SQL,
morekeywords={PROVENANCE,BASERELATION,INFLUENCE,COPY,ON,TRANSPROV,TRANSSQL,TRANSXML,CONTRIBUTION,COMPLETE,TRANSITIVE,NONTRANSITIVE,EXPLAIN,SQLTEXT,GRAPH,IS,ANNOT,THIS,XSLT,MAPPROV,cxpath,OF,TRANSACTION,SERIALIZABLE,COMMITTED,INSERT,INTO,WITH,SCN,UPDATED},
extendedchars=false,
keywordstyle=\bfseries,
mathescape=true,
escapechar=@,
sensitive=true
}

\lstdefinestyle{psqlcolor}
{
tabsize=2,
basicstyle=\scriptsize\upshape\ttfamily,
language=SQL,
morekeywords={PROVENANCE,BASERELATION,INFLUENCE,COPY,ON,TRANSPROV,TRANSSQL,TRANSXML,CONTRIBUTION,COMPLETE,TRANSITIVE,NONTRANSITIVE,EXPLAIN,SQLTEXT,GRAPH,IS,ANNOT,THIS,XSLT,MAPPROV,cxpath,OF,TRANSACTION,SERIALIZABLE,COMMITTED,INSERT,INTO,WITH,SCN,UPDATED},
extendedchars=false,
keywordstyle=\bfseries\color{lstpurple},
deletekeywords={count,min,max,avg,sum},
keywords=[2]{count,min,max,avg,sum},
keywordstyle=[2]\color{lstblue},
stringstyle=\color{lstreddark},
commentstyle=\color{lstgreen},
mathescape=true,
escapechar=@,
sensitive=true
}

\lstdefinestyle{datalog}
{
basicstyle=\footnotesize\upshape\ttfamily,
language=prolog
}

\lstdefinestyle{pseudocode}
{
  tabsize=3,
  basicstyle=\small,
  language=c,
  morekeywords={if,else,foreach,case,return,in,or},
  extendedchars=true,
  mathescape=true,
  literate={:=}{{$\gets$}}1 {<=}{{$\leq$}}1 {!=}{{$\neq$}}1 {append}{{$\listconcat$}}1 {calP}{{$\cal P$}}{2},
  keywordstyle=\color{lstpurple},
  escapechar=&,
  numbers=left,
  numberstyle=\color{lstgreen}\small\bfseries, 
  stepnumber=1, 
  numbersep=5pt,
}

\lstdefinestyle{xmlstyle}
{
  tabsize=3,
  basicstyle=\small,
  language=xml,
  extendedchars=true,
  mathescape=true,
  escapechar=£,
  tagstyle=\color{keywordpurple},
  usekeywordsintag=true,
  morekeywords={alias,name,id},
  keywordstyle=\color{lstred}
}

\lstdefinestyle{smtlib2}
{
tabsize=2,
basicstyle=\scriptsize\upshape\ttfamily,
numbers=left,
stepnumber=1,
breaklines=true,
stringstyle=\color{lstreddark},
commentstyle=\color{lstgreen},
mathescape=true,
escapechar=@,
sensitive=true
}

\lstdefinestyle{TinyJSON}{
    language=,
    basicstyle=\ttfamily\small, 
    keywordstyle=\color{mygreen}\bfseries, 
    stringstyle=\color{purple}, 
    commentstyle=\color{gray}\itshape, 
    numbers=none, 
    numberstyle=\tiny\color{gray}, 
    stepnumber=1, 
    numbersep=10pt, 
    tabsize=2, 
    captionpos=b, 
    breaklines=true, 
    breakatwhitespace=true, 
    showspaces=false, 
    showstringspaces=false, 
    escapeinside={(*@}{@*)}, 
    frame=none, 
    aboveskip=0pt, 
    belowskip=0pt, 
    extendedchars=true,
    literate={á}{{\'a}}1 {ã}{{\~a}}1 {é}{{\'e}}1 {í}{{\'i}}1 {ó}{{\'o}}1 {ú}{{\'u}}1 {ç}{{\c{c}}}1,
}


%% file: sections/1-intro.tex
\section{Introduction}
\label{sec:intro}

Relational databases form the backbone of many data-intensive applications and scalable data analytics.
Despite its age, SQL continues to retain its prevalence and importance due to its highly \emph{declarative} nature (i.e., specifying what the answer should be rather than how to compute it) and its extensive set of features that have only grown over time. However, SQL is difficult to understand and debug. In debugging general-purpose programming languages (GPLs, e.g., C++ or Python) that are typically \emph{procedural} (i.e., explicitly describing the steps required to compute the answer), it is natural to trace the execution of programs to debug them. However, this method becomes much trickier for SQL. 

As a first attempt, one may consider ``tracing'' a query's logical or physical plan, a tree whose leaves represent base tables and internal nodes represent relational operators. Through this approach, the user can examine the intermediate results produced by each of the plan nodes.
Unfortunately, there are several problems with this approach.
Firstly, the database optimizer often compiles a SQL query into a plan that bears no resemblance to the original query,
making plan tracing unhelpful in finding and fixing logical bugs in the original query.
Secondly, debugging is usually an iterative process: the user may examine the execution multiple times, sometimes with minor modifications to the query.
However, even small changes can lead the optimizer to choose a very different plan; for example, adding or removing even a simple condition in \sql{WHERE} can enable or disable an index-scan opportunity.
Even if the physical plan remains the same, there is no guarantee that the execution and result order are reproducible.
For example, the size of the buffer memory, the choice of the hash function, and variations in the speed of parallel threads at run-time
can change the ordering of intermediate result rows.
Such non-repeatable and seemingly inconsistent behaviors significantly complicate debugging.

Perhaps, one possible workaround would be to restrict the database optimizer to avoid optimization across syntactic blocks of a query,
such that each subquery corresponds to some subtree in the plan, and the user can at least inspect the result of each subquery. However, this may result in inefficient handling of complex queries. Furthermore, correlated subqueries, a frequently used SQL construct, render this workaround ineffective for many queries. In \Cref{sec:debug-scene}, we give concrete examples that illustrate this challenge and demonstrate how \oursys\ helps debug these types of queries.

In this work with \oursys, we focus on finding \emph{logical errors} instead of fixing \emph{performance issues}, and we aim to build an interactive SQL debugger with the following desiderata:

\begin{enumerate}[leftmargin=*, topsep=0pt]
\item The debugger should conceptually execute a SQL query in a completely reproducible manner that is faithful to how it is written and must be easy for programmers to understand.
\item The debugger should offer features analogous to those in GPL debuggers so that they are easy to learn and adopt.
\item Unlike GPL debuggers, which must execute underlying programs in full to reach points of interest,
    this debugger should support efficient implementation of powerful features that allow the user to jump directly to points of interest.
\item The debugger must scale to large databases and gracefully handle prohibitively large intermediate results.
\item The debugger should be simple to run (e.g., from a remote browser) and easy to deploy on top of a database,
    without modifying database system internals or requiring extensive preparation of the database for debugging.
\end{enumerate}

At first glance, (1) necessitates executing the SQL query ``literally'' using a completely unoptimized plan,
which runs counter to (4).
Our key insight is that, at any given point in time, the user can examine only a small ``window'' of the entire execution.
It suffices to support fast access to any given ``window'' without incurring the full cost of the unoptimized plan.
Supporting such accesses, along with (2) and (3) while respecting (5), requires novel optimization --- SQL's built-in \sql{OFFSET} and \sql{LIMIT} constructs fail to deliver acceptable performance for interactive debugging.

To this end, \oursys\ makes the following contributions:
\begin{itemize}[leftmargin=*, topsep=0pt]
\item \oursys\ introduces a novel debugging paradigm for SQL that draws many parallels to GPL debugging. Each query block is viewed as a function, and correlated subqueries are considered functions with arguments.
We define the \emph{canonical execution} of a SQL query,
which is reproducible and faithful to query syntax. 
\item \oursys\ supports various debugging features analogous to GPL debugging, including
    stepping through execution,
    pausing to examine a particular point during execution (breakpoints),
    pausing automatically at points of interest (watchpoints),
    drilling down into subqueries (stepping ``into'' a function),
    and row-level tracing (information flow analysis) in both forward (from input to output) and backward (from output to input) directions.
\item While performing the canonical execution would have been extremely inefficient,
    we designed query optimization techniques for \oursys\ to support fast ``teleporting'' from one point of execution to another without paying the cost of execution in between.
\item \oursys\ has a web-based frontend and a middleware backend that runs on top of a database system, making it easy to adopt and deploy.
    It imposes no overhead to prepare a database for debugging and maintains no state in the database during active debugging.
\item Our performance evaluation using the TPC-H benchmark shows the scalability of \oursys\ on large databases. It demonstrates the advantage of our optimization techniques over standard database (PostgreSQL) support for retrieving windows of query results.
\item We evaluate the efficacy of \oursys\ with a user study of 100+ students in a database course. Its findings indicate that \oursys\ significantly improves students' efficiency in debugging SQL queries. 
\end{itemize}

%% file: sections/1-1-example.tex
\section{Example Use of \oursys\ for Debugging}
\label{sec:debug-scene}

Since \oursys\ conceptually executes SQL queries as they are written,
it is particularly suitable for novices who are learning how SQL queries work on the logical level. It also serves as a powerful tool for novices and data professionals alike to find and fix logical bugs. This section provides a walk-through of how to use \oursys\ to debug a query; we introduce the interface, concepts, and features of \oursys. 

\begin{Example}\itshape
\label{eg:wrong-query}
Consider the toy database in~\Cref{fig:beer-db},
which stores information about beers, bars serving them, and drinkers who like beers and frequent bars.
We want to write a query for the following task:
every time a drinker frequents a bar, they buy one bottle of every beer they like or any beer priced \$2 or lower;
find the expected weekly revenue of each bar and rank them by revenue from high to low.
A user may come up with the following (incorrect) query:
\begin{lstlisting}[language=SQL, mathescape=true, basicstyle=\footnotesize\ttfamily]
SELECT s.bar, SUM(f.times_a_week * s.price) AS revenue -- $Q$
FROM Serves s, Frequents f
WHERE f.bar = s.bar
AND (s.price <= 2 OR EXISTS (
         SELECT * FROM Likes l WHERE f.drinker = l.drinker -- $Q_\narrow{inner}$
     ))
GROUP BY s.bar;
\end{lstlisting}
The above query intends to first find drinkers and beers available for purchase using a join between \sql{Serves} and \sql{Frequents}.
It additionally applies the two (alternative) conditions for purchase:
1)~the price is lower than \$2, and
2)~the beer is liked by the drinker.
Then, the query groups the intermediate results by bar and calculates the sum of revenue.
There is a bug in the \sql{EXISTS} subquery $Q_\narrow{inner}$, but the question for now is: how would a user examine the result of $Q_\narrow{inner}$?

Note that $Q_\narrow{inner}$ is \emph{correlated}, with the value for \sql{f.drinker} coming from the outer (i.e. enclosing) block.
As a result, there is no way to inspect this result independently.
This situation cannot be handled by a query plan with relational operators, where the result of each subtree depends on this subtree alone.

Indeed, most database optimizers will rewrite the above query for execution such that the subquery is decorrelated.
The decorrelated subquery would be a join involving both \sql{Likes} and \sql{Frequents} to compute the original subquery for all possible \sql{drinker} values in a single effort.
Then, the result will be combined with the rest of the outer query.
The new plan now consists of only relational operators and can be computed/debugged in a bottom-up fashion,
but unfortunately, it is vastly different from the original query.
Users without in-depth knowledge of query optimization will likely be confused.
\end{Example}

\begin{figure}[t]\small\setlength{\tabcolsep}{2pt}
\centering
   \subfloat[\small\mdseries \sql{Serves}]{
    \begin{minipage}[b]{0.3\columnwidth}\centering
     {\scriptsize
      \begin{tabular}[b]{|c|c|c|}\hline
         {\tt bar} & {\tt beer} & {\tt price} \\ \hline
         Apex & Corona & 1 \\
         Apex & Dixie  & 2  \\
         Edge & Amstel & 4 \\
         Edge & Corona & 1.5 \\
         Tavern & Amstel & 3 \\
         Tavern & Erdinger & 1 \\ \hline
      \end{tabular}
      }
    \end{minipage}
  }
  \subfloat[\small\mdseries \sql{Frequents}]{
    \begin{minipage}[b]{0.33\columnwidth}\centering
     {\scriptsize
      \begin{tabular}[b]{|c|c|c|}\hline
         {\tt drinker} & {\tt bar} & {\tt times} \\ \hline
         Amy & Apex & 1  \\
         Ben & Edge & 4  \\
         Coy & Tavern & 2  \\ 
         Dan & Edge & 3  \\ \hline
      \end{tabular}
      }
    \end{minipage}
  }
   \subfloat[\small\mdseries \sql{Likes}]{
    \begin{minipage}[b]{0.3\columnwidth}\centering
     {\scriptsize
      \begin{tabular}[b]{|c|c|}\hline
        {\tt drinker} & {\tt beer}  \\ \hline
         Amy & Erdinger \\
         Ben & Budweiser \\
         Ben & Dixie \\
         Coy & Amstel \\
         Dan & Amstel \\
         Dan & Corona \\ \hline
      \end{tabular}
      }
    \end{minipage}
  }
  \caption{\label{fig:beer-db}\mdseries A toy database about beers, bars, and drinkers.}
\end{figure}

\vspace{-3mm}
\begin{Example}\itshape
\label{eg:debug-scene}
This incorrect query described above returns the result shown in the output table in \Cref{fig:debug-scene} for the database in \Cref{fig:beer-db}.


\noindent Based on the user's knowledge of the database instance,
the revenue of bar \sql{Edge} seems to be higher than expected.
We now walk through how to use \oursys\ to debug the query, starting with this observation.

\input{fig/fig-ex1}

\oursys\ presents a panel of UI debugging elements for each block of the query.
\Cref{fig:debug-scene} illustrates the UI for the outer query block $Q$
(for simplicity, we do not show the actual interface here as it contains other details that may be distracting for this discussion).
\oursys\ shows the execution of this block in stages, from top to bottom.
At the very top, \oursys\ shows all input tables in \FROM.
Note that one row from each input table is highlighted;
this combination of input tuples (or ``input combo,'' formally defined in \Cref{sec:debug-paradigm:model}) intuitively defines the current point of execution being examined. 
Then, \oursys\ presents the ``joined \& filtered'' result, which is the intermediate output after the \WHERE\ clause is applied.
The intermediate result row produced by the current input combo is automatically highlighted.
Between this result table and the input tables, a ``filter expression'' tree shows how the \WHERE\ condition evaluates over the current input combo.
The user can examine the value of each subexpression and see how the truth values (color-coded) are combined by logical connectives.
Following the joined \& filtered result, \oursys\ shows the \GROUPBY\ result.
For each group, in addition to the \GROUPBY\ value, each group member's contribution to the final \sql{SUM} aggregate is also shown.
Again, the group and the member that the current input combo contributes to are automatically highlighted.
Finally, the output table shows the final result of the query block.

For the convenience of subsequent discussion, we show a symbolic identifier (e.g., $s_0$, $f_2$, $j_6$...) for each row.
We do assign internal row identifiers, whose purposes will be explained later in \Cref{sec:debug-paradigm},
but they are not explicitly displayed by the UI.

Given the unexpectedly high revenue of \sql{Edge},
the user naturally wants to examine how that output row was computed.
\oursys\ supports tracing backward from output to input using a more general mechanism called \emph{pinning}, denoted here by the red \textcolor{red}{$\dagger$} next to $o_1: \langle \sql{Edge}, 38.5 \rangle$.
A pinned output row intuitively narrows the execution down to only parts that are ``relevant'' to it
(which we define formally later in \Cref{sec:debug-paradigm}).
As shown in \Cref{fig:debug-scene}, relevant rows in upstream tables are automatically colored red.
Specifically, input rows $s_2$ and $s_3$ join with $f_1$ and $f_3$ in a nested-loop fashion to produce rows $j_2$ through $j_5$ in the joined \& filtered table;
they are further grouped into $g_1$ in the group table before finally producing $o_1$.

As soon as the user pins $o_1$, \oursys\ identifies the relevant input combos
and ``positions'' execution at the first input combo in lexicographical order ---
this is how the input combo $\langle s_2, f_1 \rangle$ was activated in the first place in \Cref{fig:debug-scene}.
Starting with this input combo, the user can step through other input combos relevant to $o_1$,
or manually choose any combo to investigate.
Throughout the process, \oursys\ automatically updates highlighting in downstream tables as well as the state of any expression evaluation trees,
in effect supporting forward tracing.

When examining the execution for $\langle s_2, f_1 \rangle$ as shown in \Cref{eg:debug-scene},
the user notices that $\langle s_2, f_1 \rangle$ contributes a value of 16 to the final sum.
According to the filter expression tree,
the price of Amstel is higher than \$2, but \sql{EXISTS($Q_\narrow{inner}$)} returns true,
meaning that Ben should like Amstel per query intention.
\oursys\ allows the user to ``drill down'' into the execution of $Q_\narrow{inner}$ in this context.
Recall that $Q_\narrow{inner}$ is a correlated subquery.
To a programmer, the analogy of evaluating $Q_\narrow{inner}$ is
a function call with a parameter setting for \sql{f.drinker},
which takes its value from the current input combo.
Hence, ``drilling down'' naturally corresponds to ``stepping into'' a function call in GPL debugging.

Once the user drills down to $Q_\narrow{inner}$, \oursys\ creates a new panel for debugging this subquery block,
illustrated in \Cref{fig:debug-scene-1}.
The UI makes it clear that we are executing
\begin{lstlisting}[language=SQL, basicstyle=\small\ttfamily]
SELECT * FROM Likes l WHERE f.drinker = l.drinker;
\end{lstlisting}\vspace*{-1mm}
with parameter \sql{f.drinker} set to Ben.
In this case, a quick glance at the output table or the filter expression tree for this block should reveal the problem ---
nothing supports that Ben likes Amstel.
In fact, this subquery does not even look for Amstel.
Therefore, to fix the query, we should let the outer query ``call'' $Q_\narrow{inner}$ with an additional parameter \sql{s.beer},
and let $Q_\narrow{inner}$ additionally test whether \sql{s.beer} \sql{=} \sql{l.beer}.

The user can modify the original query accordingly and restart the debugging session to verify that it fixes the problem.
Because the new query is syntactically similar to the old,
\oursys\ will produce a very consistent experience for the user:
the execution of the new query will be nearly identical to the old,
with the exact same stages and exact same ordering of input combos and intermediate result rows.
\end{Example}


\vspace{-2mm}
\mypar{A Note on Scalability}
While \Cref{eg:debug-scene} simplifies our discussion by assuming a small database instance,
realistic databases are much larger.
Even with moderately sized databases, joins can easily produce large intermediate results.
As we will see in \Cref{sec:experiments},
for some TPC-H benchmark queries, even with a moderate scaling factor of 1,
a single intermediate result table can easily take an hour to print out entirely on the database server console.
Hence, it is impractical to cache entire results at any location (database server, middleware, or user browser) or send them over the network.

Meanwhile, users generally cannot examine many rows simultaneously.
Therefore, \oursys\ UI supports a pagination mechanism to display each table. 
The user only sees a page's worth of data at once in each table.
Data outside the visible page can be computed and fetched on demand (and subsequently cached or evicted).
As the user interacts with the UI, \oursys\ adjusts the visible portions of all displayed tables accordingly,
so that each reflects the execution point defined by the current input combo.
Similarly, the user can visit any page of input tables to select a new input combo for tracing;
\oursys\ would automatically adjust downstream table displays to show corresponding intermediate result rows.

Hence, efficient pagination is key to scalable SQL debugging.
With efficient pagination, \oursys\ effectively allows the user to ``teleport'' across points of interest without incurring the execution cost in between,
making \oursys\ much more powerful than GPL debugging. Pagination also provides a more manageable and less overwhelming experience for users. 
While most database systems support efficient pagination of input tables,
it is challenging to paginate intermediate results and associated debugging information.
Furthermore, \oursys's requirement of making execution reproducible imposes specific ordering of intermediate results that complicates optimization.
We discuss our solution in \Cref{sec:system}.

%% file: fig/fig-ex1.tex
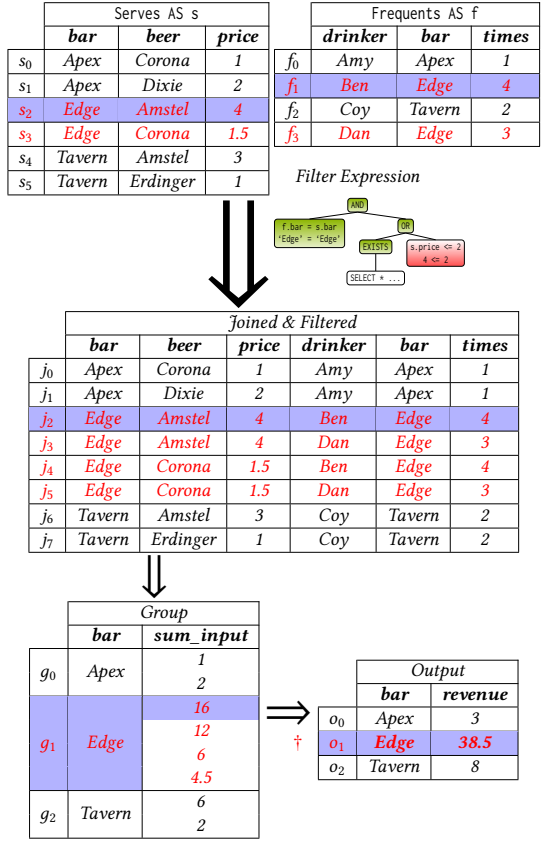
\begin{figure}[t]
    \begin{minipage}[b]{1.0\columnwidth}\centering
    \scalebox{0.8}{
            \begin{tabular}[t]{|c|c|c|c|}
                \cline{2-4}
                \multicolumn{1}{c|}{} & \multicolumn{3}{c|}{\sql{Serves AS s}}  \\ \cline{2-4}
                \multicolumn{1}{c|}{} & \textbf{bar} & \textbf{beer} & \textbf{price} \\ \hline
                $s_0$ & Apex & Corona & 1 \\ \hline
                $s_1$ & Apex & Dixie & 2 \\ \hline
                \rowcolor{blue!30} \textcolor{red}{$s_2$} & \textcolor{red}{Edge} & \textcolor{red}{Amstel} & \textcolor{red}{4} \\ \hline
                \textcolor{red}{$s_3$} & \textcolor{red}{Edge} & \textcolor{red}{Corona} & \textcolor{red}{1.5} \\ \hline
                $s_4$ & Tavern & Amstel & 3 \\ \hline
                $s_5$ & Tavern & Erdinger & 1 \\ \hline
            \end{tabular}
            \begin{tabular}[t]{|c|c|c|c|}
                \cline{2-4}
                \multicolumn{1}{c|}{} & \multicolumn{3}{c|}{\sql{Frequents AS f}}  \\ \cline{2-4}
                \multicolumn{1}{c|}{} & \textbf{drinker} & \textbf{bar} & \textbf{times} \\ \hline
                $f_0$ & Amy & Apex & 1 \\ \hline
                \rowcolor{blue!30} \textcolor{red}{$f_1$} & \textcolor{red}{Ben} & \textcolor{red}{Edge} & \textcolor{red}{4} \\ \hline
                $f_2$ & Coy & Tavern & 2 \\ \hline
                \textcolor{red}{$f_3$} & \textcolor{red}{Dan} & \textcolor{red}{Edge} & \textcolor{red}{3} \\ \hline
            \end{tabular}
    }

    \vspace*{-3.5ex}

    \scalebox{0.5}{
        \hspace*{10em}
        \scalebox{10.0}{$\Downarrow$}
        \begin{tikzpicture}[scale=.55,
            level 1/.style={level distance=1cm, sibling distance=1mm},
            level 2/.style={sibling distance=1mm, level distance=1cm}, 
            level 3/.style={level distance=1.5cm,sibling distance=2mm},
            true/.style = {shape=rectangle, rounded corners, draw, align=center, top color=applegreen, bottom color=applegreen!20},
            false/.style = {shape=rectangle, rounded corners, draw, align=center, top color=white, bottom color=red!70},
            sub/.style = {shape=rectangle, rounded corners, draw, align=center, top color=white, bottom color=white}
        ]
            \huge
            \Tree
            [.\node[true] {\sql{AND}}; 
                [.\node[true] {\sql{f.bar = s.bar} \\ \sql{`Edge' = `Edge'}};]
                [.\node[true] {\sql{OR}}; 
                    [.\node[true] {\sql{EXISTS}};
                        [.\node[sub] {\sql{SELECT * ...}};]
                    ]
                    [.\node[false] {\sql{s.price <= 2} \\ \sql{4 <= 2}};]
                ]
            ]
            \draw (0, 1) node {Filter Expression};
        \end{tikzpicture}
    }

    \scalebox{0.8}{
        \begin{tabular}{|c|c|c|c|c|c|c|}
            \cline{2-7}
            \multicolumn{1}{c|}{} & \multicolumn{6}{c|}{Joined \& Filtered}  \\ \cline{2-7}
            \multicolumn{1}{c|}{} & \textbf{bar} & \textbf{beer} & \textbf{price} & \textbf{drinker} & \textbf{bar} & \textbf{times} \\ \hline
            $j_0$ & Apex & Corona & 1 & Amy & Apex & 1 \\  \hline
            $j_1$ & Apex & Dixie & 2 & Amy & Apex & 1 \\ \hline
            \rowcolor{blue!30} \textcolor{red}{$j_2$} & \textcolor{red}{Edge} & \textcolor{red}{Amstel} & \textcolor{red}{4} & \textcolor{red}{Ben} & \textcolor{red}{Edge} & \textcolor{red}{4} \\ \hline
            \textcolor{red}{$j_3$} & \textcolor{red}{Edge} & \textcolor{red}{Amstel} & \textcolor{red}{4} & \textcolor{red}{Dan} & \textcolor{red}{Edge} & \textcolor{red}{3} \\ \hline
            \textcolor{red}{$j_4$} & \textcolor{red}{Edge} & \textcolor{red}{Corona} & \textcolor{red}{1.5} & \textcolor{red}{Ben} & \textcolor{red}{Edge} & \textcolor{red}{4} \\ \hline
            \textcolor{red}{$j_5$} & \textcolor{red}{Edge} & \textcolor{red}{Corona} & \textcolor{red}{1.5} & \textcolor{red}{Dan} & \textcolor{red}{Edge} & \textcolor{red}{3} \\ \hline
            $j_6$ & Tavern & Amstel & 3 & Coy & Tavern & 2 \\ \hline
            $j_7$ & Tavern & Erdinger & 1 & Coy & Tavern & 2 \\ \hline
        \end{tabular}
    }

    \hspace*{-10em}\scalebox{2}{$\Downarrow$}
    
    \scalebox{0.8}{
        \begin{tabular}{|c|c|c|}
            \cline{2-3}\multicolumn{1}{c|}{} & \multicolumn{2}{c|}{Group}\\
            \cline{2-3}\multicolumn{1}{c|}{} & \textbf{bar} & \textbf{sum\_input} \\ \hline
            \multirow{2}{*}{$g_0$} & \multirow{2}{*}{Apex} & 1 \\
                                                         & & 2 \\\hline
            \cellcolor{blue!30} & \cellcolor{blue!30} & \cellcolor{blue!30}\textcolor{red}{16} \\
            \cellcolor{blue!30} & \cellcolor{blue!30} & \textcolor{red}{12} \\
            \cellcolor{blue!30} & \cellcolor{blue!30} & \textcolor{red}{6} \\
            \cellcolor{blue!30}\multirow{-4}{*}{\textcolor{red}{$g_1$}} & \cellcolor{blue!30}\multirow{-4}{*}{\textcolor{red}{Edge}} & \textcolor{red}{4.5} \\\hline
            \multirow{2}{*}{$g_2$} & \multirow{2}{*}{Tavern} & 6 \\
                                                           & & 2 \\\hline
        \end{tabular}

        \scalebox{2.5}{$\Rightarrow$}

        \begin{tabular}{|c|c|c|}
            \cline{2-3}
            \multicolumn{1}{c|}{} & \multicolumn{2}{c|}{Output}  \\ \cline{2-3}
            \multicolumn{1}{c|}{} & \textbf{bar} & \textbf{revenue} \\ \hline
            $o_0$ & Apex & 3 \\ \hline
            \rowcolor{blue!30}\textcolor{red}{$o_1$\shifttext{-5em}{$\dagger$}} & \textcolor{red}{\textbf{Edge}} & \textcolor{red}{\textbf{38.5}} \\ \hline
            $o_2$ & Tavern & 8 \\ \hline
        \end{tabular}
    }
    \end{minipage}
    \vspace*{1ex}
    \caption{\mdseries Debugging context for outer query block, \Cref{eg:debug-scene}.}
    \label{fig:debug-scene}
\end{figure}

\begin{figure}[t]
    \vspace*{-4ex}
    \begin{tabular}{cc}\small
        Bindings from enclosing queries: & \sql{f.drinker = 'Ben'} \\\hline
    \end{tabular}
    \vspace*{1ex}

    \begin{minipage}[b]{1.0\columnwidth}\centering
        \scalebox{0.8}{
            \begin{tabular}{|c|c|c|}
                \cline{2-3}
                \multicolumn{1}{c|}{} & \multicolumn{2}{c|}{\sql{Likes AS l}}  \\ \cline{2-3}
                \multicolumn{1}{c|}{} & \textbf{drinker} & \textbf{beer} \\ \hline
                \rowcolor{blue!30} $l_0$ & Amy & Erdinger \\ \hline
                $l_1$ & Ben & Budweiser \\ \hline
                $l_2$ & Ben & Dixie \\ \hline
                $l_3$ & Coy & Amstel \\ \hline
                $l_4$ & Dan & Amstel \\ \hline
                $l_5$ & Dan & Corona \\ \hline
            \end{tabular}

            \parbox{0.35\columnwidth}{
            \begin{tikzpicture}[scale=.4,
                level 1/.style={level distance=1cm, sibling distance=1mm},
                level 2/.style={sibling distance=1mm, level distance=1cm}, 
                level 3/.style={level distance=1.5cm,sibling distance=2mm},
                true/.style = {shape=rectangle, rounded corners, draw, align=center, top color=applegreen, bottom color=applegreen!20},
                false/.style = {shape=rectangle, rounded corners, draw, align=center, top color=white, bottom color=red!70}
            ]
                \huge
                \Tree
                [.\node[false] {\sql{f.drinker = l.drinker} \\ \sql{'Ben' = 'Amy'}}; 
                ]
                \draw (0,1) node {\large Filter Expression};
            \end{tikzpicture}

            \vspace*{1ex}
            \centerline{\scalebox{3}{$\Rightarrow$}}
            }

            \begin{tabular}{|c|c|c|}
                \cline{2-3}
                \multicolumn{1}{c|}{} & \multicolumn{2}{c|}{Filtered / Output}  \\ \cline{2-3}
                \multicolumn{1}{c|}{} & \textbf{drinker} & \textbf{beer} \\ \hline
                $o_0$ & Ben & Budweiser \\ \hline
                $o_1$ & Ben & Dixie \\ \hline
            \end{tabular}
        }
        \vspace*{1ex}
        \caption{\mdseries Debugging context for inner query block, \Cref{eg:wrong-query}.}
        \label{fig:debug-scene-1}
    \end{minipage}
\end{figure}
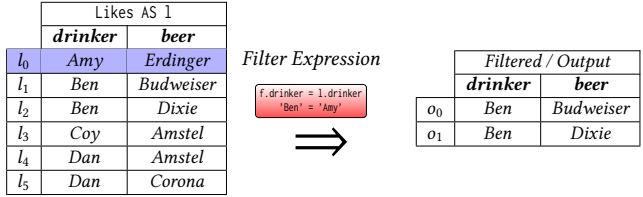

%% file: sections/2-debug-paradigm.tex
\section{Debugging Paradigm}
\label{sec:debug-paradigm}

This section describes the debugging paradigm of \oursys.
We start with the data and execution model (\Cref{sec:debug-paradigm:model}),
discuss various debugging operations (\Cref{sec:debug-paradigm:operations}),
and conclude with a summary of what SQL constructs \oursys\ currently supports and how it can be further extended (\Cref{sec:debug-paradigm:limitations}).

\subsection{Data and Execution Model}
\label{sec:debug-paradigm:model}

\subsubsection{Table Model and IIDs}

\oursys\ allows users to interact with two types of tables: \emph{base} tables and \emph{derived} tables. Base tables are those that exist in the database, while derived tables are computed from the base tables during execution. 
To support reproducible execution (including ordering of intermediate result rows),
we depart from the default unordered multiset semantics of SQL
and instead model each table as an ordered list of rows, each associated with an \emph{internal row id} (\emph{IID}).
A table's IIDs must be drawn from a totally ordered domain and uniquely identify the rows within the table.
The IIDs do not influence the semantics of SQL query operators and should not be considered as extra columns by these operators.
For example, if two rows have identical values for all columns (ignoring IID),
they are still considered duplicates by SQL though their IIDs differ.

For a base table with a primary key declaration, we simply define IID to be its primary key value.
Otherwise, we choose a compact \sql{UNIQUE} key if one is available.
In case that the table has no key, we use the database's internal row id (e.g., PostgreSQL's \sql{ctid});
such ids are unique among duplicate rows.
For a derived table, we define its IID according to how the table is computed.
For each SQL query operator, we define a \emph{canonical execution procedure} and a \emph{result IID synthesis function}
(more details will be provided shortly).
The IID synthesis function can compute the IID for each result row on the fly during canonical execution,
such that the result rows are always produced in the IID order.
In addition to maintaining a reproducible order, \oursys\ uses IIDs to support a variety of debugging features.
Therefore, good IID designs positively impact performance, as we shall see in \Cref{sec:system}.
Although it is technically possible to make the IID of a row simply its sequence number in the result set,
such numbers are not useful by themselves for tracing, pinning, or possible query rewrite optimizations.
As we will see below, most of the IIDs in \oursys\ are ``logical'' instead of physical,
encoding \emph{data provenance}~\cite{green2007provenance} helpful to tracing and pinning.

\subsubsection{Query Blocks and Debugging Contexts}

Given a query, \oursys\ defines a \emph{canonical execution procedure} that is always (conceptually) followed when debugging.
A complex query can be viewed as a collection of syntactic blocks with dependencies among them.
\oursys\ models its canonical execution in terms of function calls:
\begin{itemize}[leftmargin=*]
\item The outermost query block defines a function that computes the query result when executed.
\item A subquery block defines a function that can be called by the function corresponding to its enclosing query block. 
    A correlated subquery is analogous to a helper function with parameters, while a non-correlated subquery is comparable to a helper function without parameters.
    For instance, for $Q_\narrow{inner}$ in \Cref{eg:wrong-query},
    the caller is responsible for passing the values of external column references (e.g., \sql{f.drinker} \sql{=} \sql{`Ben'} in \Cref{fig:debug-scene-1}) as parameters.
\item Each table defined by \WITH\ is a function that computes the table contents when the table is referenced.
\end{itemize}
Overall, the canonical execution starts by executing the function defined by the outermost query block,
which then calls functions corresponding to its subqueries or tables defined by \WITH,
which may further call functions for their subqueries, etc.

As a function may be called multiple times, for each \emph{invocation} of a function (i.e., an execution of a query block), \oursys\ creates a \emph{debugging context} as needed,
analogous to an ``activation frame'' in GPLs.
As an example, $Q$ in \Cref{eg:wrong-query} calls $Q_\narrow{inner}$ multiple times, with potentially different \sql{f.drinker} values as arguments, resulting in different executions.
Each debugging context holds information specific to the particular execution of the query block,
such as the values of external column references (i.e., parameter values).
A more complex example is shown in \Cref{fig:debugging-context}.

\begin{figure}[t]
\includegraphics[width=\linewidth]{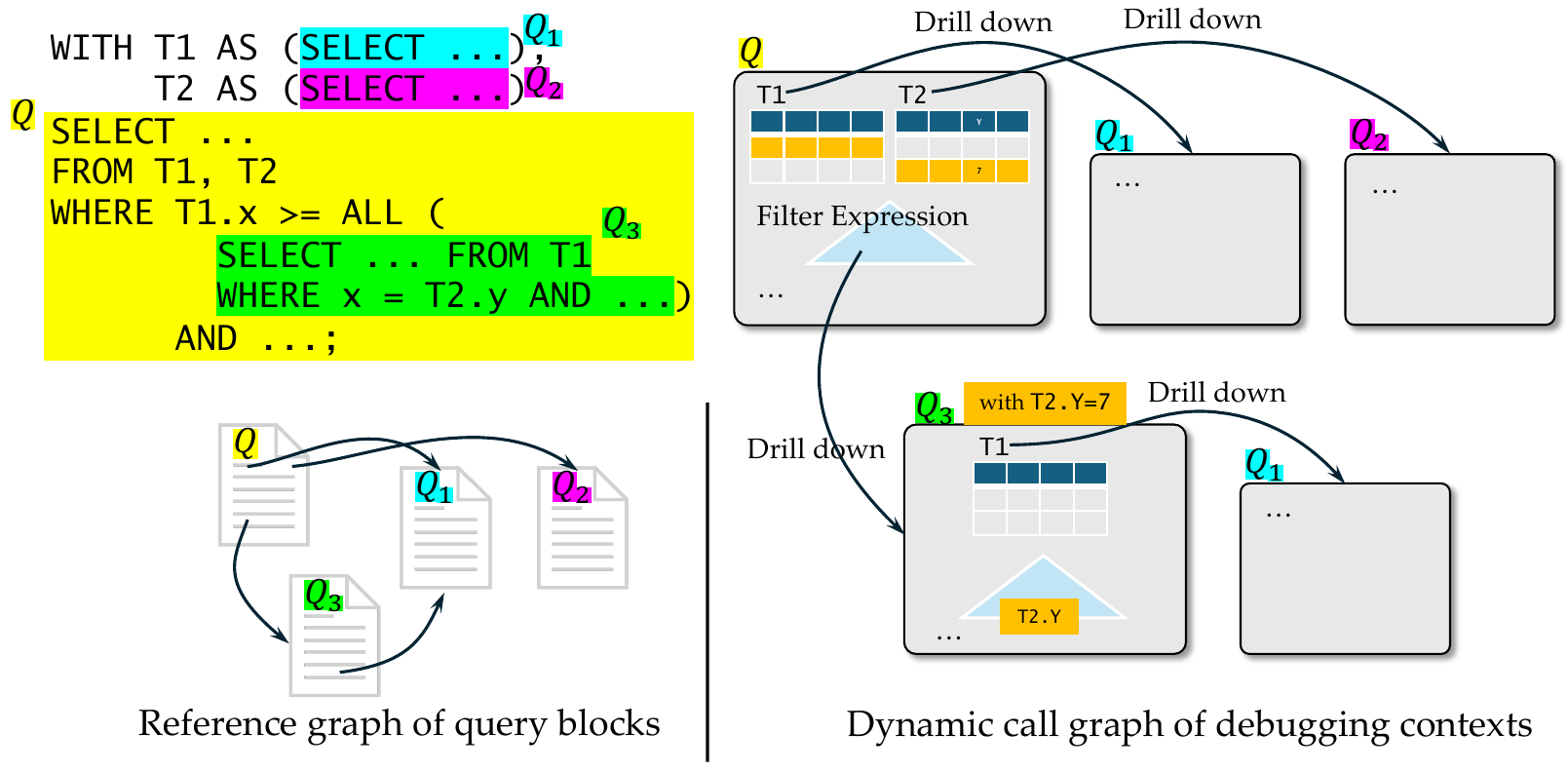}
\caption{\mdseries Static reference graph of query blocks vs.\ dynamic call graph of debugging contexts, for a more complex example.
Note that there are two debugging contexts for $Q_1$ with different states; the parent of $Q_3$'s debugging context provides the binding for $\sql{T2.Y}$.}\label{fig:debugging-context}
\end{figure}

\subsubsection{Canonical Execution for Each Query Block}
\label{sec:debug-paradigm:model:canonical-exec}

Having discussed the overall canonical execution procedure, we now zoom in on the canonical execution of each query block.
For a complex construct such as \SELECT,
we further decompose its execution into \emph{stages},
where each stage can be seen as an operator with its own canonical execution procedure and result synthesis function.
We only discuss \SELECT\ with inner cross joins below; the canonical execution of SQL set/bag operations and general join expressions (including outer joins) are discussed in~\Cref{appendix:debug-paradigm:set} and \Cref{appendix:debug-paradigm:join}.

The first stage in a \SELECT\ block is \emph{join \& filter}.
Its canonical execution is nested for-loops iterating through rows, one for each input table in \sql{FROM}, in order.
In the innermost loop, we test the \sql{WHERE} condition
(which may involve calling subqueries with parameter values obtained from the loop variables).
The result row IID is synthesized as a vector whose components are the IIDs of the joining input rows.
Note that the lexicographic order of these vector IIDs is consistent with the row production order. Considering \Cref{eg:debug-scene} and \Cref{fig:debug-scene},
the IID for \sql{Serves} is its primary key $(\sql{bar}, \sql{beer})$,
and the IID for \sql{Frequents} is its primary key $(\sql{drinker}, \sql{bar})$.
Therefore, the IID for the derived joined \& filtered table has the format $((\sql{s.bar},\sql{s.beer}), (\sql{f.drinker},\sql{f.bar}))$.
In particular, $j_2$'s IID would be $(s_2, f_1)$,
where we abuse notation slightly and use $s_2$ and $f_1$ to denote their IIDs
$(\sql{Edge},\sql{Amstel})$ and $(\sql{Ben},\sql{Edge})$ respectively.

If the block contains grouping or aggregation, a \emph{grouping} stage will be next.
Its canonical execution is a stable sort of input rows according to the \GROUPBY\ expressions%
\footnote{Aggregate queries w/o \GROUPBY\ are the same as having empty \GROUPBY\ list.}
in order.
Each result row starts with the \GROUPBY\ values as columns,
followed by additional columns needed to evaluate the remainder of the query (e.g. \HAVING\ or \SELECT\ expressions).
The result row IID is synthesized as a vector whose components are the \GROUPBY\ expressions in order,
followed by the input row IID.
This order puts member rows of a group together, allowing the UI to detect and display group boundaries (\Cref{fig:debug-scene}).
The \emph{stable} sort ensures that the IID is consistent with the row production order.
For instance, in \Cref{fig:debug-scene}, the first member row of $g_1$ has IID $(\sql{Edge}, j_2)$,
where $\sql{Edge}$ also identifies the $g_1$ group.

The final stage produces the \emph{final output} for the entire \SELECT\ block.
The canonical execution processes the input rows in order. 
If \HAVING\ is present, we filter out input rows whose group does not pass the \HAVING\ condition.
Next, if the previous stage is grouping, we produce one result row for each group, using the leading portion of the IID corresponding to \GROUPBY\ expressions as the result IID.
For instance, $o_1$ in \Cref{fig:debug-scene} would be \sql{Edge}, same as $g_1$'s.
Otherwise, no aggregation is involved, and we simply produce one result row for each input row, using the same IID for the result.
In either case, the IID is consistent with the result row production order.

Several special cases associated with the final stage are worth noting,
including support for \sql{DISTINCT} and \sql{ORDER} \sql{BY}.
If any \HAVING\ or \SELECT\ expression contains subqueries, they would be handled the same way as in the join \& filter stage.
If \SELECT\ is followed by \sql{DISTINCT},
the canonical execution will further perform a sort of all result rows using all columns in some order
(optimized to be maximally consistent with the input IID order)
and output only distinct rows;
the result IID will be synthesized as a vector whose components are all the columns in the order chosen.
If \sql{ORDER} \sql{BY} is present,%
\footnote{Per SQL standard, every \sql{ORDER} \sql{BY} expression must correspond to an output column.
For SQL dialects that do not have this restriction, complex \sql{ORDER} \sql{BY} may necessitate a separate stage to help with debugging. We will not elaborate here.}
the canonical execution will further perform a stable sort of all result rows according to the \sql{ORDER} \sql{BY} expressions,
and the result IID will be synthesized as a vector whose components start with the \sql{ORDER} \sql{BY} expressions and end with the input IID.

\subsection{Debugging Operations}
\label{sec:debug-paradigm:operations}

We describe what operations a user can perform in a debugging context,
focusing on those requiring formalization and in-depth discussion;
others with straightforward semantics (e.g., visualization of expression tree) are omitted.
We describe the operations mostly in the context of \SELECT\ blocks with inner cross joins, although we have generalized them to other SQL constructs.

\subsubsection{Input Combo Space and Execution Positioning}

Recall that a debugging context refers to a specific execution of a query block.
To represent the entire execution of a debugging context, we define its \emph{input combo space}
as an ordered set of combinations of input tuples called ``\emph{input combos}'' (illustrated in \Cref{sec:debug-scene}), each representing a particular point in execution.
The (conceptual) current point of execution is called the \emph{active input combo} for the debugging context.
For a \SELECT\ debugging context with $n$ input tables,
the active input combo is the $n$ input rows being examined inside the innermost loop by the canonical execution of the join \& filter stage,
and it is represented by an $n$-dimensional vector whose components are the IIDs of these input rows.
For instance, the active input combo of the \SELECT\ debugging context shown in \Cref{fig:debug-scene} is $\langle s_2, f_1 \rangle$,
drawn from the input combo space $\{s_0, \dots, s_5\} \times \{f_0, \dots, f_3\}$.

The user can position the execution of a debugging context at a particular point using either \emph{stepping} or \emph{teleporting}.
With \emph{stepping}, \oursys\ automatically advances the active input combo to its successor (or predecessor if stepping in reverse order) in the input combo space.
For example, (ignore the pin and) suppose the active input combo in \Cref{fig:debug-scene} were $\langle s_0, f_3 \rangle$;
stepping would advance it to $\langle s_1, f_0 \rangle$, followed by $\langle s_1, f_1 \rangle$, consistent with the processing order of the canonical execution.
With \emph{teleporting}, given any input table, through the paginated display,
the user can jump or scroll to any page and select a particular row as active;
\oursys\ will then update the active input combo accordingly.
Even more advanced execution positioning can be achieved by pinning, as described later.


\subsubsection{Forward Tracing}

An active input combo in the debugging context can produce \emph{derivative} rows in the downstream result tables produced by the stages.
Intuitively, the input combo contributes to the computation of derivative rows;
formally, the input combo participates in the \emph{how-provenance}~\cite{green2007provenance} of the derivative row.
As discussed in \Cref{sec:debug-scene}, once the active input combo is set,
\oursys\ automatically refreshes all visualizations of expression trees,
such that they reflect evaluation over the active input combo or its derivative rows.
\oursys\ also automatically refreshes all result table displays,
such that they show the pages containing and highlighting the derivative rows.
This \emph{forward tracing} feature allows the user to examine the effect of input rows on subsequent processing
and potentially understand why a desired effect is not achieved.
Note that for a \SELECT\ block, a given input combo can contribute to at most one result row per stage,
so there is no ambiguity in which derivative rows to show downstream.
For example, in \Cref{fig:debug-scene}, the active input combo $\langle s_2, f_1 \rangle$ 
has one derivative row per result table, namely $j_2$, the first row in $g_1$, and $o_1$.
In \Cref{fig:debug-scene-1}, there is no derivative row for the active input combo $\langle l_0 \rangle$.
By design, \oursys\ ensures this rule of at most one derivative per stage for other blocks (such as set/bag operations) as well.

Note that by design, across stages in a debugging context,
the result IIDs naturally encode the provenance information linking rows from one stage to the next.
For example, as discussed in \Cref{sec:debug-paradigm:model:canonical-exec}, 
the IID of the joined \& filtered table, synthesized from the IIDs of the input tables,
essentially encodes how-provenance.

\subsubsection{Pinning: Tracing and Watchpointing}

We first describe the semantics of \emph{pinning} and then discuss its use for tracing and watchpointing.
Consider all tables displayed for a debugging context.
\oursys\ allows the user to \emph{pin} up to one row from each of these tables.
Formally, each pinned row defines a subset of the input combo space, called the row's \emph{pinned (input combo) space}.
Overall, the pinned space for the debugging context is its input combo space \emph{intersected} with each of the pinned spaces defined by the pinned rows.
Intuitively, the pinned space allows the user to narrow the execution down to the points of interest while debugging. 

Consider a \SELECT\ block joining $n$ tables in \FROM\ with input combo space $\mathbf{R} = \prod_{i=1}^n R_i$,
where each $R_i$ denotes the ordered list of IIDs for the $i$-th input table.
A pinned row in the $j$-th input table with IID $x$ defines a pinned space of $\smash{\prod_{i=1}^{j-1} R_i \times \{x\} \times \prod_{j=i+1}^n R_i}$;
i.e., the user is interested only in input combos with $x$ participating.
A pinned row in a result (intermediate or final) table defines a pinned space of
$\{ v \mid v \in \mathbf{R} \land \text{$x$ is a derivative row of $v$} \}$;
i.e., the user is interested only in input combos that contribute to $x$.
Considering \Cref{fig:debug-scene},
the pinned space for the pinned $o_1$ is $\{ \langle s_2, f_1 \rangle, \langle s_2, f_3 \rangle, \langle s_3, f_1 \rangle, \langle s_3, f_3 \rangle\}$ because these input combos contribute to the total revenue in $o_1$.
If the user additionally pins $f_1$ in \sql{Frequents};
this pinned row will have a pinned space of $\{ s_0, \dots, s_5 \} \times \{ f_1 \}$.
Thus, the overall pinned space will be shifted to $\{ \langle s_2, f_1 \rangle, \langle s_3, f_1 \rangle \}$,
meaning the user only wants to investigate how Dan contributes to Edge's total revenue.


Pinning has multiple uses.
First, it augments \oursys's tracing capability:
pinning effectively allows \emph{backward tracing} from a pinned result row produced by any stage to the pinned space of input combos,
and then from there, forward tracing to result rows further downstream.
Second, combined with stepping, pinning provides a form of \emph{watchpointing}.
With a pinned space in effect, \oursys\ restricts stepping to the pinned space,
effectively setting a watchpoint that pauses execution only at points relevant to the pinned rows. In \Cref{fig:debug-scene}, with $o_1$ pinned, execution automatically stops at the 4 relevant input combos $\langle s_2, f_1 \rangle$, $\langle s_2, f_3 \rangle$, $\langle s_3, f_1 \rangle$, $\langle s_3, f_3 \rangle$, skipping irrelevant portions of execution before, after, and in between.

Note that the result IIDs across stages naturally encode provenance for efficient pinning.
For example, although the IID of an aggregate result row does not encode its provenance by itself,
the group-by values therein allow a group's how-provenance to be recovered by examining the IIDs for the preceding grouping stage (which contains these values too) and/or by querying.

\subsubsection{Drilling Down and Pulling Up}
As illustrated in \Cref{eg:debug-scene}, \oursys\ allows the user to \emph{drill down} into a subquery,
analogous to ``stepping into'' a function call.
Besides drilling down into subqueries in \WHERE, \HAVING, and \SELECT\ expressions,
\oursys\ also allows drilling down into subqueries in \FROM,
which can be either directly nested therein or via a reference to some \WITH\ definition.
In these cases, the user can drill down directly through a particular row in a derived input table;
\oursys\ will open the debugging context for the subquery responsible for producing that table
and automatically pin that row in the subquery's final output table.
When debugging a complex query with many nested blocks,
\oursys\ essentially maintains a ``call stack'' of debugging contexts (similar to \Cref{fig:debugging-context}).
To let the user \emph{pull up} from a subquery debugging context,
\oursys\ simply returns the user to the previous debugging context on the stack,
which belongs to the enclosing query block.
The state of the subquery debugging context is still preserved
until the user changes the active input combo in the enclosing block's debugging context,
which forces ``stepping out'' from the last subquery function call.

\subsection{Supported Features and Limitations}\label{sec:debug-paradigm:limitations}

\oursys\ supports a rich set of SQL query features, such as \SELECT\ queries,
set/bag operations (e.g., \sql{UNION}, \sql{INTERSECT}, \sql{EXCEPT}),
subqueries (including correlated ones),
outer joins and joins expressed in general \sql{JOIN} syntax (except \sql{LATERAL}),
and \sql{WITH}.
\oursys\ currently does not support recursive \WITH, \sql{LATERAL} joins and \sql{WINDOW} functions.

However, \oursys's debugging paradigm is very general.
By modeling a query as a collection of syntactic blocks that act like function calls,
any relational operator that deterministically transforms input tuples into output tuples can be adapted to this framework.
Assigning totally ordered, logical identifiers (IIDs) allows \oursys\ to enforce reproducible execution stages across complex operators.
We briefly outline strategies for several possible future extensions:
1)~for \sql{WINDOW} functions, \oursys\ could introduce a new execution stage where canonical execution performs a stable sort based on \sql{PARTITION} \sql{BY} and \sql{ORDER} \sql{BY} expressions;
2)~\sql{LATERAL} joins can be modeled in a way similar to correlated subqueries, where the left table acts as the caller passing bound parameters to the right table's context;
3)~recursive \WITH\ can potentially be supported by treating each recursive step as a context whose input tables are produced by the previous step, until a fixed point is reached.

\mypar{Limitations}
While \oursys\ handles complex queries, it does have some limitations.
First, as \oursys\ relies on deterministic execution,
it does not support built-in functions whose results cannot be reliably reproduced (e.g., \sql{RAND}).
The behavior of queries involving such functions may depend on the execution plan chosen by the optimizer,
and cannot be reproducible by \oursys's canonical execution.

Second, while features such as forward tracing and pinning allow a user to quickly examine execution around points of interest,
\oursys\ does not by itself suggest potential points of interest.
Currently, \oursys\ sorts and paginates all intermediate and final result tables in IID order,
which makes it easier for the user to locate a row by its IID (or prefix thereof).
A future extension is to support searching for specific result rows by an arbitrary filter condition,
which can be efficiently implemented by issuing a rewritten query to compute only the filtered rows
(in a similar vein as a \emph{page-fetch query} in \Cref{sec:system:pagination}).
A more intriguing extension would be to support pinning by condition instead of by row.

%% file: sections/3-system.tex
\section{System and Optimizations}
\label{sec:system}

We now describe the implementation and optimizations of \oursys.
Given a query $Q$ being debugged, a straightforward approach would be to carry out the canonical execution of $Q$ and collect all debugging information in one go,
but this approach is not scalable.

\begin{Example}\label{ex:tpch-motivating}\itshape
Consider a TPC-H~\cite{tpch} database instance generated with a scale factor of $1$ (i.e., total size of all tables is $1$GB), and the following subquery in the \sql{FROM} clause of benchmark query Q8:
\begin{lstlisting}[language=SQL, basicstyle=\scriptsize\ttfamily]
SELECT EXTRACT(year from o_orderdate) as o_year,
       l_extendedprice * (1 - l_discount) as volume,
       n2.n_name as nation
FROM part, supplier, lineitem, orders, customer, nation n1, nation n2, region
WHERE p_partkey = l_partkey AND s_suppkey = l_suppkey
AND ...  -- omitted for simplicity
\end{lstlisting}

If we were to compute and display the ten entire tables (8 input, 1 joined \& filtered, 1 output) and lineage data, we would ship at least $2$GB of data from the database server to the \oursys\ client, introducing unacceptable overhead over the network and client-side memory. 

\end{Example}

To address this challenge, \oursys\ uses three high-level ideas.
(1)~Instead of showing the entire canonical execution of $Q$,
we let the user examine one small, relevant window of this execution at a time.
(2)~To obtain all debugging information needed for a particular execution window,
instead of performing the canonical execution and instrumenting it,
we can formulate SQL queries based on $Q$ to compute such information directly and declaratively.
(3)~We can judiciously compute some summary data
and then use it to further rewrite these queries to be more efficient. 

As described in \Cref{sec:debug-scene}, \oursys\ realizes idea (1) using a paginated display for each table.
For each debugging context, the active input combo 
marks the current point of execution and controls which pages to display by default: pages containing the input combo for input tables, and the page containing the derivative row for each subsequent stage.
Together, these pages define the ``window'' of execution seen by the user.
The optimization of tracing and pinning builds on pagination; see \Cref{appendix:system:positioning} for details.

\subsection{Optimizing Pagination}
\label{sec:system:pagination}

Given a base or derived table to display, the potential savings of pagination are easy to see:
by focusing on one page at a time, we only need to compute, transmit, and render content on this page alone.
The baseline solution to pagination supported by SQL uses its \OFFSET\ and \LIMIT\ features.
However, \OFFSET\ and \LIMIT\ alone do not make queries faster
(verified in \Cref{sec:experiments}).
The optimizer typically has insufficient knowledge to skip directly to the \OFFSET-th result row,
so the query often executes from the very beginning to $\OFFSET+\LIMIT$, creating enormous waste.
Furthermore, for queries that enforce result ordering
(which is the norm in \oursys\ as it aims to provide consistent and reproducible orderings for all results),
simply determining the order would often involve computing and sorting all result rows, saving no computation cost.

To overcome this limitation,
we observe that if we know which input rows contribute to the particular result page,
we can use such information to prefilter the input rows and reduce execution cost.
For example, in the joined \& filtered table of \Cref{fig:debug-scene},
suppose a user only needs to retrieve the page containing rows $j_2$ through $j_5$.
It turns out that we only need $\{ s_1, s_2 \}$ from \sql{Serves} and $\{ f_1, f_3 \}$ from \sql{Frequents} to compute this page.
In general, explicitly enumerating a set of such input rows is not scalable,
but we can instead compute a compact summary of this set, at the expense of potentially introducing some false positives
(but never any false negatives).

\oursys\ has three types of summary-based filters, described further below, to cover the range of possibilities:
1)~IID-based filters,
2)~``sargable''~\cite{DBLP:conf/sigmod/SelingerACLP79} filters, and
3)~Bloom filters~\cite{DBLP:journals/cacm/Bloom70}.
All three require precomputing a summary of what input rows contribute to each page.
This step, performed when the debugging context is initialized,
executes a \emph{milestone query} (generated automatically by \oursys)
to produce a \emph{milestone table} for each input table of the debugging context and for the result table of each stage.
For each page of table contents, a milestone row contains the summaries of input rows that contribute to the page.
The milestone tables are cached by the \oursys\ client.
\footnote{\oursys\ assumes the client sees a static snapshot of the underlying database; otherwise, data updates will invalidate the cached milestones.}
Subsequently, when a particular page of a table is requested,
\oursys\ consults the corresponding milestone table to generate a \emph{page-fetch query} that incorporates filtering using the summaries pertaining to the page requested.

\begin{table}[t]\scriptsize
    \renewcommand{\tabcolsep}{1mm}
    \begin{tabular}[b]{@{}r|p{16mm}|c|c|c|c|} \cline{2-6}
        & \multicolumn{1}{c|}{\sql{min\_iid}} & \sql{s.bar} & \sql{s.beer} & \sql{f.drinker} & \sql{f.bar}    \\ \cline{2-6}
        0 & ((Apex, Corona), (Amy, Apex)) & [Apex, Edge] & [Amstel, Dixie] & [Amy, Ben] & [Apex, Edge]  \\ \cline{2-6}
        1 & ((Edge, Amstel), (Dan, Edge)) & [Edge, Edge] & [Amstel, Corona] & [Ben, Dan] & [Edge, Edge]  \\ \cline{2-6}
        2 & ((Tavern, Amstel), (Coy, Tavern)) & [Tavern, Tavern] & [Amstel, Erdinger] & [Coy, Coy] & [Tavern, Tavern]  \\ \cline{2-6}
    \end{tabular}
    \smallskip
    \caption{\label{tab:milestone}\mdseries Milestone table for the joined \& filtered table in \Cref{fig:debug-scene},
        with $3$ pages and page size of $3$ (rows).
        The IIDs have the format $((\sql{s.bar},\sql{s.beer}), (\sql{f.drinker},\sql{f.bar}))$, and 3 min IIDs are those of the rows $j_0$, $j_3$, and $j_6$ respectively.
        The Bloom filter column is omitted.}
\end{table}

\subsubsection{IID-Based Filtering}

Recall that \oursys\ sorts every table by its IID to ensure consistency and reproducibility.
Hence, pages of a table partition its rows into consecutive, non-overlapping IID ranges.
To enable IID-based filtering, we precompute, in the milestone table, the minimum IID among rows on each page.
For example, \Cref{tab:milestone} shows the milestone table for the joined \& filtered table in \Cref{fig:debug-scene},
with the minimum IID for each page captured by the \sql{min\_iid} column.
Using this information, \oursys\ can add a tight range condition on IID to the \WHERE\ clause of a page-fetch query.

Continuing with the same example, the following query fetches the contents of the second page of the joined \& filtered table,
while also synthesizing the IID for each result row. Note that the lower IID bound is the minimum IID of the second page in \Cref{tab:milestone},
and the (open) upper IID bound is the minimum IID of the next page:
\begin{lstlisting}[language=SQL, basicstyle=\scriptsize\ttfamily]
SELECT *, ((s.bar,s.beer),(f.drinker,f.bar)) AS _iid -- synthesize IID
FROM Serves s, Frequents f
WHERE (...) -- original WHERE conditions
-- IID-based filtering:
AND ((s.bar,s.beer),(f.drinker,f.bar))>=(('Edge','Amstel'),('Dan','Edge'))
AND ((s.bar,s.beer),(f.drinker,f.bar))<(('Tavern','Amstel'),('Coy','Tavern'))
ORDER BY 1; -- order by IID
\end{lstlisting}

IID-based filtering serves two purposes.
First, it rejects any result row not on the requested page.
This feature is indispensable because the other types of filtering implemented may introduce false positives and admit rows outside the requested page;
therefore, \oursys\ always activates IID-based filtering to ensure correctness.
Second, IID-based filtering can enable efficient execution.
However, its potential is limited:
from a range bound on a multi-component IID, we can safely infer a range bound only on the leading component, but not on subsequent components.
For example, in the page-fetch query above,
the query optimizer may infer that \sql{s.bar} (and hence \sql{f.bar} by transitivity) falls within $[\sql{Edge}, \sql{Tavern}]$ and use an index on \sql{s.bar} (or \sql{f.bar}),
but nothing is known about \sql{s.beer} or \sql{f.drinker}.
This limitation motivates other types of filtering below.

\subsubsection{Sargable Filtering}

To enable efficient page-fetch queries,
we aggressively look for opportunities to inject safe, \emph{sargable}~\cite{DBLP:conf/sigmod/SelingerACLP79} predicates that enable index plans.
To this end, for each column $A$ in an input table, where an index already exists on $A$, \oursys\ computes a concise summary of the $A$ values, in the form of a single range $[\min, \max]$, over all input rows that contribute to each page of the result table.
For example, \Cref{tab:milestone} shows the milestone table for the joined \& filtered table in \Cref{fig:debug-scene},
assuming that \sql{s.bar}, \sql{s.beer}, \sql{f.drinker}, \sql{f.bar} are the indexed columns. \oursys\ injects a range condition for each in the \WHERE\ clause of a page-fetch query.

Using the same example, the page-fetch query for the second page of the joined \& filtered table can now be augmented as follows:
\begin{lstlisting}[language=SQL, basicstyle=\scriptsize\ttfamily]
SELECT ... FROM Serves s, Frequents f
WHERE (...) -- original WHERE conditions
AND s.bar BETWEEN 'Edge' AND 'Edge' -- sargable filtering
AND s.beer BETWEEN 'Amstel' AND 'Corona'
AND f.drinker BETWEEN 'Ben' AND 'Dan'
AND f.bar BETWEEN 'Edge' AND 'Edge'
AND ... -- IID-based filtering
ORDER BY 1; -- order by IID
\end{lstlisting}
Since the sargable filters are always on indexed columns,
they enable the optimizer to consider index plans that access only the relevant parts of the input tables.
Even when index plans are not the most optimal,
the inexpensive filter conditions still reduce the number of rows involved in downstream processing (e.g., join) and hence overall execution cost. 

Pushing down filters too aggressively may have a negative effect when the filters are not selective,
which can happen if the page size is not small and the result row ordering does not correlate with the filter column value.
The optimizer may underestimate the output cardinality of the filter and choose a secondary index scan that is less efficient than a table scan.
Therefore, we only inject a sargable filter for column $A$ if its range on the requested page covers a small percentage of the entire domain of $A$.
\oursys\ uses $30\%$ as a cutoff, which works well empirically.

Finally, instead of using one range to summarize input column values for a page, we can use multiple ranges to reduce false positives,
at the expense of higher precomputation and storage costs for milestones.
This trade-off is worth investigating as future work.

\subsubsection{Bloom Filtering}

Correlated subqueries in \WHERE\ can still bottleneck query execution,
especially if the external columns they reference are not indexed and therefore not covered by sargable filtering.
One approach to avoid evaluating expensive correlated subqueries is memoization:
regarding each correlated subquery $q$ as a function,
we can cache all value settings of the external columns that $q$ is invoked with,
along with the corresponding results returned by $q$.
However, this approach is not scalable as the number of possible value settings can be high
(even with sophisticated decorrelation techniques~\cite{DBLP:conf/icde/SeshadriPL96} to restrict such settings)
and results can be large, requiring large cache spaces.

To balance space efficiency and performance, we take an approximate approach at a coarser grain:
instead of capturing the behavior of each subquery precisely,
we consider the entire \WHERE\ condition (including any constituent subqueries),
and use a Bloom filter~\cite{DBLP:journals/cacm/Bloom70} to track, for each page, the set of ``relevant'' input column values for which the entire \WHERE\ evaluates to true.
We choose relevant columns to be all input table column externally referenced by any subquery,
plus any input table column involved in an atomic condition together with some subquery.
For instance, in our running example ($Q$ from \Cref{eg:wrong-query}), the only relevant input column is \sql{f.drinker}.
As another example, if the \WHERE\ condition is $A$ \sql{IN} \sql{(}$q(B)$\sql{)}, where $q$ is a correlated subquery with external column reference $B$,
the set of relevant input columns would be $\{A,B\}$.
A single SQL query can conveniently precompute such Bloom filters using a user-defined aggregate, along with the rest of the milestone table.

Continuing with our running example, the page-fetch query for the second page of the joined \& filtered table is now updated to:
\begin{lstlisting}[language=SQL, basicstyle=\scriptsize\ttfamily]
SELECT ... FROM Serves s, Frequents f
WHERE (...) -- original WHERE conditions
AND ... -- sargable filtering
AND ... -- IID-based filtering
AND BLOOM_CHECK('10101010', f.drinker) -- Bloom filtering
ORDER BY 1; -- order by IID
\end{lstlisting}
Here, \sql{BLOOM\_CHECK($F$,$e$)} is a user-defined function that checks if an entry $e$ is in the Bloom filter $F$.
Given the probabilistic design of Bloom filters, \sql{BLOOM\_CHECK} may return false positives but not false negatives.
\sql{`10101010'} is the Bloom filter computed for the second page,
encoding all \sql{f.drinker} values over input combos that contribute to this page.
Failing \sql{BLOOM\_CHECK} allows query execution to save time by bypassing the evaluation of the expensive subquery.
In this example, row $f_2$, with \sql{f.drinker} of \sql{Coy}, fails the subquery condition,
so the computed Bloom filter will not have the contribution of \sql{Coy}.
Hence, when evaluating the above page-fetch query,
even though \sql{Coy} still passes the sargable filter, it will fail \sql{BLOOM\_CHECK},
skipping rest of \WHERE\ including the subquery.
Yet, because of false positives, passing \sql{BLOOM\_CHECK} does not mean \WHERE\ must evaluate to true;
the original \WHERE\ condition still must be included.

\begin{proposition}\label{thm:bloom}
Suppose functions $\sql{BLOOM\_GEN}: \mathcal{P}(\text{tuples}) \to \text{bitstrings}$
and $\sql{BLOOM\_CHECK}: \text{bitstrings} \times \text{tuples} \to \{\text{true}, \text{false}\}$
satisfy $e \in \mathcal{V} \Rightarrow \sql{BLOOM\_CHECK}(\sql{BLOOM\_GEN}(\mathcal{V}), e)$
for any tuple $e$ and any set $\mathcal{V}$ of tuples of the same sort.
Let $\{R_i\}$ denote a list of (potentially aliased) input tables and $\Theta$ a condition over $\{R_i\}$.
For any subset $\mathcal{A}$ of columns in $\{R_i\}$, the following two queries are equivalent:
\begin{lstlisting}[language=SQL, mathescape=true, basicstyle=\scriptsize\ttfamily]
`$Q_1$`: SELECT * FROM $\dots, R_i, \dots$ WHERE $\Theta$;
`$Q_2$`: SELECT * FROM $\dots, R_i, \dots$ WHERE $\Theta$ AND BLOOM_CHECK($F$, $\langle \mathcal{A} \rangle$);
    $\text{where}\;F = $ BLOOM_GEN(SELECT $\langle \mathcal{A} \rangle$ FROM $\dots, R_i, \dots$ WHERE $\Theta$)$\text{.}$
\end{lstlisting}
\end{proposition}

The proposition above implies considerable freedom in choosing the set of relevant input columns to track for a Bloom filter.
Instead of tracking the IIDs of the input rows (which can be many), we decided to track input column values that influence the outcome of evaluating conditions involving subqueries,
because there are usually far fewer distinct values to track.
This heuristic has worked well for \oursys.
Finally, Bloom filter conditions are not sargable and cannot be used to avoid full table scans.
Evaluating them introduces overhead.
Therefore, \oursys\ uses Bloom filtering only for ``short-circuiting'' evaluation of expensive \WHERE\ clauses with subqueries.
If a Bloom filter returns too many false positives,
the overhead of \sql{BLOOM\_CHECK} may outweigh the savings achieved by skipping subquery evaluation. \oursys\ automatically estimates the false positive rate of each Bloom filter and inject the \sql{BLOOM\_CHECK} condition only if this rate is less than $50\%$.

\subsubsection{Summary}\label{sec:system:pagination:summary}

\oursys\ optimizes pagination by generating page-fetch queries that automatically combine three types of summary-based filters to prune unnecessary data. First, IID-based filtering uses the minimum IID of a page to create a precise page boundary; it is \emph{always} applied to guarantee correctness by rejecting rows outside the requested page. Second, sargable filtering injects dynamic range bounds for indexed columns to encourage efficient index scans. Sargable filters are available only when the query block has base tables with indexed columns, and they are automatically applied on top of IID filters. Finally, Bloom filtering tracks relevant input column values, skipping the evaluation of expensive correlated subqueries. \oursys\ proactively detects the existence of correlated subqueries and automatically computes and applies Bloom filters along with IID and sargable filters whenever possible.

\subsection{System Overview}
\label{sec:system:overview}


We briefly describe the overall \oursys\ system while leaving most implementation and optimization details to the \Cref{appendix:system}.

\oursys\ operates a highly scalable, stateless client-middleware-server architecture.
The frontend client is built with React, which handles user interaction logic, visualizes the debugging contexts, and manages the caching of data and SQL code to ensure a smooth, interactive experience.
The middleware server performs the heavy lifting of query analysis and rewriting using Apache Calcite~\cite{begoli2018apache};
it ``compiles'' the original query to be debugged into a collection of rewritten SQL query templates to support debugging.
Each debugging action initialized by the client is supported by executing appropriately instantiated query templates on the underlying database.
\oursys\ extends the database server using user-defined functions (e.g., for Bloom filters) but does not modify its internals otherwise.

Importantly, to support many concurrent users without creating bottlenecks, \oursys\ maintain no session-specific state for any active debugging session in the middleware or the underlying database.
When a client starts a debugging session, the middleware performs a compilation step to analyzes the original query.
The query is divided into logical blocks (\Cref{fig:debugging-context}),
and for each block (which will produce a debugging context when invoked at runtime),
the middleware automatically generates milestone queries, page-fetch queries, and queries supporting other debugging features,
as parameterized SQL templates.
In the step, the middleware also employs several equivalent query rewrites
(including scalar subquery optimization, sargable filter injection, and recursive filter pushdown, further described in ~\Cref{appendix:system:details})
to help enable optimizations that tend to be missed by the underlying database optimizer.
Finally, the result of this compilation step is shipped to and cached in the client.
As a user enters a specific debugging context,
the client initializes the context by executing the cached milestone queries for the corresponding block and caching the result milestone tables.
Subsequently, as the user navigates within the debugging context, the client dynamically instantiates page-fetch and other supporting queries on demand using the cached milestone tables, enabling fast exploration without computing the full query.





%% file: sections/4-experiments.tex
\vspace{-3mm}
\section{Performance Experiments}
\label{sec:experiments}

We conduct experiments to evaluate the performance and scalability of \oursys\ under the optimizations discussed in \Cref{sec:system}.
We focus on evaluating page-fetch queries (\Cref{sec:system:pagination}) and milestone computation,
as they are the most expensive queries that \oursys\ uses;
other operations (e.g., tracing and pinning) either use page-fetch queries or relatively cheap queries with highly selective conditions.
The baseline for fetching pages is to use SQL \OFFSET\ and \LIMIT.
We enable various pagination optimizations in \oursys\ to evaluate their respective benefits.
IID-based filtering is always enabled (for correctness).
\Cref{sec:experiments:bloom} evaluates pagination with Bloom filtering but not other optimizations, while \Cref{sec:experiments:sargable} evaluates pagination with sargable filtering but not others.
Finally, \Cref{sec:experiments:all} enables all pagination optimizations and evaluates both the performance of pagination and milestone computation.


We use the TPC-H benchmark~\cite{tpch} and generate 3 database instances of sizes $1$GB (benchmark default), $5$GB, and $10$GB.
In addition to the default indexes on the primary keys,
we also create a reasonable set of secondary indexes (details in \Cref{appendix:experiments}) that simulate typical usage.
For \Cref{sec:experiments:bloom,sec:experiments:sargable}, we show one query for each,
where the respective optimization is applicable and can be best evaluated;
for the general evaluation in \Cref{sec:experiments:all}, we show results for all 22 benchmark queries.
All experiments were done on a 64-bit Ubuntu 22.04 LTS server with four Intel(R) Xeon(R) 6530P CPUs @ 2.30GHz, 64GB RAM, and 256GB disk space.
We used PostgreSQL 18.3~\cite{postgres} with \sql{work\_mem=128MB} and \sql{shared\_buffers=8GB};
we also turned off parallel execution (\sql{max\_parallel\_workers\_per\_gather=0}) to reduce its potentially confounding effect, so results are easier to interpret.

\vspace{-3mm}
\subsection{Effectiveness of Bloom Filtering}
\label{sec:experiments:bloom}

This experimental setup uses a variant of TPC-H Q2:
\begin{lstlisting}[language=SQL, basicstyle=\footnotesize\ttfamily]
SELECT ROW(p_partkey, s_suppkey, ps_partkey, ps_suppkey, n_nationkey, r_regionkey), *
FROM part, supplier, partsupp, nation, region
WHERE ... AND ps_supplycost = (
    SELECT MIN(ps_supplycost)
    FROM partsupp, supplier, nation, region
    WHERE p_partkey = ps_partkey  -- p_partkey from outer query
    AND ...);
\end{lstlisting}
As discussed in \Cref{sec:system:pagination},
\oursys\ precomputes a Bloom filter for the \sql{p\_partkey} values for input rows that contribute to each result page,
and when fetching a page, uses the Bloom filter to short-circuit the evaluation of the rest of \WHERE\ containing an expensive correlated subquery. However, the bloom filter is only used for membership checking if the false positive rate is less than $50\%$.

By default, we set the number of Bloom filter bits to $m=1024$;
we conservatively estimate the number of unique \sql{p\_partkey} in each page (denoted by $n$) as the page size;
accordingly, we set the number of Bloom filter hash functions to $\frac{m}{n} \times \ln{2}$.

We conduct two sets of experiments.
First, we fix the database size to $1$GB and vary the page size.
Second, we fix the page size to 50 and vary the database size.
We compare page-fetch queries with Bloom filtering with the baseline using \OFFSET\ and \LIMIT.
We collect the execution times reported by the \sql{EXPLAIN} \sql{ANALYZE} command for all queries and show them in \Cref{fig:bloom-exp}.

\begin{figure}[t]
\begin{subfigure}[t]{0.47\columnwidth}
    \includegraphics[width=1.05\linewidth]{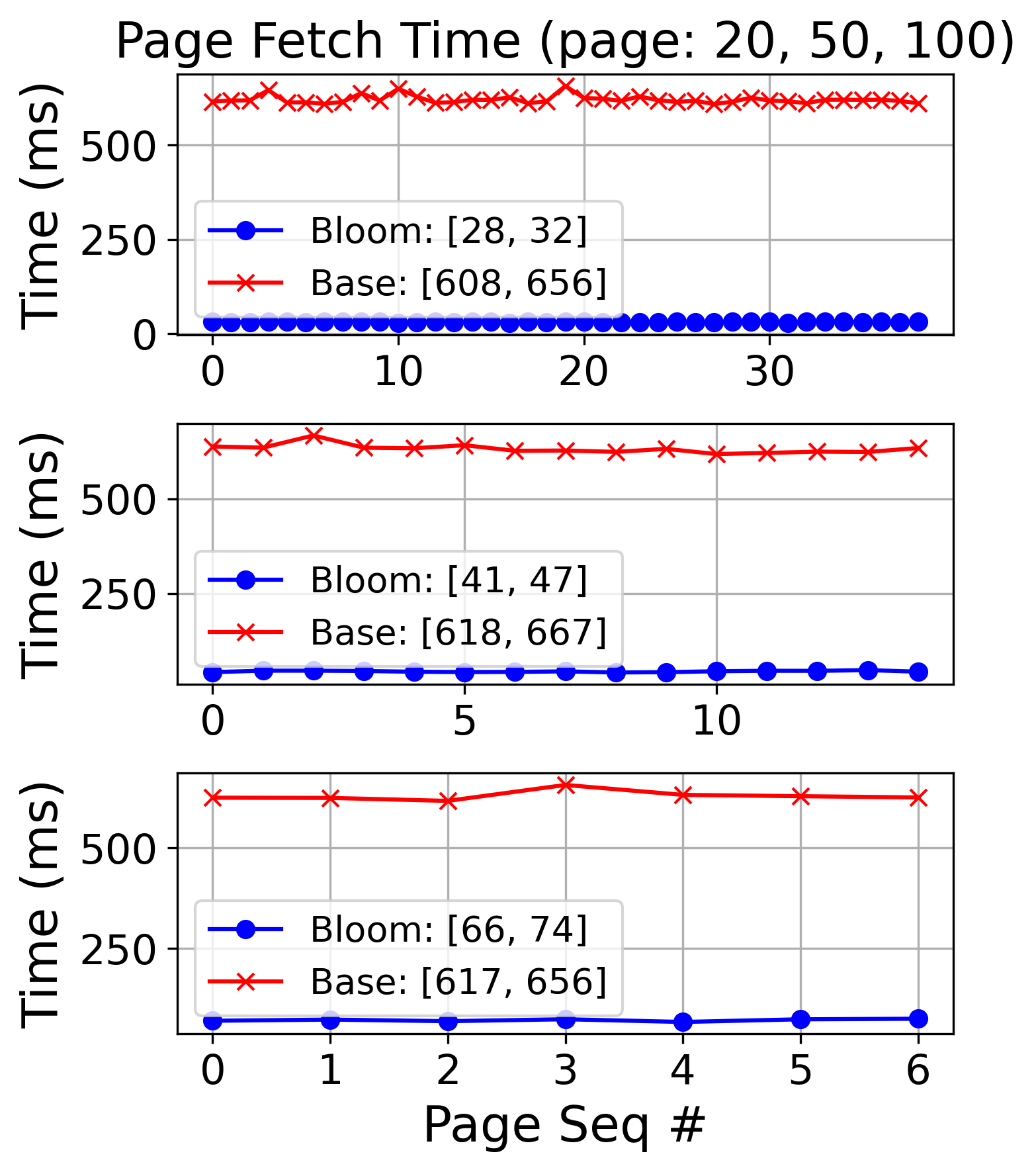}
    \caption{\mdseries\itshape Varying page size; $1$GB database.}
    \label{fig:bloom-vary-pg}
\end{subfigure}
\hfill
\begin{subfigure}[t]{0.52\columnwidth}
    \includegraphics[width=0.95\linewidth]{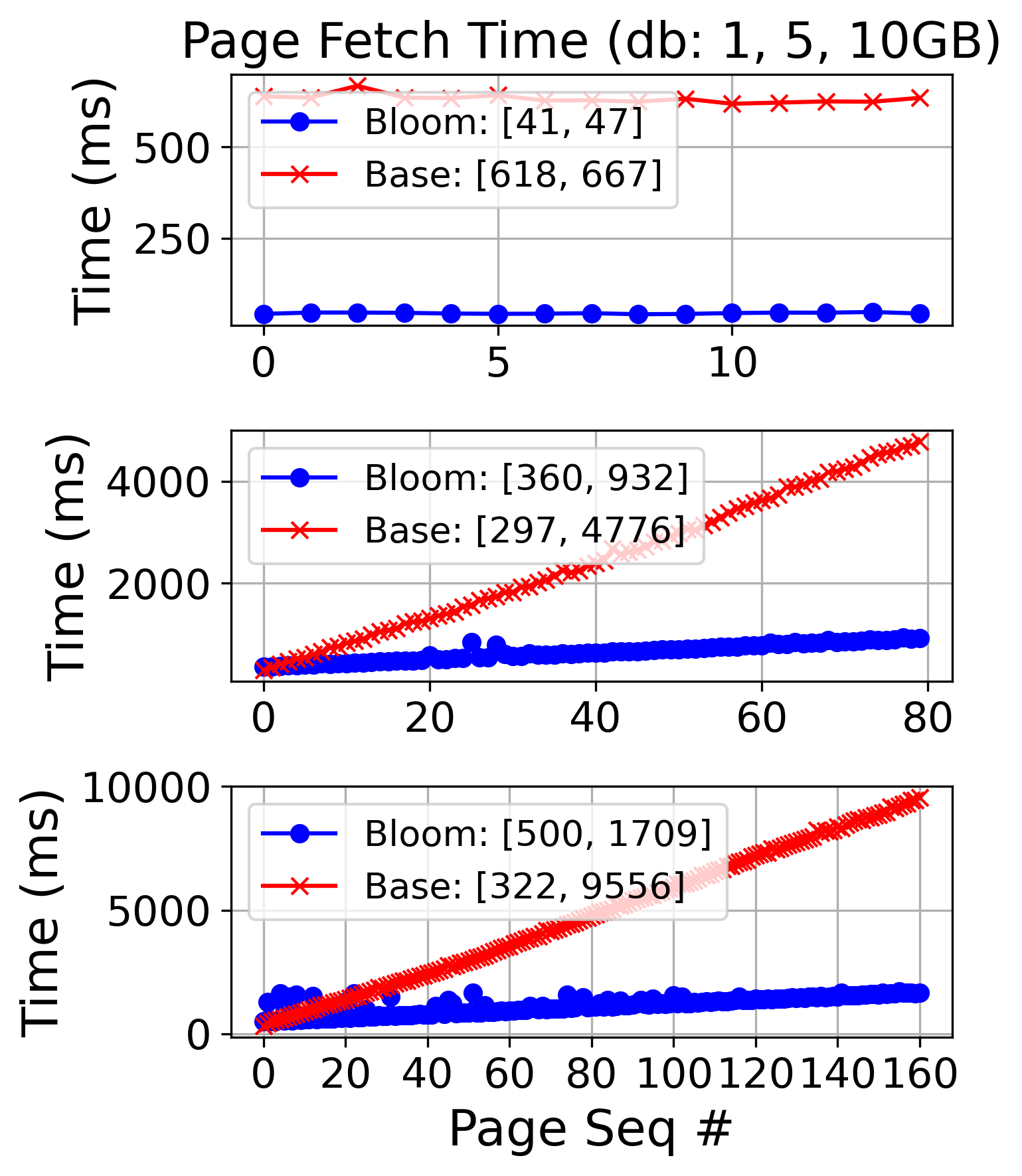}
    \caption{\mdseries\itshape Varying database size; $50$-row pages.}
    \label{fig:bloom-vary-db}
\end{subfigure}
\caption{\mdseries Bloom filtering vs.\ baseline: time to fetch each page.
    $[\cdot,\cdot]$ in legend shows min/max page fetch times.}
\label{fig:bloom-exp}
\end{figure}

Overall, the results show that Bloom filtering is very effective for Q2.
When the database size is relatively small (\Cref{fig:bloom-vary-pg}),
the baseline consistently takes around $600$ms to reproduce a page,
because its execution cost is dominated by the computation of the entire result and sorting them first;
in contrast, Bloom filtering takes $28$ to $74$ms, with larger pages requiring more time.
Under larger database sizes  (\Cref{fig:bloom-vary-db}),
baseline switches to a different plan whose cost grows linearly with the starting position of the page fetch,
eventually matching the cost of executing the entire query when fetching pages near the tail;
Bloom filtering performance remains scalable and depends much less on the fetch position,
saving as much as $3.8$s ($5.1\times$ speedup) and $7.8$s ($5.6\times$ speedup) for $5$GB and $10$GB databases, respectively.

\vspace{-2mm}
\subsection{Effectiveness of Sargable Filtering}
\label{sec:experiments:sargable}

This setup uses the joined \& filtered table for benchmark Q7:
\begin{lstlisting}[language=SQL, basicstyle=\footnotesize\ttfamily]
SELECT ROW(s_suppkey, l_orderkey, l_linenumber, o_orderkey, c_custkey, n1.n_nationkey, n2.n_nationkey), *
FROM supplier, lineitem, orders, customer, nation n1, nation n2
WHERE ...; -- same as original query
\end{lstlisting}
We create sargable range filters for all index columns in the input tables based on the precomputed milestones,
but only if the ranges cover no more than 30\% of the domain, as discussed in \Cref{sec:system:pagination}.
We run two sets of experiments similar to those for Bloom filtering,
with results shown in \Cref{fig:pushdown-exp}.
Overall, we observe that the baseline using \OFFSET\ and \LIMIT\ performs progressively worse when it fetches later pages, until its running time plateaus when the plans switch to essentially computing and sorting the entire result;
in contrast, sargable filtering performs well across all pages regardless of their position,
and its advantage over the baseline widens dramatically toward later pages.
With a $1$GB database (\Cref{fig:pushdown-vary-pg}), baseline can take up to $821$ms to fetch a page,
while sargable filtering takes no more than $73$ms.
Fixing the page size at $50$ and using larger databases (\Cref{fig:pushdown-vary-db}), sargable filtering's benefit is even greater.
For the $10$GB database, the baseline takes about $19$s to fetch a page positioned at 1/6-th of the entire result or later,
while sargable filtering takes about $1.2$s ($16\times$ speedup) in the worst case.

\begin{figure}[t]
\begin{subfigure}[t]{0.47\columnwidth}
    \includegraphics[width=1.04\linewidth]{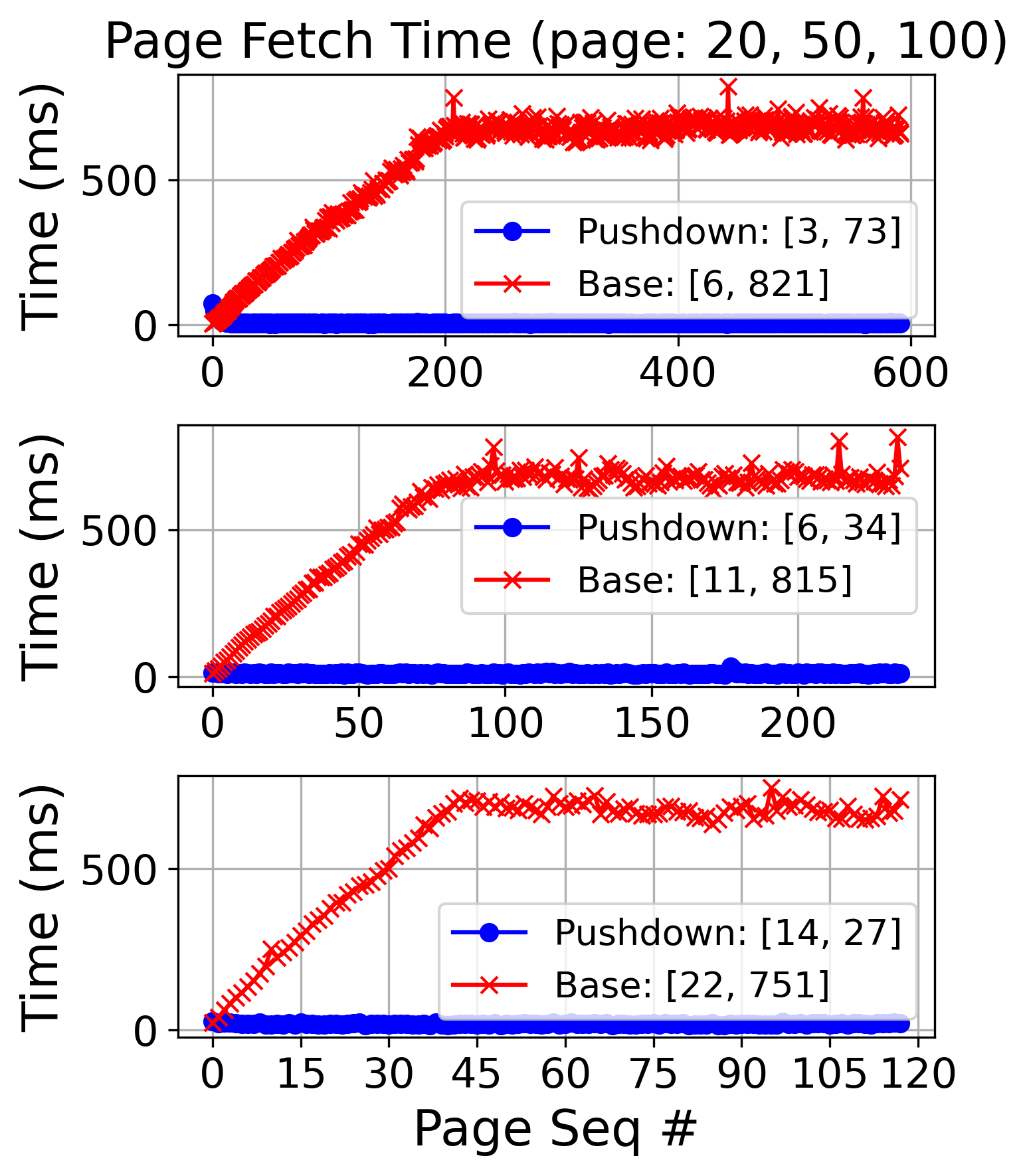}
    \caption{\mdseries\itshape Varying page size; $1$GB database.}
    \label{fig:pushdown-vary-pg}
\end{subfigure}
\hfill
\begin{subfigure}[t]{0.52\columnwidth}
    \includegraphics[width=0.92\linewidth]{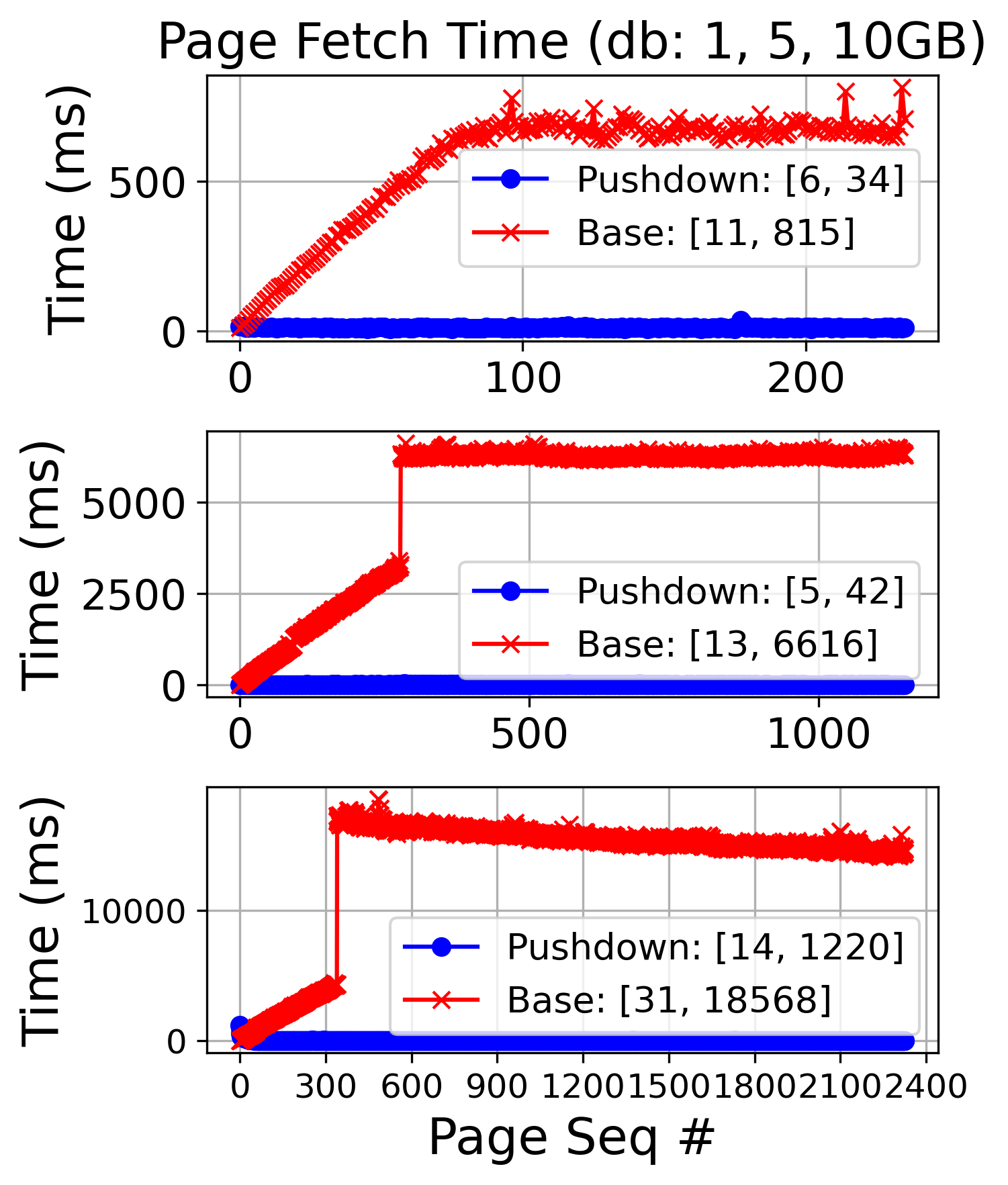}
    \caption{\mdseries\itshape Varying database size; $50$-row pages.}
    \label{fig:pushdown-vary-db}
\end{subfigure}
\caption{\mdseries Sargable filtering vs.\ baseline: time to fetch each page.
    $[\cdot,\cdot]$ in legend shows min/max page fetch times.}
\label{fig:pushdown-exp}
\end{figure}

\vspace{-3mm}
\subsection{General Evaluation}
\label{sec:experiments:all}

For general evaluation, we apply all \oursys\ optimizations to all TPC-H queries,
and compare the execution times of the resulting page-fetch queries with those of the baseline approach.
IID-based and sargable filtering can be applied to all queries, while Bloom filtering is only applicable to Q4, Q16, Q18, and Q20--22.
Because of limited space, we only show in \Cref{tab:page-query-result} the results for the joined \& filtered table of Q18, which is the most expensive query in our experiments;
remaining results are in \Cref{appendix:experiments}.
The results generally echo the findings from the experiments on Bloom/sargable filtering,
and confirm the benefit of combining optimizations.
Specifically in \Cref{tab:page-query-result}, \oursys\ performs no worse than the baseline for pages at the beginning of the results, and significantly better for later pages. 

\begin{table}[t]
    \centering\scriptsize
    \renewcommand{\tabcolsep}{1mm}
    \begin{tabular}{|c|r|r|r|r|r|r|}
         \cline{2-7}
         \multicolumn{1}{c|}{} & \multicolumn{2}{c|}{1GB (ms)} & \multicolumn{2}{c|}{5GB (ms)} & \multicolumn{2}{c|}{10GB (ms)} \\ \hline
         Page & \oursys\ & Baseline & \oursys\ & Baseline & \oursys\ & Baseline  \\ \hline
         head & 2,532.81 & 5,870.729 & 14,184.028 & 46,234.889 & 37,419.067 & 120,576.875 \\
         middle & 2,456.265 & 7,991.925 & 16,215.405 & 55,316.965 & 32,434.602 & 133,646.155  \\
         tail & 2,787.514 & 10,113.122 & 13,983.933 & 64,399.041 & 33,393.882 &  146,715.434  \\ \hline
    \end{tabular}
    \caption{\mdseries \oursys\ vs.\ baseline: time to fetch a page in the joined \& filtered table of Q18;
        $50$-row pages and varying database size.
        Here, head refers to the first query while tail refers to the second to the last (as the last page sometimes is not full).}
    \label{tab:page-query-result}
\end{table}

Finally, we evaluate the overhead of debugging context initialization, which is dominated by the cost of running milestone queries.
We compare their costs with those of computing the entire results of the corresponding original query blocks.
Because of limited space, we only show in \Cref{tab:milestone-vs-table} the results for the three most expensive milestone/original query pairs, with a $1$GB database and $50$-row pages;
the remaining results are in \Cref{appendix:experiments}.
All three pairs come from Q18's stages.
Since milestone queries conceptually summarize the output of the original queries,
we do not expect them to run faster than the latter.
As \Cref{tab:milestone-vs-table} shows, their performance is comparable to the original queries.
Recall from earlier experiments that baseline page-fetch queries often cost as much as the original queries, this observation implies that \oursys\ can benefit from milestone-enabled optimizations for many page-fetch queries by paying only a one-time overhead equivalent to a single baseline page-fetch query. The execution times reported are taken from \sql{EXPLAIN} \sql{ANALYZE},
which only measures the time spent in executing the query but not transmitting its result.
We intentionally chose to exclude the latter because the outputs from the original queries can be too large to transmit.
For example, printing the entire Q18's joined \& filtered table took more than an hour for a \sql{psql} client running on the database console.
In contrast, as \Cref{tab:milestone-vs-table} shows,
milestone queries return much smaller results, which are feasible to transmit to and handle by the client.

The end-to-end latency of \oursys\ consists of time for page-fetch query execution, network transmission, and frontend rendering. While query times naturally vary and network latency fluctuates unpredictably, the frontend rendering is consistently lightweight. On average, rendering a 50-tuple page takes $\sim$600ms, forward tracing $\sim$300ms, and pinning $\sim$500ms, providing an interactive experience.

\begin{table}[t]
    \vspace{-5mm}
    \centering\scriptsize
    \renewcommand{\tabcolsep}{1mm}
    \begin{tabular}{|c|r|r|r|r|}
         \cline{2-5}
         \multicolumn{1}{c}{} & \multicolumn{2}{|c|}{Milestone Query} & \multicolumn{2}{c|}{Original Query} \\ \hline
         Q18 & Exec.\ time (ms) & Output (MB) & Exec.\ time (ms) & Output (MB) \\ \hline
         joined \& filtered & 15,984.578 & 52 & 12,086.886 & 1272 \\
         group & 13,193.212 & 9 & 12,602.33 & 64 \\
         output & 13,808.861 & 9 & 14,452.551 & 75 \\ \hline
    \end{tabular}
    \caption{\mdseries Milestone vs.\ original queries: execution time and output size;
        $50$-row pages and $1$GB database.}
    \label{tab:milestone-vs-table}
\end{table}

\mypar{Limitations and Opportunities}
While \oursys\ aims at making SQL debugging interactive,
for queries that are inherently hard to optimize and expensive to evaluate, debugging them remains challenging.
For example, the latency in debugging context initialization is dominated by milestone queries,
which cannot be expected to outperform the original query (other than transmitting less result data).
Possible directions to increase interactivity include on-demand and/or approximate milestone computation,
which are promising venues for future research.
An orthogonal approach is to find small database instances that can help reveal potential mistakes in the query being debugged.
Sampling is one method, but smarter query-driven methods have also been studied; see \Cref{sec:related} for more discussion.
This approach complements \oursys\ by reducing the size of the database used to debug in the first place.

%% file: sections/5-user-study.tex
\section{User Study}
\label{sec:user-study}

We conducted a user study in an undergraduate database course to evaluate the effectiveness of \oursys\ on two aspects: (1) whether \oursys\ helps users catch more logical bugs, and (2) whether \oursys\ reduces the time to find bugs. 

\mypar{Participants}
We had 237 student participants who had just become familiar with SQL at the time of the study.
Participation was voluntary.
We considered recruiting participants from other sources (e.g., Amazon Mechanical Turk),
but decided against it as it was hard to ensure participants' SQL familiarity was at a similar level.
Since SQL familiarity significantly impacts debugging time, the lack of control can make results difficult to interpret.

\mypar{Preparation and Setup}
The study was conducted for two consecutive weeks during the course's once-a-week 75-minute discussion sessions in the Fall 2022 semester. 
In the first session, students were given a tutorial on \oursys\ and informed about the format of the quiz, which contained two SQL debugging problems P1 and P2.
After the first discussion, \oursys\ is made public for students to try.
In the second session, students completed the quiz synchronously in a proctored environment, where they were asked not to discuss with classmates.
For each problem in the quiz, students were provided a problem statement, an incorrect query, its incorrect output, and the correct output (details in \Cref{appendix:user-study}).
To create treatment and control groups, students received the two problems in a random order.
For the first problem they received, they were free to use any tool of their choice (e.g., Gradescope autograder, \sql{psql} console, \sql{pgAdmin} Web interface) except \oursys.
For the second problem, \oursys\ was made \emph{available but optional} along with other tools. This user study was conducted before any AI tool was released; thus, no AI assistant was available to students.
For each problem, students were to describe the bugs found in a free-response text box.
P1 had two bugs and P2 had three, and the students were not told how many.
While students self-paced, they were recommended to spend 15-20 minutes on each problem.
They were not allowed to move on to the second problem until they submitted an answer to the first problem.

\mypar{Results and Analysis}
We collected the time it took students to solve each problem, whether they chose to use \oursys, and the bugs they found.
Of 237 students, 140 completed both problems and provided legitimate answers. 
Therefore, we based our analysis on these 140 responses. 
For the first problem received, which must be completed without \oursys,
73 students received P1, and 67 received P2.
For the second problem received, where using \oursys\ was an option,
73 students received P2, but 36 of them did not use \oursys;
67 students received P1, but 29 of them did not use \oursys.
In summary, for P1, we have 38 submissions using \oursys, and 102 not using \oursys;
for P2, we have 37 submissions using \oursys, and 103 not.



\begin{figure}[t]
\centering
\includegraphics[scale=0.23]{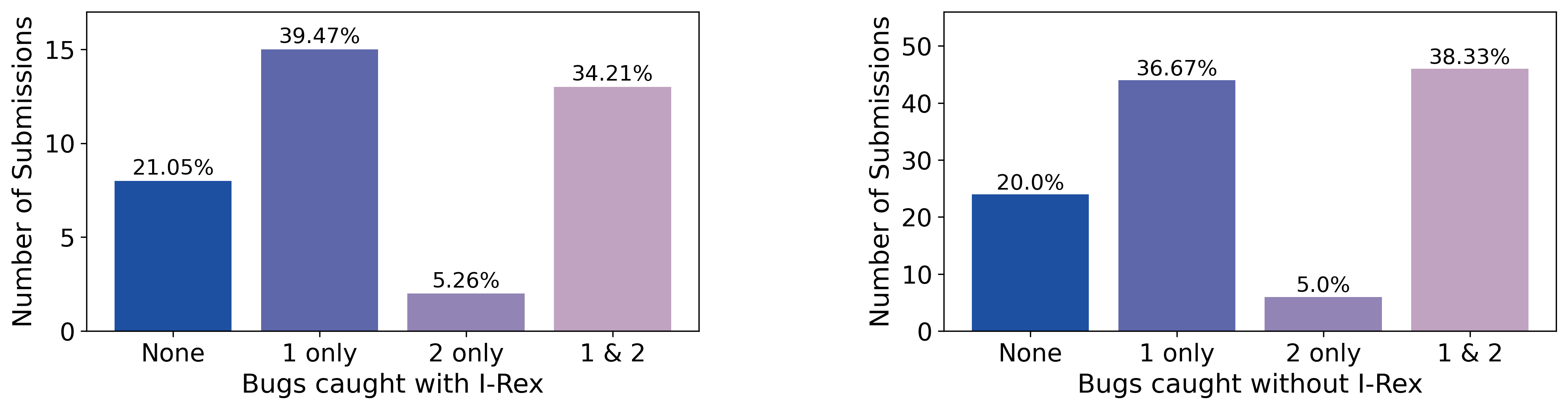}
\caption{\mdseries Bugs caught (out of two) for P1, using \oursys\ vs.\ not.}
\label{fig:error-dist-q1}
\end{figure}

\begin{figure}[t]
\centering
\includegraphics[scale=0.23]{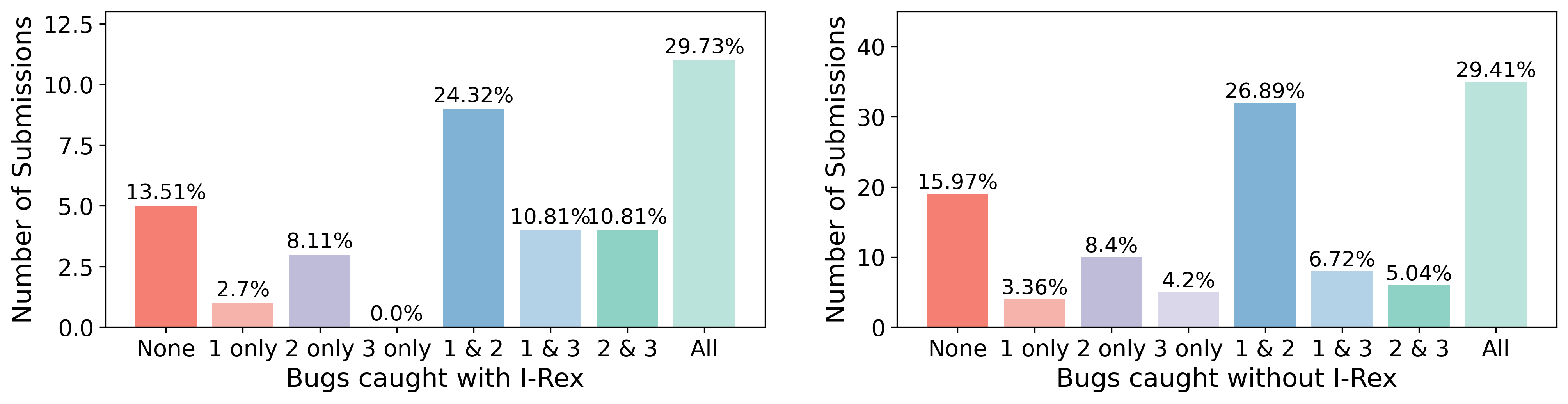}
\caption{\mdseries Bugs caught (out of three) for P2, using \oursys\ vs.\ not.}
\label{fig:error-dist-q2}
\end{figure}

We manually reviewed each response and checked whether the student correctly identified the bugs.
\Cref{fig:error-dist-q1} shows the fraction of the submissions that correctly identified each possible subset of the bugs for P1,
with (left) or without (right) help from \oursys;
\Cref{fig:error-dist-q2} does the same for P2.
While the total submission numbers differ between using/not using \oursys, the distributions (fractions for possible subsets) are similar. The use of \oursys\ had no discernible impact on the mean of bugs identified, as shown in the last column of \Cref{tab:timing}. While \oursys\ did not seem to help students find more bugs, the comparison of debugging times, reveals that the use of \oursys\ reduces the mean debugging time by around 8 minutes for both problems. We further prove \oursys's efficiency through statistical tests. Due to the non-normal distribution of the collected time data (tested by the Shapiro-Wilk Test~\cite{shapiro1965analysis}), we performed the Mann-Whitney $U$ Test~\cite{mann1947test} for P1 and P2 separately. With the null hypothesis being ``\oursys\ makes no difference or slows students down'', the p-values for P1 and P2 are 0.0001 and 0.0007, respectively. Since they are below the standard 0.05 threshold, we reject the null hypothesis and conclude that \oursys\ significantly improves students' efficiency in finding bugs without compromising accuracy.

\mypar{Student Feedback}
Students also provided anonymous feedback on the use of \oursys.
A common theme is that \oursys\ has a noticeable learning curve:
since it exposes step-by-step logical execution and displays numerous visual elements simultaneously,
some students initially found it intimidating.
However, those who pushed past this familiarization phase found it a powerful and effective debugging aid.
There were also students who preferred using traditional clients (like \sql{pgweb}) to manually construct auxiliary queries,
highlighting that debugging preferences remain highly subjective and tied to individual cognitive styles.
Just like programmers have different preferences for debugging in general
(from printing to console to sophisticated IDE-supported execution tracing and state inspection),
we envision \oursys\ not as a one-size-fit-all SQL debugging solution, but as a useful addition to the SQL debugging toolbox.

\mypar{Discussion}
There are several limitations to this user study.
First, its participation was voluntary with no extrinsic incentives, so students might not feel compelled to invest the upfront time needed to overcome the tool's learning curve.
This likely explains why many opted against using \oursys\ for the second problem.
Besides the lack of incentives, the strict time constraints may have also reduced their overall debugging effort.
To address this, future evaluations should develop better incentives for participants to fully acclimate.

Second, limited by time and subject pool, this study only evaluated undergraduates on pedagogical queries.
Future work is needed to assess the tool's effectiveness among industry professionals debugging complex queries over enterprise-scale databases.
Such settings may also help answer the question of whether \oursys\ can help users find new, harder bugs,
beyond making bug finding faster.

Third, this study does not explore the use of large language models (LLMs) as a potential alternative or augmentation to \oursys.
The fast evolution of LLMs is rapidly changing how SQL queries are composed and debugged.
While many text-to-SQL frameworks \cite{ghosh-etal-2025-sqlgenie, gao2024text, talaei2024chess, li2023resdsql, pourreza2025chasesql, pourreza2023din} have shown good accuracy on benchmarks such as BIRD~\cite{li2024can} and Spider 2.0~\cite{lei2024spider}, query debugging remains a major challenge as shown in~\cite{ye2026texttosqlllmsreallydebug}.
To investigate the increasingly common use of LLMs for SQL debugging, we conducted a small case study (details in~\Cref{appendix:user-study})
by feeding the incorrect queries from our user study to an LLM.
Despite the relative simplicity of these educational queries, the LLM still occasionally misidentified bugs or hallucinated incorrect interpretations.
Because LLMs lack formal correctness guarantees and produce inherently non-deterministic outputs,
conducting a direct, reproducible comparison is difficult.
We postulate that LLMs and \oursys\ serve orthogonal purposes:
LLMs act as probabilistic oracles designed to instantly suggest a patch,
whereas \oursys\ is an interactive tool for investigating logical execution to understand query behavior under unforeseen circumstances.
Rather than competing, \oursys\ provides the deterministic ground truth that LLMs currently lack.
Specifically, future LLM-driven agents could directly call \oursys\ backend API server
to examine intermediate query execution states accurately and efficiently,
thereby significantly improving their reliability.

\begin{table}[t]
\centering\scriptsize
\begin{tabular}{|l|r|r|r|r|}
\hline
\textbf{Problem} & \textbf{\# response} & \textbf{Mean time} & \textbf{Mean bugs found} \\ \hline
P1 w/ \oursys & 38 & 12.85 min & 1.13 / 2 \\ 
P1 w/o \oursys & 102 & 20.59 min & 1.18 / 2 \\ \hline
P2 w/ \oursys & 37 & 19.12 min & 1.91 / 3 \\ 
P2 w/o \oursys & 103 & 27.63 min & 1.85 / 3 \\ \hline
\end{tabular}
\caption{\mdseries Mean debugging time and bugs caught per problem\cut{, with and without \oursys}.}
\label{tab:timing}
\end{table}




%% file: sections/6-related-work.tex
\section{Related Work}
\label{sec:related}

Earlier versions of \oursys\ \cite{miao2020rex, hu2022rex} precomputed and stored all debugging information in the client, which did not scale.
The version described in this paper has refined existing features and introduced new ones,
and most importantly, added scalability support, which requires a redesign of the system and new optimization techniques.

\mypar{Finding Logical Errors in Queries}
Work toward finding logical errors in queries can be classified into two categories.
The first category
assumes knowledge of a correct reference query
and uses it to help identify errors in an incorrect query.
XData~\cite{DBLP:journals/vldb/ChandraCKRS015} checks the correctness of a query by running the query on self-generated testing datasets.
Cosette~\cite{DBLP:conf/cidr/ChuWWC17, DBLP:conf/pldi/ChuWCS17, DBLP:journals/pvldb/ChuMRCS18}, SQLSolver~\cite{DBLP:journals/pacmmod/DingWYZXCPL23}, and QED~\cite{DBLP:journals/pvldb/WangPC24} test query equivalence using constraint solvers and theorem provers.
RATest~\cite{DBLP:conf/sigmod/MiaoRY19} aims at constructing small and illustrative database instances to show the differences between two queries.
\cite{DBLP:conf/comad/ChandraBHJ021} develops a grading system that canonicalizes queries with rewrite rules and then decides query similarity using edit distance between the resulting logical plans.
SQLRepair~\cite{DBLP:conf/icse/Presler-Marshall21} and QR-Hint~\cite{DBLP:journals/pacmmod/HuGSRY24} focus on fixing the wrong query by proposing syntactical edits.
\oursys\ differs from this line of work as it assumes no reference query.
Dropping this assumption fundamentally changes the problem and makes existing solutions inapplicable in more general settings,
although they can complement \oursys\ when reference queries are available.

The second category of work helps users examine a query without a given reference query.
Qex~\cite{DBLP:conf/lpar/VeanesTH10} generates input relations and parameter values for unit-testing SQL queries.
SQLLint~\cite{DBLP:conf/gvd/BrassG04, DBLP:conf/qsic/BrassG05, DBLP:journals/jss/BrassG06} looks for patterns in the query indicative of common semantic errors and alerts users to them.
C-instances~\cite{DBLP:conf/sigmod/GiladMRY22} are abstract instances aimed at illustrating different ways a query can be satisfied.
Interactive query builders and visualizers~\cite{rapidsql, pgadmin, msaccess, dbforge, abouzied2012dataplay, haas1989extensible, cerullo2007system, jaakkola2003visual, leventidis2020queryvis, miedema2021sqlvis} use diagrams to gain intuitive understanding of query logic. 
Frameworks for data exploration~\cite{abouzied2012dataplay, le2019explique, dimitriadou2014explore, akbarnejad2010sql} use signals such as query history and user feedback to suggest queries, which may serve as a debugging approach.
A line of work known as algorithmic debugging~\cite{caballero2012algorithmic, caballero2012declarative} guides users through a series of questions on whether intermediate query steps produce intended results.
Also useful to debugging is explaining what queries do in natural language~\cite{koutrika2010explaining, shu2021logic, iyer2016summarizing, xu2018sql, gehrmann2018end, brown2020language}.
None of the above supports debugging by tracing execution, as \oursys\ and most GPL debuggers do.

Two systems in this second category, DESQL~\cite{DBLP:journals/pacmse/Haroon0G24} and Habitat~\cite{grust2013observing, dietrich2015sql}, are the closest to \oursys\ in approach.
DESQL~\cite{DBLP:journals/pacmse/Haroon0G24} is a debugger for Spark SQL, which decomposes the query into subqueries and helps users examine subqueries' output; it does not support correlated subqueries.
At a conceptual level, DESQL fundamentally differs from our work because
it adopts an operator-based view of execution that is closer to optimized execution plans than to the way SQL queries are written.
Technical approaches also differ significantly:
DESQL relies on Spark as the execution backend instead of traditional database systems,
and assumes that debugging commences only after the instrumented query fully executes.
Habitat~\cite{grust2013observing, dietrich2015sql} allows users to mark SQL subexpressions and inspect intermediate results side by side, connecting related result rows.
It also lets users filter these results to focus on a subset.
However, Habitat differs from \oursys\ in important ways.
First, on scalability, \oursys\ avoids computing full results of a query,
while Habitat executes a query in full and shows its results in bulk.
Although Habitat's focus filters can restrict the query, they need to be defined manually and explicitly;
in contrast, \oursys's pinning and automatic pagination are intuitive and demand less user effort,
and \oursys\ has optimizations to improve query efficiency.
Second, \oursys\ and Habitat have different conceptual designs:
\oursys\ defines its canonical query execution to be row-oriented and reproducible (including row ordering),
while Habitat presents a set-oriented view of SQL execution.
Finally, Habitat is not publicly available for comparison.

\mypar{Other Areas of Related Work}
Several areas of research are related to some of the ideas used by \oursys.
First, tracing and pinning in \oursys\ use information about data provenance~\cite{buneman2001and, green2007provenance, huang2008provenance, chapman2009not, bidoit2014query}, which explains why and how a particular row is or is not produced by a query.
Many previous works~\cite{agrawal2006trio, cui2000tracing, karvounarakis2013collaborative, glavic2009perm, arab2018gprom, lee2019pug, psallidas2018smoke, interlandi2015titian, diestelkamper2020tracing, amsterdamer2011putting} have focused on making provenance capture a built-in feature in database systems.
In particular, \cite{DBLP:journals/pvldb/MullerDG18} proposes techniques for capturing provenance by rewriting SQL queries, just as \oursys\ uses query rewriting to help support tracing and other debugging operations.
However, \cite{DBLP:journals/pvldb/MullerDG18} captures full provenance information for all query results.
In contrast, for scalability and interactivity, \oursys\ computes provenance information on demand to avoid the overhead of full computation and materialization. 

Second, efficient pagination of query results has been studied in query optimization literature.
Recent advances have been made on direct access to query answers~\cite{DBLP:journals/tods/CarmeliTGKR23, DBLP:conf/pods/CarmeliZBKS20, DBLP:conf/icdt/EldarCK24, DBLP:journals/pvldb/TziavelisGR21},
but their techniques require specialized indexes and apply only to conjunctive queries,
less general than what \oursys\ supports.
Work on data skipping and predicate pushdown~\cite{DBLP:journals/pvldb/SudhirTLHCM23, DBLP:journals/pvldb/NiuGLLGKLP21, DBLP:journals/pacmmod/YanLH23, DBLP:journals/pvldb/OrrKC19} also requires significant changes to database system internals,
but \oursys\ chooses to leverage existing database systems for efficient pagination.
\oursys's design also motivates unique optimizations, such as client-cached milestones.

%% file: sections/7-conclusion.tex
\section{Conclusion}
\label{sec:conclusion}

In this paper, we have presented \oursys, a novel, interactive, scalable, and easy-to-deploy SQL debugger that, given a database instance, enables users to visualize the logical execution of SQL queries and thus debug query semantics in a GPL style. \oursys\ executes a query in a canonical fashion that faithfully follows how the query is written,
and lets users examine the execution using a rich set of features, including but not limited to those found in GPL debuggers.
\oursys\ supports efficient exploration of any arbitrary point of execution without fully executing the underlying query.
It does so by fully leveraging the database system: it formulates the selective computation of debugging information around the point of interest as SQL queries and employs materialization and query rewrite strategies to ensure efficient execution.
Our performance experiments and user study demonstrate the efficiency and effectiveness of \oursys.


There are multiple directions for future work.
First, we can extend \oursys\ to cover the SQL constructs that we do not already support, as mentioned in \Cref{sec:debug-paradigm}.
First, it is worth investigating milestone designs beyond single ranges and Bloom filters used in \Cref{sec:system:pagination}.
Second, more debugging features can be added, such as setting watchpoints to observe filter conditions in actions.
Finally, we plan to extend \oursys\ so that it can isolate and help debug runtime SQL errors (such as division by zero),
for which database systems usually provide uninformative feedback.

%% file: sections/appendix/appendix.tex
\appendix
\input{sections/appendix/debug-paradigm}
\input{sections/appendix/system}

\input{sections/appendix/experiments}
\input{sections/appendix/user-study}

%% file: sections/appendix/debug-paradigm.tex
\section{Additional Details on Debugging Paradigm}
\label{appendix:debug-paradigm}

\subsection{Set/Bag Operation Blocks}
\label{appendix:debug-paradigm:set}

A SQL set/bag operation block has the form
\begin{lstlisting}[language=SQL, mathescape=true, basicstyle=\footnotesize\ttfamily]
$Q_R$ UNION|INTERSECT|EXCEPT [ALL] $Q_S$
\end{lstlisting}
where $Q_R$ and $Q_S$ are subqueries.
Let $R$ and $S$ denote the tables returned by $Q_R$ and $Q_S$, respectively.
The canonical execution procedure of this block consists of two stages:
a sort stage that sorts both $R$ and $S$ by their row contents
and a final output stage that merges them to produce the result table.
To determine the order for the sort stage,
we pick a particular ordering of all columns, preferring to reuse $R$'s order (dictated by its IID) as much as possible.
Further, to ensure a stable sort in the case of duplicates, we append each input table's IID to the column ordering as needed when sorting the table.
The synthesized sort result IIDs have the same format.
For instance, suppose that the IID for $R(A_1, A_2, A_3)$ is $(A_3, A_1)$, which implies that $R$ is free of duplicates.
In this case, $R$ is already sorted by $(A_3, A_1, A_2, \text{IID}(R)))$,
and we would sort $S(B_1, B_2, B_3)$ accordingly by $(B_3, B_1, B_2, \text{IID}(S))$, assuming that $S$ may contain duplicates.
As another example, suppose that both input tables come from unsorted base tables with duplicates, and these tables' internal row ids serve as the IIDs.
In this case, we would sort both tables by their columns, in order, followed by their respective IIDs.

The final output stage merges the two sorted tables.
The output row IID always contains the value of columns in the same order as the sort stage IID, followed by
a Boolean flag indicating the source input table ($0$ for $R$ and $1$ for $S$)
and a sequence number ($0$-based) indicating its position among duplicates in the source input table.
Depending on the set/bag operation involved, the stage's behavior and the last two components of the IID are defined differently.
In the following, let $t$ denote the content of a row in either $R$ or $S$, and
$R[t]$ and $S[r]$ the lists of duplicate rows in $R$ and $S$ (respectively) with this content.
\begin{itemize}[leftmargin=*]
\item \INTERSECT: We consider only the case when $R[t] \neq \emptyset$, and output only the first row of $R[t]$,
    with the last two components of the IID set to $(0, 0)$.
\item \EXCEPT: We consider only the case when $R[t] \neq \emptyset$.
    Only if $S[t] = \emptyset$, we output the first row of $R[t]$,
    with the last two components of the IID set to $(0, 0)$.
\item \UNION: If $R[t] \neq \emptyset$, we output the first row of $R[t]$,
    with the last two components of the IID set to $(0, 0)$.
    otherwise, we output the first row of $S[t]$,
    with the last two components of the IID set to $(1, 0)$.
\item \INTERSECT\ \sql{ALL}: Let $m = \min\{|R[t]|, |S[t]|\}$.
    We output the first $m$ rows of $R[t]$,
    with the last two components of the IID set to $(0, 0), \dots, (0, m-1)$.
\item \EXCEPT\ \sql{ALL}: Let $m = |R[t]| - |S[t]|$.
    Only if $m > 0$, we output the first $m$ rows of $R[t]$,
    with the last two components of the IID set to $(0, 0), \dots, (0, m-1)$.
\item \UNION\ \sql{ALL}: Let $m = |R[t]| + |S[t]|$.
    We output all rows of $R[t]$ followed by all rows of $S[t]$,
    with the last two components of the IID set to $(0, 0), \dots, (0, |R[t]|-1)$ and then $(1, 0), \dots, (1, |S[t]-1|)$.
\end{itemize}

The input (``combo'' would be somewhat a misnomer here) space for the debugging context is the concatenation of the rows of $R$ followed by those of $S$.
Hence, from the two input tables, only one row can be highlighted as currently active,
and no input pinning is allowed because it is not useful in this context.
Although the definitions above associate each output row with a particular input row,
the \oursys\ interface hides this detail, so users will still observe the SQL semantics.
To be more specific, given an input row,
we define its derivative ``row'' (if one exists) in the final output table as the group of all duplicate rows,
which can only be pinned/unpinned together.
The interface also provides an explanation for the number of duplicates or the absence of output rows for the given input row.

\subsection{Join Expressions}
\label{appendix:debug-paradigm:join}

\oursys\ treats explicit \JOIN\ expressions in \FROM\ essentially as subqueries.
Hence, each such expression gives rise to a separate debugging context with two input tables (base or derived).
Inner joins can be handled in the same way as a more general \SELECT\ block (\Cref{sec:debug-paradigm:model}).
We focus on outer joins here.
Let $R$ and $S$ denote the two input tables.
We augment the input combo space by adding a special IID $\bot$ to each input table whose side of the output can be padded with \sql{NULL}s.
Intuitively, $\bot$ stands for ``no row'' from that input table, and behaves in the \oursys\ interface as a special last row of the table.
For example, for $S$ in $R$ \sql{LEFT} \JOIN\ $S$,
we add $\bot$ to indicate ``no row from $S$.''
If a row $r \in R$ joins with no row from $S$, the left outerjoin result should contain a row for $r$ padded with \sql{NULL}s in $S$ columns;
its IID would be $(r, \bot)$.

When ordering by IIDs, $\bot$ should be considered last in value.
Internally, \oursys\ represents $\bot$ in SQL using an IID whose components are all \sql{NULL}s,
assuming input rows do not have \sql{NULL}s for all their IID components.%
\footnote{In the extremely rare case when this assumption is violated, we choose an unused value from each domain instead.}
When generating queries involving comparison with IIDs possibly containing \sql{NULL}s, however,
we need to replace the \sql{ROW(...)} comparisons with special code
because SQL comparisons involving \sql{NULL}s always yield the \sql{UNKNOWN} truth value.
For example, instead of generating \sql{ROW($A$,$B$)>=ROW(1,2)} when $B$ can be \sql{NULL},
we would generate
\begin{lstlisting}[language=SQL, mathescape=true, basicstyle=\footnotesize\ttfamily]
$A$>1 OR ($A$=1 AND ($B$ IS NULL OR $B$>=2))
\end{lstlisting}

Selection of the active input combo and stepping work the same way as in \SELECT\ (\Cref{sec:debug-paradigm:model}),
except when the input combo involves a $\bot$,
\oursys\ show an explanation of why a \sql{NULL}-padded result row is or is not produced,
instead of the filter expression tree.

%% file: sections/appendix/system.tex
\section{Additional Details on Debugging Paradigm}
\label{appendix:system}

\subsection{Proof of \Cref{thm:bloom}}

\begin{proof}
First note that when $\Theta \Leftrightarrow \bot$ (i.e., $\Theta$ is equivalent to logical false), both queries are obviously equivalent as they return empty results on all possible database instances regardless of the evaluation of \sql{BLOOM\_CHECK}. Therefore, we now proceed to prove that both queries are equivalent when at least $Q_1$ returns a non-empty result (i.e., $\Theta$ is not equivalent to $\bot$).

Now assuming $Q_1$ and $Q_2$ are not equivalent, then there must exist input tables $\{R'_i\}$ where $Q_1, Q_2$ return results $\mathcal{V}_1$ and $\mathcal{V}_2$ respectively, and a tuple $e' \in \mathcal{V}_1$ but $e' \notin \mathcal{V}_2$ (it cannot happen vice versa as $\Theta \Leftarrow \Theta \text{ AND } \sql{BLOOM\_CHECK}(F, \langle \mathcal{A} \rangle)$), thus there must exist a subset $\mathcal{A'}$ from $e'$ such that \sql{BLOOM\_CHECK}($F, \langle e'[\mathcal{A'}] \rangle$) is evaluated to false. Observing that $F$ is obtained from $\sql{BLOOM\_CHECK}(\mathcal{V}_1[\mathcal{A'}])$ and $e' \in \mathcal{V}_1$, this implies \sql{BLOOM\_CHECK}($F, \langle e'[\mathcal{A'}] \rangle$) must be true as Bloom filter returns no false negative. Therefore, $e'$ must also be in $\mathcal{V}_2$ and $\{R'_i\}$ does not exist. As a result, $Q_1$ and $Q_2$ must be equivalent. 
\end{proof}

\subsection{Optimizing Tracing and Pinning}
\label{appendix:system:positioning}

With optimized page fetches, users can freely move around a table with a paginated display.
However, a debugging context displays multiple tables,
and \oursys\ must coordinate the pages displayed across tables to show the end-to-end derivation from the current input combo to output,
so that users can forward- or backward-trace (using pinning).
We now discuss how to optimize these operations.

\subsubsection{Forward Tracing}

For forward tracing, we already have access to the input combo as well as their IIDs and row contents.
Conceptually, to forward-trace through a particular stage,
it suffices to determine the IID (say $t$) of the derivative row produced by this stage, if one exists.
Then, given the target IID $t$, \oursys\ searches the milestone table of the stage's result table for the page whose IID range contains $t$,
and fetches this page.
If the fetched page indeed contains $t$, we have successfully forward-traced through the stage;
otherwise, we know that the input combo yields no result rows.

Determining the target IID in the first place is usually straightforward, thanks to the logical nature of our IIDs.
For example, given the input combo $\langle s_2, f_1 \rangle$ in \Cref{fig:debug-scene},
the target IID for the join \& filter stage is simply the concatenation of the IIDs of $s_2$ and $f_1$.
In cases where we cannot easily determine the target IID, \oursys\ can compute it using a query.
For example, to forward-trace through the group stage, we need to compute the \GROUPBY\ expression value,
the leading component of the result row IID.
In this running example, the \GROUPBY\ expression is simply \sql{s.bar},
so we could have read its value \sql{Edge} directly from row $s_2$.
Otherwise, in general cases where the \GROUPBY\ expression is complex (e.g., involving subqueries),
\oursys\ would generate a query to compute its value, e.g.:
\smallskip
\begin{lstlisting}[language=SQL, basicstyle=\scriptsize\ttfamily]
SELECT s.bar -- arbitrary GROUP BY expression
FROM Serves s, Frequents f -- same as the original query block
WHERE (s.bar, s.beer) = ('Edge', 'Amstel') AND (f.bar, f.beer) = ('Ben', 'Edge'); -- use IID values to specify input combo
\end{lstlisting}
The cost of such queries is negligible because of the highly specific \WHERE\ condition.

\subsubsection{Pinning for Backward Tracing and Watchpointing}

Without loss of generality, assume that a single derivative row is pinned.%
\footnote{Since \oursys\ enforces that there is only one derivative row per stage,
it suffices to consider the derivative row in the earliest stage.}
Our goal is to determine the first (lexicographically) input combo in the pinned subspace.
If only one input combo contributes to the pinned derivative row, we can simply infer the former from the pinned row's IID
(the same applies when backward-tracing from a stage to its preceding stage).
For example, in \Cref{fig:debug-scene}, if a user pins the first member row of $g_1$ in the group table,
the IID of this row, $(\sql{Edge}, (s_2, f_1))$, will reveal the input combo $\langle s_2, f_1 \rangle$.

If multiple input combos contribute to one pinned derivative row, the situation is more complicated.
Such cases arise when backward-tracing from a pinned row in a post-grouping stage,
e.g., the final stage of \Cref{fig:debug-scene}.
Here, we first infer the group, identified by its \GROUPBY\ expression value (say $g$), from the IID of the pinned row.
There are two cases.
First, if there are no additional pins on the input rows, we simply need to determine the first input combo contributing to group $g$.
To this end, \oursys\ searches the milestone table of the group table for the last page whose \sql{min\_iid} is less than $(g, -\infty)$,
and fetches that page.
If the fetched page contains any member row of $g$, the IID of the first such row will reveal the desired input combo.
Otherwise, it can be shown that the next page's \sql{min\_iid} must give the desired input combo.
Hence, recovering the input combo takes only one page fetch in the worst case.

\begin{example}\label{ex:single-pin-output}
Considering \Cref{fig:debug-scene} again. Assuming a page size of 3 (rows) for all tables and $o_1$ is pinned, \oursys\ can quickly determine that the first page in the group table (which contains all tuples from $g_0$ and one tuple from $g_1$) is the last page whose \sql{min\_iid} is smaller than $\langle \sql{Edge}, -\infty \rangle$. It thus generates the following query to fetch the first page of the group table:

\begin{lstlisting}[language=SQL, basicstyle=\scriptsize\ttfamily]
SELECT (s.bar) -- group IID
       ((s.bar, s.beer), (f.drinker, f.bar)) -- input combo IID
       s.price * f.times -- sum_input
FROM Serves s, Frequents f
WHERE (...) -- original WHERE
AND ... -- sargable filtering
-- group IID and input combo IID filtering
AND ((s.bar), ((s.bar, s.beer), (f.drinker, f.bar))) >= (('Apex'), (('Apex', 'Corona'), ('Amy', 'Apex')))
AND ((s.bar), ((s.bar, s.beer), (f.drinker, f.bar))) < (('Edge'), (('Edge', 'Amstel'), ('Dan', 'Edge')))
ORDER BY 1, 2; -- first order by group, then by input combo
\end{lstlisting}

After obtaining the following result, it can be easily identified that the last tuple in the page carries the first input combo IIDs that contribute to $o_1$:

\centerline{\small
\begin{tabular}[b]{|c|c|c|}\hline
{\tt group IID} & {\tt input combo IID} & {\tt sum\_input} \\ \hline
 (Apex) & ((Apex, Corona), (Amy, Apex))) & 1 \\
 (Apex) & ((Apex, Dixie), (Amy, Apex))) & 2 \\
 (Edge) & ((Edge, Amstel), (Ben, Edge))) & 16 \\ \hline
\end{tabular}
}
\end{example}

Second, when there are additional pins on the input rows,
\oursys\ instead generates a query to determine the first input combo in the subspace further constrained by the additional pins.
This query inherits the \FROM\ and \WHERE\ clauses from the original query,
but further includes in \WHERE\ conditions to restrict input rows by the pins and ensure that \GROUPBY\ expression evaluates to $g$;
it then computes the minimum input combo in \SELECT.

\begin{example}
Continuing from \Cref{ex:single-pin-output}. When $s_2$ is further pinned in addition to $o_1$, to compute the first input combo in the subspace of the group table, \oursys\ generates the following query:

\begin{lstlisting}[language=SQL, basicstyle=\scriptsize\ttfamily]
SELECT (s.bar) -- group IID
       ((s.bar, s.beer), (f.drinker, f.bar)) -- input combo IID
FROM Serves s, Frequents f
WHERE (...) -- original WHERE
AND (s.bar, s.beer) = ('Edge', 'Amstel')
ORDER BY 1, 2 -- first order by group, then by input combo
LIMIT 1; -- only need first input combo
\end{lstlisting}

The above query locates $\langle s_2, f_1 \rangle$ as the first input combo, and \oursys\ can now decide the page to fetch in the group table by a simple binary search using the group IID and input combo IID.
\end{example}

Recall that to support watchpointing, \oursys\ automatically colors rows in each table relevant to the pinned subspace.
The case where the pinned subspace contains only one input combo is straightforward:
relevant rows are simply those in the input combo and those that are its derivatives.
Otherwise, \oursys\ extends the page-fetch query to return an additional column that indicates whether each row is relevant to the pinned subspace.
The \SELECT\ expression for this column is a Boolean expression
testing whether the input rows conform to the pins and the \GROUPBY\ expression evaluates to the same value as the pinned group.

\subsection{System Implementation Details}
\label{appendix:system:details}

\paragraph{Architecture}

\oursys\ has a client-server setup, with a Web frontend as the client and a middleware server between the client and the database server.
In a common use case in education, many student users may run debugging sessions against one large, shared, read-only database.
To ensure scalability and easy deployment,
we adhere to a strictly stateless design, where neither the middleware nor the database server stores any session-specific information:
\begin{itemize}[leftmargin=*]
\item When a user starts a debugging session,
\oursys\ middleware analyzes the query and decomposes it into a graph of query blocks in an internal representation.
This representation specifies how each block is further broken down into stages,
and contains metadata such as how to resolve external column references for correlated subqueries.
The client caches this representation.
\item When the user enters a debugging context (i.e., ``calling'' a query block),
the client sends the cached representation of this block back to the middleware along with any parameter settings for the call.
The middleware's \emph{debugging context initializer} computes the milestone tables for all input tables of the debugging context and all result tables of its constituent stages, by querying the database (discussed below).
The debugging context initializer also generates SQL query templates needed to support various debugging operations discussed earlier in this section.
The client caches these milestone tables and SQL templates.
\item When user carries out various operations in the debugging context,
the client consults the cached milestone tables to instantiate required SQL queries according to the cached templates.
It also caches the results of page-fetch queries,
so repeated accesses to the same page, which can happen frequently during debugging, do not need recomputation.
Cached pages can be evicted to make space for new ones. 
\end{itemize}
Overall, note that debugging sessions leave no state on the middleware or the database.
The middleware handles all complex SQL query analysis and rewrites, and the client simply needs to fill in the values of certain query parameters.
The only cached data whose size depends on the database are the milestone tables.
However, since there is only one milestone row per display page, and all columns are compact summaries of constant size,
the milestone tables are orders of magnitude smaller than the corresponding tables.
To make them even more scalable, milestones could be made hierarchical and refined on demand,
but as shown later in \Cref{sec:experiments}, the current single-level design already works well.

\paragraph{Additional Details on Query Rewriting}

As discussed earlier, all information needed for debugging can be computed by rewriting the original SQL query in some fashion.
\oursys\ performs most query rewriting in its debugging context initializer.
We have already shown how to generate queries for fetching pages (\Cref{sec:system:pagination}) and supporting other debugging operations (\Cref{appendix:system:positioning}).
As for computing milestones for each table in the debugging context, a single SQL query suffices, as illustrated by the following, which computes \Cref{tab:milestone} (including the Bloom filters) for our running example:
\begin{lstlisting}[language=SQL, mathescape=true, basicstyle=\scriptsize\ttfamily]
WITH tmp(seq, iid, s_bar, s_beer, f_drinker, f_bar) AS (
  SELECT ROW_NUMBER() OVER (ORDER BY s.bar,s.beer,f.drinker,f.bar) - 1, -- result row sequence number when sorted by IID
    ((s.bar, s.beer), (f.drinker, f.bar)), -- IID
    s.bar, s.beer, f.drinker, f.bar -- relevant for sargable & Bloom filtering
  FROM Serves s, Frequents f WHERE (...) -- same as original query
)
SELECT MIN_IID(iid), -- per-page minimum IID
  -- per-page range bounds for sargable fitering:
  ARRAY[MIN(s_bar), MAX(s_bar)], ARRAY[MIN(s_beer), MAX(s_beer)],
  ARRAY[MIN(f_drinker), MAX(f_drinker)], ARRAY[MIN(f_bar), MAX(f_bar)],
  -- per-page Bloom filter:
  BLOOM_GEN(p_partkey)
FROM tmp GROUP BY seq / $\mathit{page\_size}$ ORDER BY seq / $\mathit{page\_size}$;
-- page_size is the number of rows per page
\end{lstlisting}

After generating the SQL queries as discussed above and earlier, \oursys\ further applies several rewrite optimizations.
First, if a scalar subquery contains no external column references,
the debugging context initializer will simply precompute its result and replace the uses of this subqueries by its result.
Second, for any sargable condition injected into a query, we consider pushing it down further into a subquery.
Such cases often arise when, for example, we identify a range bound on some indexed column,
and this column (or another column equated to it by \WHERE) is referenced by a subquery as an external column;
here, we inject the range bound into the subquery as well.
Third, if the query's \FROM\ contains a subquery or table defined by \WITH,
we check whether the injected conditions imply an equality or range condition the IID of the input table.
If yes, we consult the milestones of the input table to construct sargable filters to further inject into the subquery defining the input table.
We call this last optimization \emph{recursive pushdown}.
Finally, pushdown through outer joins, which is especially tricky, is discussed in \Cref{appendix:system}.

\subsection{Pushdown through Outerjoins}

Consider a page-fetch query that \oursys\ generates for a full outer join debugging context,
which has the form:
\begin{lstlisting}[language=SQL, mathescape=true, basicstyle=\footnotesize\ttfamily]
SELECT ROW($R$.$K$, $S$.$K$) AS _iid
FROM $R$ FULL JOIN $S$ ON $R$.$A$ = $S.B$
WHERE $\Theta$($R$, $S$);
\end{lstlisting}
Here, $\Theta$ is condition based on the IID range of the requested page, which constrains $R$ and $S$.
Efficient execution of this query requires pushing down filters inferred from $\Theta$ to $R$ and $S$.
Unfortunately, such pushdowns are not always safe through outer joins.
For example, suppose that the requested page is the last one,
and its IID range implies the filter on $R$ to be $R.K$ \sql{IS} \sql{NULL} \sql{OR} \sql{$R.K$>=100}.
Further suppose that some $s \in S$ joins with a single row $r \in R$ with \sql{$R.K$<100}.
The above page-fetch query should not return any output for $s$,
since $\langle r, s \rangle$ does not belong to the requested page and hence fails the final \WHERE.
However, if we push the filter on $R$ to below the outer join,
the outer join will return a row containing $s$ and $R$ columns padded with \sql{NULL}s,
which passes the final \WHERE.

To avoid the above issue and still enable pushdown,
\oursys\ first computes the inner join between $R$ and $S$ with pushdown.
Then, only if the number of returned rows falls below the desired number on the requested page,
which can be determined from the milestone table and should happen rarely,
we issue additional queries to find rows in the outer join result but not in the inner.

%% file: sections/appendix/experiments.tex
\section{Additional Experimental Results}
\label{appendix:experiments}

We present the remaining experiment results over TPC-H benchmark. 

For each table in the TPC-H schema, we have the following indexes (all indexes are btree indexes in PostgreSQL):

\begin{itemize}[leftmargin=*]
    \item customer
          \begin{itemize}
              \item primary index: c\_custkey
              \item secondary indexes: None
          \end{itemize}
    \item lineitem
          \begin{itemize}
              \item primary index: (l\_orderkey, l\_linenumber)
              \item secondary indexes: l\_partkey, l\_suppkey, l\_shipdate
          \end{itemize}
    \item nation
          \begin{itemize}
              \item primary index: n\_nationkey
              \item secondary indexes: n\_name, n\_regionkey
          \end{itemize}
    \item orders
          \begin{itemize}
              \item primary index: o\_orderkey
              \item secondary indexes: o\_custkey, o\_orderdate
          \end{itemize}
    \item part
          \begin{itemize}
              \item primary index: p\_partkey
              \item secondary indexes: None
          \end{itemize}
    \item partsupp
          \begin{itemize}
              \item primary index: (ps\_partkey, ps\_suppkey)
              \item secondary indexes: ps\_suppkey
          \end{itemize}
    \item region
          \begin{itemize}
              \item primary index: r\_regionkey
              \item secondary indexes: None
          \end{itemize}
    \item supplier
          \begin{itemize}
              \item primary index: s\_suppkey
              \item secondary indexes: s\_name, s\_phone
          \end{itemize}
\end{itemize}

We ran experiments for three different page size: 50, 100 and 200 for all tables (stages) in all TPC-H queries. For each table, we prepared three queries: milestone query, page query and table query and thus collecting the following data over all testing instances (i.e., 1GB, 5GB and 10GB):
\begin{itemize}
    \item The execution time and output size of the milestone query and the table query.
    \item The execution time of the page query and baseline query (by rewritting the table query with \OFFSET\ and \LIMIT) for retrieving the first page (``head''), middle page (``mid'') and the second last (``tail'') page. 
\end{itemize}
Since each query potentially contains multiple query blocks and subsequently multiple tables, we present the statistics for the largest table (measured in MB) to compute for each query. For page size 50, 100 and 200, the experiment results are shown in \Cref{tab:general-result-pg-50}, \Cref{tab:general-result-pg-100} and \Cref{tab:general-result-pg-200} respectively. 

In summary, we make the following conclusions:
\begin{itemize}[leftmargin=*]
    \item The milestone queries can run slower than the table queries, but with acceptable delays since they are run at the beginning of the debugging session without affecting debugging operations later. In some cases, the milestone queries run faster than the table queries. On the other hand, the output sizes of the milestone queries are almost always smaller than those of the table queries by a rough factor of the page size.
    \item The optimization for page query almost always outperforms the baseline, and differences between the execution time grow as the database size grows, especially for ``mid'' and ``tail'' pages. The optimizations are sensitive to page size but insensitive to the database size. There are only few cases where the optimizations ``over-hint'' the PosgreSQL optimizer and cause the execution time to be roughly the same as the baseline.
\end{itemize}

\begin{table*}[ht]
\vspace{-4mm}
    \centering
    \scriptsize
    \begin{subtable}[b]{0.55\textwidth}\centering
        \begin{tabular}{|c|c|c|c|c|c|c|c|}
             \cline{3-8}
             \multicolumn{2}{c|}{} & \multicolumn{2}{c|}{1GB} & \multicolumn{2}{c|}{5GB} & \multicolumn{2}{c|}{10GB} \\ \hline
             Query & Page & Opt. & Base & Opt. & Base & Opt. & Base  \\ \hline
             \multirow{3}{*}{ Q1 } & head & {\bf 0.053 } & 0.067 & {\bf 0.058 } & 0.106 & 0.071 & {\bf 0.059 } \\
                                    & mid & {\bf 0.038 } & 276.977 & {\bf 0.123 } & 10719.953 & {\bf 0.168 } & 21826.43 \\
                                    & tail & {\bf 0.034 } & 553.887 & {\bf 0.033 } & 21439.8 & {\bf 0.052 } & 43652.801 \\ \hline
             \multirow{3}{*}{ Q2 } & head & {\bf 5.758 } & 251.032 & {\bf 6.356 } & 1707.231 & {\bf 6.954 } & 3291.403 \\
                                    & mid & {\bf 7.236 } & 253.723 & {\bf 5.953 } & 1716.116 & {\bf 8.916 } & 3571.981 \\
                                    & tail & {\bf 6.596 } & 256.414 & {\bf 6.832 } & 1712.497 & {\bf 6.503 } & 3852.56 \\ \hline
             \multirow{3}{*}{ Q3 } & head & 12.57 & {\bf 4.225 } & {\bf 81.789 } & 407.539 & {\bf 59.029 } & 590.042 \\
                                    & mid & {\bf 6.165 } & 318.842 & {\bf 145.043 } & 4937.49 & {\bf 58.346 } & 22685.161 \\
                                    & tail & {\bf 7.179 } & 633.459 & {\bf 143.534 } & 9467.44 & {\bf 34.891 } & 44780.279 \\ \hline
             \multirow{3}{*}{ Q4 } & head & 0.445 & {\bf 0.364 } & 0.405 & {\bf 0.362 } & 3.39 & {\bf 2.122 } \\
                                    & mid & {\bf 0.445 } & 108.699 & {\bf 1.318 } & 554.128 & {\bf 5.365 } & 8912.874 \\
                                    & tail & {\bf 0.383 } & 217.034 & {\bf 2.006 } & 1107.893 & {\bf 3.216 } & 17823.626 \\ \hline
             \multirow{3}{*}{ Q5 } & head & 8.909 & {\bf 8.657 } & {\bf 9.142 } & 13.343 & 420.739 & {\bf 314.969 } \\
                                    & mid & {\bf 7.653 } & 1287.24 & {\bf 23.182 } & 10036.965 & {\bf 193.466 } & 28344.422 \\
                                    & tail & {\bf 5.98 } & 2508.011 & {\bf 7.776 } & 20060.586 & {\bf 193.096 } & 56373.876 \\ \hline
             \multirow{3}{*}{ Q6 } & head & 0.189 & {\bf 0.156 } & 0.175 & {\bf 0.154 } & 0.405 & {\bf 0.156 } \\
                                    & mid & {\bf 0.213 } & 216.07 & {\bf 0.202 } & 1168.038 & {\bf 3.059 } & 16448.599 \\
                                    & tail & {\bf 0.147 } & 431.984 & {\bf 0.197 } & 2335.922 & {\bf 1.594 } & 32897.042 \\ \hline
             \multirow{3}{*}{ Q7 } & head & 20.642 & {\bf 14.371 } & 68.84 & {\bf 14.918 } & 985.063 & {\bf 71.079 } \\
                                    & mid & {\bf 12.155 } & 395.358 & {\bf 41.39 } & 3065.17 & {\bf 851.135 } & 15909.36 \\
                                    & tail & {\bf 14.622 } & 776.345 & {\bf 41.266 } & 6105.881 & {\bf 907.042 } & 27296.999 \\ \hline
             \multirow{3}{*}{ Q8 } & head & {\bf 4.148 } & 4.407 & {\bf 4.7 } & 35.101 & 317.395 & {\bf 45.033 } \\
                                    & mid & {\bf 4.32 } & 1656.216 & {\bf 28.169 } & 6557.276 & {\bf 238.396 } & 20430.873 \\
                                    & tail & {\bf 3.984 } & 2970.844 & {\bf 4.587 } & 13079.451 & {\bf 245.715 } & 40816.712 \\ \hline
             \multirow{3}{*}{ Q9 } & head & {\bf 0.853 } & 242.203 & {\bf 143.601 } & 12191.21 & {\bf 51.197 } & 85054.771 \\
                                    & mid & {\bf 1.245 } & 1008.268 & {\bf 199.13 } & 14899.526 & {\bf 38.146 } & 69401.355 \\
                                    & tail & {\bf 0.765 } & 1774.333 & {\bf 131.757 } & 17607.842 & {\bf 35.69 } & 54535.253 \\ \hline
             \multirow{3}{*}{ Q10 } & head & 0.939 & {\bf 0.748 } & 62.441 & {\bf 1.509 } & 19.73 & {\bf 0.912 } \\
                                    & mid & {\bf 0.835 } & 800.388 & {\bf 28.382 } & 15660.853 & {\bf 76.849 } & 20184.374 \\
                                    & tail & {\bf 0.772 } & 1600.028 & {\bf 83.471 } & 18930.806 & {\bf 8.615 } & 40367.836 \\ \hline
             \multirow{3}{*}{ Q11 } & head & 23.582 & {\bf 1.898 } & 134.269 & {\bf 6.779 } & 262.009 & {\bf 7.036 } \\
                                    & mid & {\bf 22.03 } & 32.884 & {\bf 169.394 } & 221.772 & {\bf 244.886 } & 393.363 \\
                                    & tail & {\bf 17.811 } & 63.87 & {\bf 130.579 } & 436.765 & {\bf 233.001 } & 779.69 \\ \hline
             \multirow{3}{*}{ Q12 } & head & 0.858 & {\bf 0.615 } & 9.142 & {\bf 0.569 } & {\bf 1.741 } & 2.327 \\
                                    & mid & {\bf 0.808 } & 285.367 & {\bf 0.558 } & 7589.805 & {\bf 0.711 } & 17582.209 \\
                                    & tail & {\bf 0.866 } & 570.118 & {\bf 0.619 } & 12918.587 & {\bf 0.549 } & 35162.092 \\ \hline
             \multirow{3}{*}{ Q13 } & head & 0.244 & {\bf 0.197 } & 0.184 & {\bf 0.172 } & 0.193 & {\bf 0.186 } \\
                                    & mid & {\bf 0.195 } & 470.197 & {\bf 0.191 } & 2562.524 & {\bf 0.182 } & 5584.982 \\
                                    & tail & {\bf 0.172 } & 940.197 & {\bf 0.163 } & 5124.876 & {\bf 0.186 } & 11169.778 \\ \hline
             \multirow{3}{*}{ Q14 } & head & 4.875 & {\bf 1.2 } & 33.992 & {\bf 3.162 } & 7.135 & {\bf 1.846 } \\
                                    & mid & {\bf 4.455 } & 116.905 & {\bf 32.801 } & 717.536 & {\bf 7.643 } & 1573.678 \\
                                    & tail & {\bf 4.392 } & 232.611 & {\bf 30.97 } & 1431.91 & {\bf 5.695 } & 3145.51 \\ \hline
             \multirow{3}{*}{ Q15 } & head & 0.238 & {\bf 0.221 } & 0.247 & {\bf 0.234 } & 1.315 & {\bf 0.276 } \\
                                    & mid & {\bf 0.241 } & 156.641 & {\bf 0.182 } & 922.131 & {\bf 12.682 } & 1779.888 \\
                                    & tail & {\bf 0.163 } & 313.061 & {\bf 6.841 } & 1844.028 & {\bf 2.764 } & 3559.5 \\ \hline
             \multirow{3}{*}{ Q16 } & head & 1.682 & {\bf 1.39 } & 6.656 & {\bf 6.074 } & 12.276 & {\bf 11.936 } \\
                                    & mid & {\bf 1.438 } & 140.749 & {\bf 6.191 } & 844.365 & {\bf 13.452 } & 1700.343 \\
                                    & tail & {\bf 1.486 } & 280.107 & {\bf 5.777 } & 1682.656 & {\bf 12.02 } & 3388.75 \\ \hline
             \multirow{3}{*}{ Q17 } & head & 2.538 & {\bf 2.236 } & 15.809 & {\bf 2.056 } & {\bf 20.733 } & 125.812 \\
                                    & mid & {\bf 2.286 } & 1536.001 & {\bf 5.141 } & 9945.394 & {\bf 5.923 } & 25862.823 \\
                                    & tail & {\bf 2.409 } & 3069.766 & {\bf 4.89 } & 19888.732 & {\bf 4.945 } & 51599.835 \\ \hline
             \multirow{3}{*}{ Q18 } & head & {\bf 2532.81 } & 5870.729 & {\bf 14184.028 } & 46234.889 & {\bf 37419.067 } & 120576.875 \\
                                    & mid & {\bf 2456.265 } & 7991.925 & {\bf 16215.405 } & 55316.965 & {\bf 32434.602 } & 133646.155 \\
                                    & tail & {\bf 2787.514 } & 10113.122 & {\bf 13983.933 } & 64399.041 & {\bf 33393.882 } & 146715.434 \\ \hline
             \multirow{3}{*}{ Q19 } & head & {\bf 25.461 } & 34.29 & 178.18 & {\bf 51.645 } & 1011.015 & {\bf 60.365 } \\
                                    & mid & {\bf 25.324 } & 76.065 & {\bf 129.633 } & 339.639 & {\bf 138.908 } & 589.525 \\
                                    & tail & {\bf 24.978 } & 117.84 & {\bf 117.056 } & 627.633 & {\bf 59.341 } & 1118.685 \\ \hline
             \multirow{3}{*}{ Q20 } & head & 0.094 & {\bf 0.053 } & 0.086 & {\bf 0.047 } & {\bf 0.184 } & 0.659 \\
                                    & mid & {\bf 0.077 } & 49.901 & {\bf 1.133 } & 661.308 & {\bf 2.122 } & 1564.677 \\
                                    & tail & {\bf 0.086 } & 99.748 & {\bf 30.041 } & 1322.57 & {\bf 0.842 } & 3128.694 \\ \hline
             \multirow{3}{*}{ Q21 } & head & {\bf 16.769 } & 823.461 & {\bf 189.848 } & 15873.412 & {\bf 1039.691 } & 46780.254 \\
                                    & mid & {\bf 11.981 } & 830.494 & {\bf 232.896 } & 15499.488 & {\bf 909.955 } & 47334.357 \\
                                    & tail & {\bf 13.678 } & 837.527 & {\bf 284.158 } & 15125.565 & {\bf 1099.474 } & 46780.254 \\ \hline
             \multirow{3}{*}{ Q22 } & head & {\bf 0.608 } & 31.927 & {\bf 0.838 } & 159.24 & {\bf 1.707 } & 318.704 \\
                                    & mid & {\bf 0.673 } & 52.952 & {\bf 0.82 } & 261.309 & {\bf 2.324 } & 529.83 \\
                                    & tail & {\bf 0.512 } & 73.61 & {\bf 0.869 } & 363.378 & {\bf 1.838 } & 740.955 \\ \hline
        \end{tabular}
        \caption{Optimization vs. Baseline for Page Query Execution Time}
        \label{tab:optimization-vs-baseline-pg-50}
    \end{subtable}
    \quad
    \begin{subtable}[b]{0.42\textwidth}\centering
        \begin{tabular}{|c|c|c|c|c|c|}
             \cline{3-6}
             \multicolumn{2}{c}{} & \multicolumn{2}{|c|}{Time (ms)} & \multicolumn{2}{c|}{Output Size (MB)} \\ \hline
             Query & DB size & Milestone & Table & Milestone & Table \\ \hline
             \multirow{3}{*}{ Q1 } & 1 & 8331.914 & {\bf 1388.336 } & {\bf 31.159 } & 761.653 \\
                                    & 5 & 44032.912 & {\bf 23151.685 } & {\bf 155.853 } & 3809.747 \\
                                    & 10 & 87464.164 & {\bf 46420.358 } & {\bf 311.667 } & 7618.516 \\ \hline
             \multirow{3}{*}{ Q2 } & 1 & {\bf 215.464 } & 329.321 & {\bf 0.006 } & 0.129 \\
                                    & 5 & {\bf 1356.54 } & 2386.544 & {\bf 0.031 } & 0.647 \\
                                    & 10 & {\bf 2698.862 } & 4956.342 & {\bf 0.062 } & 1.31 \\ \hline
             \multirow{3}{*}{ Q3 } & 1 & {\bf 1044.567 } & 1645.413 & {\bf 0.216 } & 9.358 \\
                                    & 5 & {\bf 5933.313 } & 22270.548 & {\bf 1.056 } & 45.735 \\
                                    & 10 & {\bf 13848.889 } & 43462.104 & {\bf 2.145 } & 92.932 \\ \hline
             \multirow{3}{*}{ Q4 } & 1 & 377.848 & {\bf 347.824 } & {\bf 0.382 } & 6.367 \\
                                    & 5 & 1640.64 & {\bf 1507.237 } & {\bf 1.899 } & 31.649 \\
                                    & 10 & {\bf 11874.966 } & 25455.856 & {\bf 3.79 } & 63.165 \\ \hline
             \multirow{3}{*}{ Q5 } & 1 & {\bf 2142.641 } & 2747.291 & {\bf 0.936 } & 30.979 \\
                                    & 5 & {\bf 11543.048 } & 20464.069 & {\bf 4.675 } & 154.836 \\
                                    & 10 & {\bf 27796.453 } & 58846.691 & {\bf 9.325 } & 308.884 \\ \hline
             \multirow{3}{*}{ Q6 } & 1 & 763.76 & {\bf 595.106 } & {\bf 1.795 } & 43.884 \\
                                    & 5 & 3899.765 & {\bf 2675.542 } & {\bf 8.951 } & 218.804 \\
                                    & 10 & {\bf 9705.647 } & 34310.126 & {\bf 17.923 } & 438.125 \\ \hline
             \multirow{3}{*}{ Q7 } & 1 & {\bf 810.777 } & 987.756 & {\bf 0.152 } & 0.947 \\
                                    & 5 & {\bf 6335.476 } & 6637.603 & {\bf 0.737 } & 4.602 \\
                                    & 10 & {\bf 14646.955 } & 35110.268 & {\bf 1.489 } & 9.306 \\ \hline
             \multirow{3}{*}{ Q8 } & 1 & 4960.836 & {\bf 3186.478 } & {\bf 11.67 } & 65.639 \\
                                    & 5 & 23772.766 & {\bf 17879.496 } & {\bf 58.054 } & 326.55 \\
                                    & 10 & 50267.721 & {\bf 41981.095 } & {\bf 116.284 } & 654.091 \\ \hline
             \multirow{3}{*}{ Q9 } & 1 & {\bf 2227.58 } & 2517.948 & {\bf 4.267 } & 23.998 \\
                                    & 5 & {\bf 17141.161 } & 19767.02 & {\bf 20.989 } & 118.061 \\
                                    & 10 & {\bf 31764.758 } & 53459.127 & {\bf 42.012 } & 236.312 \\ \hline
             \multirow{3}{*}{ Q10 } & 1 & {\bf 1583.518 } & 1695.667 & {\bf 2.897 } & 77.861 \\
                                    & 5 & 9485.609 & {\bf 9060.029 } & {\bf 14.538 } & 390.698 \\
                                    & 10 & {\bf 20332.249 } & 49173.573 & {\bf 29.047 } & 780.631 \\ \hline
             \multirow{3}{*}{ Q11 } & 1 & 126.067 & {\bf 82.778 } & {\bf 0.243 } & 5.928 \\
                                    & 5 & {\bf 538.645 } & 590.908 & {\bf 1.152 } & 28.146 \\
                                    & 10 & {\bf 985.077 } & 18047.688 & {\bf 2.299 } & 56.193 \\ \hline
             \multirow{3}{*}{ Q12 } & 1 & {\bf 1108.119 } & 1502.238 & {\bf 0.788 } & 10.507 \\
                                    & 5 & {\bf 5109.859 } & 16516.219 & {\bf 3.947 } & 52.627 \\
                                    & 10 & {\bf 10723.043 } & 43430.112 & {\bf 7.879 } & 105.057 \\ \hline
             \multirow{3}{*}{ Q13 } & 1 & 4453.644 & {\bf 1820.081 } & {\bf 7.192 } & 74.4 \\
                                    & 5 & 22620.462 & {\bf 10634.663 } & {\bf 35.96 } & 372.001 \\
                                    & 10 & 46633.065 & {\bf 22689.628 } & {\bf 71.92 } & 744.001 \\ \hline
             \multirow{3}{*}{ Q14 } & 1 & {\bf 267.11 } & 344.042 & {\bf 0.539 } & 18.571 \\
                                    & 5 & {\bf 1537.443 } & 5938.491 & {\bf 2.694 } & 92.793 \\
                                    & 10 & {\bf 3053.81 } & 43108.788 & {\bf 5.392 } & 185.705 \\ \hline
             \multirow{3}{*}{ Q15 } & 1 & 697.134 & {\bf 440.953 } & {\bf 2.165 } & 52.914 \\
                                    & 5 & {\bf 3615.734 } & 6797.217 & {\bf 10.845 } & 265.099 \\
                                    & 10 & {\bf 7422.926 } & 40451.245 & {\bf 21.696 } & 530.346 \\ \hline
             \multirow{3}{*}{ Q16 } & 1 & 382.706 & {\bf 300.081 } & {\bf 0.741 } & 16.455 \\
                                    & 5 & {\bf 1855.771 } & 4245.039 & {\bf 3.678 } & 81.74 \\
                                    & 10 & {\bf 3654.478 } & 4303.77 & {\bf 7.378 } & 163.946 \\ \hline
             \multirow{3}{*}{ Q17 } & 1 & 6962.627 & {\bf 6846.472 } & {\bf 1.44 } & 11.2 \\
                                    & 5 & {\bf 18483.761 } & 22136.792 & {\bf 7.2 } & 56.0 \\
                                    & 10 & {\bf 38778.384 } & 59549.219 & {\bf 14.4 } & 112.0 \\ \hline
             \multirow{3}{*}{ Q18 } & 1 & 15984.578 & {\bf 12086.886 } & {\bf 52.215 } & 1272.737 \\
                                    & 5 & 95566.664 & {\bf 87403.308 } & {\bf 260.86 } & 6358.46 \\
                                    & 10 & {\bf 174787.04 } & 181889.789 & {\bf 521.399 } & 12709.096 \\ \hline
             \multirow{3}{*}{ Q19 } & 1 & {\bf 145.722 } & 244.134 & {\bf 0.044 } & 1.501 \\
                                    & 5 & {\bf 770.808 } & 15055.31 & {\bf 0.217 } & 7.47 \\
                                    & 10 & {\bf 1270.167 } & 68809.099 & {\bf 0.421 } & 14.506 \\ \hline
             \multirow{3}{*}{ Q20 } & 1 & 2114.795 & {\bf 1503.764 } & {\bf 6.577 } & 1056.214 \\
                                    & 5 & {\bf 10328.037 } & 25180.76 & {\bf 32.885 } & 5279.964 \\
                                    & 10 & {\bf 22092.932 } & 60763.015 & {\bf 65.691 } & 10557.545 \\ \hline
             \multirow{3}{*}{ Q21 } & 1 & {\bf 869.86 } & 1039.671 & {\bf 0.051 } & 1.332 \\
                                    & 5 & {\bf 10193.128 } & 37567.462 & {\bf 0.25 } & 6.538 \\
                                    & 10 & {\bf 25156.809 } & 102244.18 & {\bf 0.511 } & 13.417 \\ \hline
             \multirow{3}{*}{ Q22 } & 1 & 120.044 & {\bf 77.932 } & {\bf 0.046 } & 0.405 \\
                                    & 5 & {\bf 472.207 } & 497.835 & {\bf 0.229 } & 2.035 \\
                                    & 10 & {\bf 897.368 } & 2021.357 & {\bf 0.458 } & 4.069 \\ \hline
        \end{tabular}
        \caption{Milestone Query vs. Table Query}
        \label{tab:milestone-vs-table-pg-50}
    \end{subtable}
    
    \caption{Experiment results for Page Size 50. Shorter time and smaller sizes are boldfaced.}
    \label{tab:general-result-pg-50}
\end{table*}

\clearpage
\begin{table*}[ht]
\vspace{-4mm}
    \centering
    \scriptsize
    \begin{subtable}[b]{0.55\textwidth}\centering
        \begin{tabular}{|c|c|c|c|c|c|c|c|}
             \cline{3-8}
             \multicolumn{2}{c|}{} & \multicolumn{2}{c|}{1GB} & \multicolumn{2}{c|}{5GB} & \multicolumn{2}{c|}{10GB} \\ \hline
             Query & Page & Opt. & Base & Opt. & Base & Opt. & Base  \\ \hline
             \multirow{3}{*}{ Q1 } & head & {\bf 0.071 } & 0.073 & 0.079 & {\bf 0.069 } & 0.084 & {\bf 0.069 } \\
                                    & mid & {\bf 0.055 } & 296.325 & {\bf 0.069 } & 9693.89 & {\bf 2.16 } & 22897.01 \\
                                    & tail & {\bf 0.05 } & 592.576 & {\bf 0.055 } & 19387.711 & {\bf 0.094 } & 45793.951 \\ \hline
             \multirow{3}{*}{ Q2 } & head & {\bf 13.061 } & 295.438 & {\bf 15.513 } & 1734.351 & {\bf 14.555 } & 3311.46 \\
                                    & mid & {\bf 12.952 } & 275.061 & {\bf 11.556 } & 1751.559 & {\bf 14.029 } & 3591.94 \\
                                    & tail & {\bf 11.611 } & 254.684 & {\bf 17.868 } & 1738.402 & {\bf 13.057 } & 3872.421 \\ \hline
             \multirow{3}{*}{ Q3 } & head & {\bf 2.198 } & 7.366 & {\bf 62.352 } & 63.668 & {\bf 40.817 } & 485.093 \\
                                    & mid & {\bf 1.76 } & 323.121 & {\bf 76.662 } & 4827.628 & {\bf 47.888 } & 20987.232 \\
                                    & tail & {\bf 1.762 } & 638.876 & {\bf 24.29 } & 9591.589 & {\bf 35.286 } & 41449.976 \\ \hline
             \multirow{3}{*}{ Q4 } & head & 0.643 & {\bf 0.631 } & 0.656 & {\bf 0.652 } & 31.234 & {\bf 3.163 } \\
                                    & mid & {\bf 0.683 } & 110.39 & {\bf 0.753 } & 573.421 & {\bf 8.977 } & 9141.04 \\
                                    & tail & {\bf 0.665 } & 220.149 & {\bf 12.071 } & 1146.19 & {\bf 4.969 } & 18278.917 \\ \hline
             \multirow{3}{*}{ Q5 } & head & 14.197 & {\bf 13.368 } & 45.254 & {\bf 17.582 } & 683.39 & {\bf 473.463 } \\
                                    & mid & {\bf 13.586 } & 1244.566 & {\bf 43.844 } & 10020.093 & {\bf 343.844 } & 29541.882 \\
                                    & tail & {\bf 10.241 } & 2475.765 & {\bf 14.55 } & 19971.901 & {\bf 300.573 } & 57266.427 \\ \hline
             \multirow{3}{*}{ Q6 } & head & 0.385 & {\bf 0.32 } & 0.399 & {\bf 0.333 } & 1.518 & {\bf 0.319 } \\
                                    & mid & {\bf 0.363 } & 218.425 & {\bf 0.349 } & 1196.178 & {\bf 2.845 } & 16829.839 \\
                                    & tail & {\bf 0.294 } & 435.856 & {\bf 0.384 } & 2392.023 & {\bf 4.812 } & 33659.359 \\ \hline
             \multirow{3}{*}{ Q7 } & head & {\bf 26.172 } & 26.356 & 96.368 & {\bf 34.798 } & 1772.658 & {\bf 46.352 } \\
                                    & mid & {\bf 20.338 } & 387.671 & {\bf 36.935 } & 3093.804 & {\bf 1073.085 } & 16393.14 \\
                                    & tail & {\bf 22.026 } & 748.985 & {\bf 32.845 } & 6078.97 & {\bf 948.516 } & 32739.927 \\ \hline
             \multirow{3}{*}{ Q8 } & head & 9.067 & {\bf 8.502 } & {\bf 8.987 } & 48.048 & 509.672 & {\bf 34.627 } \\
                                    & mid & {\bf 6.755 } & 1695.254 & {\bf 29.111 } & 6617.096 & {\bf 422.185 } & 20162.913 \\
                                    & tail & {\bf 7.429 } & 2995.663 & {\bf 10.025 } & 13186.143 & {\bf 364.929 } & 40291.199 \\ \hline
             \multirow{3}{*}{ Q9 } & head & {\bf 1.477 } & 231.532 & {\bf 67.733 } & 8417.397 & {\bf 94.242 } & 74380.749 \\
                                    & mid & {\bf 1.375 } & 962.979 & {\bf 77.736 } & 12923.744 & {\bf 61.954 } & 64993.676 \\
                                    & tail & {\bf 1.867 } & 1694.426 & {\bf 53.516 } & 17430.09 & {\bf 66.841 } & 58048.861 \\ \hline
             \multirow{3}{*}{ Q10 } & head & 1.46 & {\bf 1.245 } & 21.716 & {\bf 2.257 } & 26.641 & {\bf 2.529 } \\
                                    & mid & {\bf 1.201 } & 797.21 & {\bf 15.774 } & 8107.626 & {\bf 22.302 } & 21555.761 \\
                                    & tail & {\bf 1.364 } & 1593.175 & {\bf 21.038 } & 13905.496 & {\bf 13.519 } & 43108.993 \\ \hline
             \multirow{3}{*}{ Q11 } & head & 36.541 & {\bf 3.3 } & 248.452 & {\bf 8.555 } & 448.46 & {\bf 12.203 } \\
                                    & mid & 38.196 & {\bf 36.848 } & {\bf 219.147 } & 222.082 & 513.017 & {\bf 390.707 } \\
                                    & tail & {\bf 38.557 } & 70.396 & {\bf 202.447 } & 435.609 & {\bf 397.68 } & 769.212 \\ \hline
             \multirow{3}{*}{ Q12 } & head & 1.105 & {\bf 1.08 } & 11.862 & {\bf 1.222 } & 2.863 & {\bf 1.216 } \\
                                    & mid & {\bf 0.976 } & 296.293 & {\bf 1.204 } & 7060.195 & {\bf 1.321 } & 17357.688 \\
                                    & tail & {\bf 1.126 } & 591.506 & {\bf 1.017 } & 12045.238 & {\bf 0.893 } & 34714.16 \\ \hline
             \multirow{3}{*}{ Q13 } & head & 0.347 & {\bf 0.346 } & {\bf 0.305 } & 0.309 & 0.341 & {\bf 0.331 } \\
                                    & mid & {\bf 0.331 } & 458.786 & {\bf 0.291 } & 2603.571 & {\bf 0.285 } & 5598.169 \\
                                    & tail & {\bf 0.331 } & 917.225 & {\bf 0.26 } & 5206.832 & {\bf 0.308 } & 11196.008 \\ \hline
             \multirow{3}{*}{ Q14 } & head & 5.512 & {\bf 2.71 } & 35.8 & {\bf 5.729 } & 78.23 & {\bf 14.473 } \\
                                    & mid & {\bf 4.76 } & 118.153 & {\bf 35.366 } & 718.17 & {\bf 76.052 } & 1600.618 \\
                                    & tail & {\bf 4.772 } & 233.595 & {\bf 35.272 } & 1430.611 & {\bf 79.114 } & 3186.763 \\ \hline
             \multirow{3}{*}{ Q15 } & head & 0.353 & {\bf 0.297 } & {\bf 0.328 } & 0.329 & 4.253 & {\bf 0.35 } \\
                                    & mid & {\bf 0.371 } & 157.853 & {\bf 2.625 } & 930.38 & {\bf 12.812 } & 1785.893 \\
                                    & tail & {\bf 0.392 } & 315.409 & {\bf 0.707 } & 1860.43 & {\bf 4.657 } & 3571.436 \\ \hline
             \multirow{3}{*}{ Q16 } & head & 2.137 & {\bf 1.756 } & 6.243 & {\bf 5.975 } & {\bf 12.268 } & 14.385 \\
                                    & mid & {\bf 1.931 } & 143.953 & {\bf 6.294 } & 854.68 & {\bf 14.295 } & 1753.734 \\
                                    & tail & {\bf 1.773 } & 286.15 & {\bf 6.041 } & 1703.385 & {\bf 13.032 } & 3493.083 \\ \hline
             \multirow{3}{*}{ Q17 } & head & 4.729 & {\bf 3.941 } & 9.248 & {\bf 4.222 } & 32.706 & {\bf 3.904 } \\
                                    & mid & {\bf 4.541 } & 1527.395 & {\bf 11.298 } & 9669.438 & {\bf 13.855 } & 27819.374 \\
                                    & tail & {\bf 4.505 } & 3050.85 & {\bf 12.892 } & 19334.653 & {\bf 20.307 } & 55634.844 \\ \hline
             \multirow{3}{*}{ Q18 } & head & {\bf 2527.82 } & 5984.869 & {\bf 15935.661 } & 45765.329 & {\bf 34808.126 } & 126117.683 \\
                                    & mid & {\bf 2618.646 } & 8093.631 & {\bf 16087.319 } & 55018.855 & {\bf 34021.36 } & 139702.667 \\
                                    & tail & {\bf 2663.236 } & 10202.394 & {\bf 17309.887 } & 64272.382 & {\bf 37965.734 } & 153287.651 \\ \hline
             \multirow{3}{*}{ Q19 } & head & {\bf 43.729 } & 68.104 & 197.71 & {\bf 105.373 } & {\bf 261.355 } & 1188.735 \\
                                    & mid & {\bf 42.12 } & 94.035 & {\bf 144.187 } & 370.888 & {\bf 375.334 } & 1146.398 \\
                                    & tail & {\bf 44.463 } & 119.965 & {\bf 151.102 } & 630.924 & {\bf 209.468 } & 1136.618 \\ \hline
             \multirow{3}{*}{ Q20 } & head & 0.158 & {\bf 0.095 } & 0.173 & {\bf 0.084 } & {\bf 0.291 } & 1.433 \\
                                    & mid & {\bf 0.178 } & 50.034 & {\bf 0.116 } & 352.756 & {\bf 2.826 } & 1696.115 \\
                                    & tail & {\bf 0.148 } & 99.973 & {\bf 1.263 } & 705.428 & {\bf 1.025 } & 3390.798 \\ \hline
             \multirow{3}{*}{ Q21 } & head & {\bf 31.819 } & 874.487 & {\bf 2616.526 } & 13686.317 & {\bf 1375.21 } & 52290.601 \\
                                    & mid & {\bf 21.822 } & 852.312 & {\bf 1529.724 } & 13691.07 & {\bf 2528.902 } & 51884.652 \\
                                    & tail & {\bf 22.996 } & 835.458 & {\bf 668.835 } & 13686.317 & {\bf 2307.124 } & 51478.704 \\ \hline
             \multirow{3}{*}{ Q22 } & head & {\bf 0.994 } & 33.647 & {\bf 0.967 } & 159.614 & {\bf 4.113 } & 322.177 \\
                                    & mid & {\bf 1.021 } & 52.94 & {\bf 1.392 } & 267.151 & {\bf 2.617 } & 534.946 \\
                                    & tail & {\bf 1.265 } & 72.233 & {\bf 1.525 } & 374.689 & {\bf 2.231 } & 747.714 \\ \hline
        \end{tabular}
        \caption{Optimization vs. Baseline for Page Query Execution Time}
        \label{tab:optimization-vs-baseline-pg-100}
    \end{subtable}
    \quad
    \begin{subtable}[b]{0.42\textwidth}\centering
        \begin{tabular}{|c|c|c|c|c|c|}
             \cline{3-6}
             \multicolumn{2}{c}{} & \multicolumn{2}{|c|}{Time (ms)} & \multicolumn{2}{c|}{Output Size (MB)} \\ \hline
             Query & DB size & Milestone & Table & Milestone & Table \\ \hline
             \multirow{3}{*}{ Q1 } & 1 & 8308.278 & {\bf 1350.574 } & {\bf 15.579 } & 761.653 \\
                                    & 5 & 43646.179 & {\bf 21979.09 } & {\bf 77.927 } & 3809.747 \\
                                    & 10 & 87349.051 & {\bf 47552.632 } & {\bf 155.834 } & 7618.516 \\ \hline
             \multirow{3}{*}{ Q2 } & 1 & {\bf 204.105 } & 303.443 & {\bf 0.003 } & 0.129 \\
                                    & 5 & {\bf 1419.539 } & 2199.755 & {\bf 0.015 } & 0.647 \\
                                    & 10 & {\bf 2718.322 } & 5082.711 & {\bf 0.031 } & 1.31 \\ \hline
             \multirow{3}{*}{ Q3 } & 1 & {\bf 1049.375 } & 1435.624 & {\bf 0.108 } & 9.358 \\
                                    & 5 & {\bf 6184.207 } & 15819.558 & {\bf 0.528 } & 45.735 \\
                                    & 10 & {\bf 14013.955 } & 42987.942 & {\bf 1.072 } & 92.932 \\ \hline
             \multirow{3}{*}{ Q4 } & 1 & {\bf 322.772 } & 368.223 & {\bf 0.191 } & 6.367 \\
                                    & 5 & 1596.769 & {\bf 1593.545 } & {\bf 0.95 } & 31.649 \\
                                    & 10 & {\bf 11752.173 } & 25611.759 & {\bf 1.895 } & 63.165 \\ \hline
             \multirow{3}{*}{ Q5 } & 1 & {\bf 2148.69 } & 2753.883 & {\bf 0.468 } & 30.979 \\
                                    & 5 & {\bf 11514.031 } & 20612.307 & {\bf 2.337 } & 154.836 \\
                                    & 10 & {\bf 27627.926 } & 58496.928 & {\bf 4.662 } & 308.884 \\ \hline
             \multirow{3}{*}{ Q6 } & 1 & 789.488 & {\bf 570.995 } & {\bf 0.898 } & 43.884 \\
                                    & 5 & 3902.241 & {\bf 2784.104 } & {\bf 4.476 } & 218.804 \\
                                    & 10 & {\bf 9796.535 } & 33640.745 & {\bf 8.962 } & 438.125 \\ \hline
             \multirow{3}{*}{ Q7 } & 1 & {\bf 826.424 } & 968.716 & {\bf 0.076 } & 0.947 \\
                                    & 5 & {\bf 6563.466 } & 6701.097 & {\bf 0.369 } & 4.602 \\
                                    & 10 & {\bf 14712.534 } & 31113.968 & {\bf 0.745 } & 9.306 \\ \hline
             \multirow{3}{*}{ Q8 } & 1 & 5090.232 & {\bf 3210.232 } & {\bf 5.835 } & 65.639 \\
                                    & 5 & 23697.023 & {\bf 17788.411 } & {\bf 29.027 } & 326.55 \\
                                    & 10 & 50386.193 & {\bf 43219.432 } & {\bf 58.142 } & 654.091 \\ \hline
             \multirow{3}{*}{ Q9 } & 1 & {\bf 2236.322 } & 2488.499 & {\bf 2.134 } & 23.998 \\
                                    & 5 & {\bf 17082.537 } & 19866.147 & {\bf 10.495 } & 118.061 \\
                                    & 10 & {\bf 32766.722 } & 53930.664 & {\bf 21.006 } & 236.312 \\ \hline
             \multirow{3}{*}{ Q10 } & 1 & {\bf 1674.93 } & 1711.13 & {\bf 1.449 } & 77.861 \\
                                    & 5 & 9465.303 & {\bf 8851.235 } & {\bf 7.269 } & 390.698 \\
                                    & 10 & {\bf 20179.256 } & 52479.808 & {\bf 14.524 } & 780.631 \\ \hline
             \multirow{3}{*}{ Q11 } & 1 & 123.001 & {\bf 84.227 } & {\bf 0.121 } & 5.928 \\
                                    & 5 & 548.9 & {\bf 523.548 } & {\bf 0.576 } & 28.146 \\
                                    & 10 & {\bf 994.286 } & 18564.246 & {\bf 1.149 } & 56.193 \\ \hline
             \multirow{3}{*}{ Q12 } & 1 & {\bf 1128.692 } & 1247.274 & {\bf 0.394 } & 10.507 \\
                                    & 5 & {\bf 5149.191 } & 15537.438 & {\bf 1.974 } & 52.627 \\
                                    & 10 & {\bf 10833.248 } & 45744.079 & {\bf 3.94 } & 105.057 \\ \hline
             \multirow{3}{*}{ Q13 } & 1 & 4595.15 & {\bf 2009.385 } & {\bf 3.596 } & 74.4 \\
                                    & 5 & 22778.525 & {\bf 10871.238 } & {\bf 17.98 } & 372.001 \\
                                    & 10 & 46638.458 & {\bf 22716.047 } & {\bf 35.96 } & 744.001 \\ \hline
             \multirow{3}{*}{ Q14 } & 1 & {\bf 292.902 } & 349.172 & {\bf 0.27 } & 18.571 \\
                                    & 5 & {\bf 1530.04 } & 5593.783 & {\bf 1.347 } & 92.793 \\
                                    & 10 & {\bf 3086.853 } & 48864.272 & {\bf 2.696 } & 185.705 \\ \hline
             \multirow{3}{*}{ Q15 } & 1 & 738.978 & {\bf 442.769 } & {\bf 1.083 } & 52.914 \\
                                    & 5 & {\bf 3648.849 } & 7119.87 & {\bf 5.423 } & 265.099 \\
                                    & 10 & {\bf 7489.047 } & 45756.265 & {\bf 10.848 } & 530.346 \\ \hline
             \multirow{3}{*}{ Q16 } & 1 & 383.589 & {\bf 301.347 } & {\bf 0.37 } & 16.455 \\
                                    & 5 & {\bf 1787.028 } & 2912.128 & {\bf 1.839 } & 81.74 \\
                                    & 10 & {\bf 3568.993 } & 4667.728 & {\bf 3.689 } & 163.946 \\ \hline
             \multirow{3}{*}{ Q17 } & 1 & 6788.831 & {\bf 6659.41 } & {\bf 0.72 } & 11.2 \\
                                    & 5 & 17970.305 & {\bf 17523.58 } & {\bf 3.6 } & 56.0 \\
                                    & 10 & {\bf 39062.359 } & 69424.705 & {\bf 7.2 } & 112.0 \\ \hline
             \multirow{3}{*}{ Q18 } & 1 & 15496.731 & {\bf 12175.682 } & {\bf 26.108 } & 1272.737 \\
                                    & 5 & {\bf 80673.913 } & 85605.048 & {\bf 130.43 } & 6358.46 \\
                                    & 10 & {\bf 172242.749 } & 188263.583 & {\bf 260.7 } & 12709.096 \\ \hline
             \multirow{3}{*}{ Q19 } & 1 & {\bf 128.724 } & 253.978 & {\bf 0.022 } & 1.501 \\
                                    & 5 & {\bf 770.826 } & 4591.412 & {\bf 0.109 } & 7.47 \\
                                    & 10 & {\bf 1244.035 } & 73082.065 & {\bf 0.211 } & 14.506 \\ \hline
             \multirow{3}{*}{ Q20 } & 1 & 2058.655 & {\bf 1527.349 } & {\bf 3.289 } & 1056.214 \\
                                    & 5 & {\bf 10373.605 } & 26763.565 & {\bf 16.443 } & 5279.964 \\
                                    & 10 & {\bf 22292.044 } & 67862.73 & {\bf 32.845 } & 10557.545 \\ \hline
             \multirow{3}{*}{ Q21 } & 1 & {\bf 879.998 } & 1045.606 & {\bf 0.026 } & 1.332 \\
                                    & 5 & {\bf 10048.959 } & 31560.391 & {\bf 0.125 } & 6.538 \\
                                    & 10 & {\bf 25372.791 } & 115913.909 & {\bf 0.256 } & 13.417 \\ \hline
             \multirow{3}{*}{ Q22 } & 1 & 118.981 & {\bf 78.522 } & {\bf 0.023 } & 0.405 \\
                                    & 5 & {\bf 538.636 } & 550.1 & {\bf 0.114 } & 2.035 \\
                                    & 10 & {\bf 904.977 } & 1966.22 & {\bf 0.229 } & 4.069 \\ \hline
        \end{tabular}
        \caption{Milestone Query vs. Table Query}
        \label{tab:milestone-vs-table-pg-100}
    \end{subtable}
    \caption{Experiment results for Page Size 100. Shorter time and smaller sizes are boldfaced.}
    \label{tab:general-result-pg-100}
\end{table*}

\clearpage
\begin{table*}[ht]
\vspace{-4mm}
    \centering
    \scriptsize
    \begin{subtable}[b]{0.55\textwidth}\centering
        \begin{tabular}{|c|c|c|c|c|c|c|c|}
             \cline{3-8}
             \multicolumn{2}{c|}{} & \multicolumn{2}{c|}{1GB} & \multicolumn{2}{c|}{5GB} & \multicolumn{2}{c|}{10GB} \\ \hline
             Query & Page & Opt. & Base & Opt. & Base & Opt. & Base  \\ \hline
             \multirow{3}{*}{ Q1 } & head & {\bf 0.096 } & 0.106 & 0.128 & {\bf 0.113 } & 0.211 & {\bf 0.106 } \\
                                    & mid & {\bf 0.081 } & 277.353 & {\bf 0.149 } & 9062.271 & {\bf 1.296 } & 21098.182 \\
                                    & tail & {\bf 0.08 } & 554.601 & {\bf 0.131 } & 18124.429 & {\bf 0.132 } & 42196.259 \\ \hline
             \multirow{3}{*}{ Q2 } & head & {\bf 14.495 } & 258.075 & {\bf 27.029 } & 1721.245 & {\bf 26.565 } & 3314.087 \\
                                    & mid & {\bf 15.403 } & 257.304 & {\bf 33.257 } & 1723.475 & {\bf 27.33 } & 3599.783 \\
                                    & tail & {\bf 14.901 } & 256.534 & {\bf 34.94 } & 1721.245 & {\bf 27.07 } & 3885.479 \\ \hline
             \multirow{3}{*}{ Q3 } & head & {\bf 3.846 } & 12.189 & {\bf 113.863 } & 264.195 & {\bf 105.038 } & 626.635 \\
                                    & mid & {\bf 3.635 } & 311.993 & {\bf 64.68 } & 4951.953 & {\bf 122.484 } & 21306.666 \\
                                    & tail & {\bf 3.79 } & 611.796 & {\bf 61.986 } & 9560.932 & {\bf 99.174 } & 40314.53 \\ \hline
             \multirow{3}{*}{ Q4 } & head & 1.661 & {\bf 1.226 } & {\bf 1.483 } & 1.534 & 31.712 & {\bf 7.704 } \\
                                    & mid & {\bf 1.355 } & 110.832 & {\bf 1.57 } & 567.687 & {\bf 13.946 } & 8836.858 \\
                                    & tail & {\bf 1.319 } & 220.437 & {\bf 2.908 } & 1133.84 & {\bf 10.08 } & 17666.011 \\ \hline
             \multirow{3}{*}{ Q5 } & head & {\bf 23.257 } & 26.016 & {\bf 32.109 } & 40.857 & 1097.932 & {\bf 735.291 } \\
                                    & mid & {\bf 21.175 } & 1265.583 & {\bf 33.95 } & 9974.61 & {\bf 648.512 } & 28437.192 \\
                                    & tail & {\bf 18.287 } & 2505.151 & {\bf 26.898 } & 19908.362 & {\bf 543.151 } & 55813.192 \\ \hline
             \multirow{3}{*}{ Q6 } & head & {\bf 0.791 } & 1.039 & 0.764 & {\bf 0.648 } & 2.651 & {\bf 0.774 } \\
                                    & mid & {\bf 0.747 } & 220.038 & {\bf 0.719 } & 1200.926 & {\bf 4.505 } & 16894.755 \\
                                    & tail & {\bf 0.677 } & 437.797 & {\bf 0.888 } & 2401.204 & {\bf 2.423 } & 33589.6 \\ \hline
             \multirow{3}{*}{ Q7 } & head & {\bf 44.788 } & 47.835 & 122.091 & {\bf 61.945 } & 2770.886 & {\bf 68.498 } \\
                                    & mid & {\bf 44.945 } & 420.81 & {\bf 71.516 } & 3093.383 & {\bf 1552.064 } & 19057.859 \\
                                    & tail & {\bf 44.876 } & 775.126 & {\bf 59.541 } & 6123.132 & {\bf 1437.85 } & 30266.428 \\ \hline
             \multirow{3}{*}{ Q8 } & head & 17.249 & {\bf 14.642 } & {\bf 25.396 } & 43.739 & 664.931 & {\bf 30.171 } \\
                                    & mid & {\bf 14.123 } & 1646.464 & {\bf 30.191 } & 6700.285 & {\bf 600.04 } & 22275.784 \\
                                    & tail & {\bf 14.213 } & 2969.118 & {\bf 18.549 } & 13356.832 & {\bf 596.333 } & 44521.396 \\ \hline
             \multirow{3}{*}{ Q9 } & head & {\bf 2.317 } & 226.431 & {\bf 98.203 } & 7932.852 & {\bf 94.028 } & 62387.44 \\
                                    & mid & {\bf 2.26 } & 984.02 & {\bf 99.992 } & 12663.04 & {\bf 94.809 } & 58133.264 \\
                                    & tail & {\bf 2.04 } & 1741.609 & {\bf 95.246 } & 17393.229 & {\bf 93.872 } & 55136.418 \\ \hline
             \multirow{3}{*}{ Q10 } & head & 2.702 & {\bf 2.553 } & 42.657 & {\bf 3.849 } & 51.023 & {\bf 3.883 } \\
                                    & mid & {\bf 2.37 } & 790.804 & {\bf 29.605 } & 7651.506 & {\bf 31.641 } & 21214.82 \\
                                    & tail & {\bf 2.596 } & 1579.054 & {\bf 30.024 } & 13908.715 & {\bf 32.623 } & 42425.757 \\ \hline
             \multirow{3}{*}{ Q11 } & head & 70.328 & {\bf 6.789 } & 377.987 & {\bf 12.112 } & 781.963 & {\bf 16.625 } \\
                                    & mid & 70.67 & {\bf 35.681 } & 391.113 & {\bf 224.163 } & 841.303 & {\bf 394.805 } \\
                                    & tail & 74.526 & {\bf 64.572 } & {\bf 369.919 } & 436.215 & {\bf 750.659 } & 772.984 \\ \hline
             \multirow{3}{*}{ Q12 } & head & 3.016 & {\bf 1.921 } & 13.035 & {\bf 2.318 } & 4.646 & {\bf 3.263 } \\
                                    & mid & {\bf 2.906 } & 284.822 & {\bf 2.219 } & 7395.704 & {\bf 3.03 } & 18224.874 \\
                                    & tail & {\bf 1.961 } & 567.722 & {\bf 1.727 } & 12139.442 & {\bf 2.067 } & 36446.485 \\ \hline
             \multirow{3}{*}{ Q13 } & head & {\bf 0.577 } & 0.583 & 0.867 & {\bf 0.495 } & 0.59 & {\bf 0.526 } \\
                                    & mid & {\bf 0.621 } & 462.139 & {\bf 0.765 } & 2654.764 & {\bf 0.553 } & 5586.485 \\
                                    & tail & {\bf 0.541 } & 923.696 & {\bf 0.947 } & 6597.823 & {\bf 0.582 } & 11172.445 \\ \hline
             \multirow{3}{*}{ Q14 } & head & 5.956 & {\bf 3.442 } & 38.885 & {\bf 22.228 } & 99.691 & {\bf 18.986 } \\
                                    & mid & {\bf 5.508 } & 116.254 & {\bf 37.352 } & 741.654 & {\bf 97.37 } & 1605.492 \\
                                    & tail & {\bf 5.345 } & 229.065 & {\bf 35.672 } & 1461.079 & {\bf 97.398 } & 3191.998 \\ \hline
             \multirow{3}{*}{ Q15 } & head & 0.618 & {\bf 0.526 } & 0.594 & {\bf 0.507 } & 10.788 & {\bf 0.67 } \\
                                    & mid & {\bf 0.542 } & 157.453 & {\bf 9.813 } & 946.187 & {\bf 9.923 } & 1816.009 \\
                                    & tail & {\bf 0.62 } & 314.379 & {\bf 2.305 } & 1891.866 & {\bf 12.64 } & 3631.348 \\ \hline
             \multirow{3}{*}{ Q16 } & head & 2.014 & {\bf 1.749 } & {\bf 6.852 } & 9.216 & 12.87 & {\bf 12.089 } \\
                                    & mid & {\bf 2.402 } & 138.499 & {\bf 7.195 } & 868.328 & {\bf 14.227 } & 1774.517 \\
                                    & tail & {\bf 2.364 } & 275.248 & {\bf 6.7 } & 1727.44 & {\bf 13.133 } & 3536.946 \\ \hline
             \multirow{3}{*}{ Q17 } & head & 8.908 & {\bf 7.311 } & 25.304 & {\bf 6.955 } & 48.554 & {\bf 8.495 } \\
                                    & mid & {\bf 8.768 } & 1493.372 & {\bf 21.24 } & 8114.816 & {\bf 29.671 } & 27939.016 \\
                                    & tail & {\bf 8.532 } & 2979.433 & {\bf 20.793 } & 16222.677 & {\bf 29.471 } & 55869.537 \\ \hline
             \multirow{3}{*}{ Q18 } & head & {\bf 2620.295 } & 5958.982 & {\bf 19362.324 } & 55544.906 & {\bf 36723.62 } & 124199.714 \\
                                    & mid & {\bf 2629.855 } & 8101.114 & {\bf 18941.112 } & 63877.35 & {\bf 37481.778 } & 139121.515 \\
                                    & tail & {\bf 2625.798 } & 10243.246 & {\bf 19651.261 } & 72209.794 & {\bf 38882.822 } & 154043.316 \\ \hline
             \multirow{3}{*}{ Q19 } & head & {\bf 53.896 } & 116.645 & {\bf 249.324 } & 937.218 & {\bf 358.445 } & 1142.017 \\
                                    & mid & {\bf 57.568 } & 116.792 & {\bf 241.046 } & 789.953 & {\bf 536.344 } & 1126.027 \\
                                    & tail & {\bf 58.638 } & 116.939 & {\bf 127.755 } & 642.689 & {\bf 241.034 } & 1142.017 \\ \hline
             \multirow{3}{*}{ Q20 } & head & 0.277 & {\bf 0.111 } & 0.271 & {\bf 0.094 } & {\bf 1.557 } & 3.551 \\
                                    & mid & {\bf 0.247 } & 49.705 & {\bf 0.239 } & 410.093 & {\bf 3.099 } & 1619.594 \\
                                    & tail & {\bf 0.257 } & 99.299 & {\bf 0.301 } & 820.092 & {\bf 1.567 } & 3235.638 \\ \hline
             \multirow{3}{*}{ Q21 } & head & {\bf 62.057 } & 822.659 & {\bf 665.833 } & 15239.98 & {\bf 4164.863 } & 51430.138 \\
                                    & mid & {\bf 42.706 } & 823.801 & {\bf 368.077 } & 15062.294 & {\bf 4097.071 } & 51530.172 \\
                                    & tail & {\bf 45.318 } & 824.942 & {\bf 204.653 } & 14884.608 & {\bf 3898.247 } & 51430.138 \\ \hline
             \multirow{3}{*}{ Q22 } & head & {\bf 1.883 } & 34.867 & {\bf 1.865 } & 157.814 & {\bf 3.948 } & 319.443 \\
                                    & mid & {\bf 1.962 } & 53.499 & {\bf 2.578 } & 263.99 & {\bf 3.639 } & 532.957 \\
                                    & tail & {\bf 1.787 } & 72.131 & {\bf 1.9 } & 368.158 & {\bf 3.279 } & 746.472 \\ \hline
        \end{tabular}
        \caption{Optimization vs. Baseline for Page Query Execution Time}
        \label{tab:optimization-vs-baseline-pg-200}
    \end{subtable}
    \quad
    \begin{subtable}[b]{0.42\textwidth}\centering
        \begin{tabular}{|c|c|c|c|c|c|}
             \cline{3-6}
             \multicolumn{2}{c}{} & \multicolumn{2}{|c|}{Time (ms)} & \multicolumn{2}{c|}{Output Size (MB)} \\ \hline
             Query & DB size & Milestone & Table & Milestone & Table \\ \hline
             \multirow{3}{*}{ Q1 } & 1 & 8266.691 & {\bf 1421.319 } & {\bf 7.79 } & 761.653 \\
                                    & 5 & 43465.288 & {\bf 19863.953 } & {\bf 38.964 } & 3809.747 \\
                                    & 10 & 87255.985 & {\bf 44610.896 } & {\bf 77.917 } & 7618.516 \\ \hline
             \multirow{3}{*}{ Q2 } & 1 & {\bf 211.127 } & 316.426 & {\bf 0.003 } & 0.129 \\
                                    & 5 & {\bf 1368.572 } & 2102.53 & {\bf 0.008 } & 0.647 \\
                                    & 10 & {\bf 2731.286 } & 4999.87 & {\bf 0.016 } & 1.31 \\ \hline
             \multirow{3}{*}{ Q3 } & 1 & {\bf 1076.937 } & 1434.405 & {\bf 0.054 } & 9.358 \\
                                    & 5 & {\bf 6026.493 } & 17564.81 & {\bf 0.264 } & 45.735 \\
                                    & 10 & {\bf 14199.75 } & 39549.761 & {\bf 0.536 } & 92.932 \\ \hline
             \multirow{3}{*}{ Q4 } & 1 & {\bf 322.347 } & 347.205 & {\bf 0.096 } & 6.367 \\
                                    & 5 & 1572.678 & {\bf 1554.638 } & {\bf 0.475 } & 31.649 \\
                                    & 10 & {\bf 11860.375 } & 25225.163 & {\bf 0.948 } & 63.165 \\ \hline
             \multirow{3}{*}{ Q5 } & 1 & {\bf 2166.316 } & 2745.94 & {\bf 0.234 } & 30.979 \\
                                    & 5 & {\bf 11554.084 } & 20418.954 & {\bf 1.169 } & 154.836 \\
                                    & 10 & {\bf 28022.606 } & 58548.771 & {\bf 2.332 } & 308.884 \\ \hline
             \multirow{3}{*}{ Q6 } & 1 & 757.998 & {\bf 583.419 } & {\bf 0.449 } & 43.884 \\
                                    & 5 & 3864.204 & {\bf 2686.491 } & {\bf 2.238 } & 218.804 \\
                                    & 10 & {\bf 9770.35 } & 34763.653 & {\bf 4.481 } & 438.125 \\ \hline
             \multirow{3}{*}{ Q7 } & 1 & {\bf 842.105 } & 1012.523 & {\bf 0.038 } & 0.947 \\
                                    & 5 & {\bf 6398.453 } & 6636.314 & {\bf 0.184 } & 4.602 \\
                                    & 10 & {\bf 15016.243 } & 35809.983 & {\bf 0.372 } & 9.306 \\ \hline
             \multirow{3}{*}{ Q8 } & 1 & 4944.503 & {\bf 3274.571 } & {\bf 2.918 } & 65.639 \\
                                    & 5 & 23695.185 & {\bf 17903.258 } & {\bf 14.514 } & 326.55 \\
                                    & 10 & 50859.253 & {\bf 44384.607 } & {\bf 29.071 } & 654.091 \\ \hline
             \multirow{3}{*}{ Q9 } & 1 & {\bf 2152.975 } & 2460.966 & {\bf 1.067 } & 23.998 \\
                                    & 5 & {\bf 17039.021 } & 19611.865 & {\bf 5.247 } & 118.061 \\
                                    & 10 & {\bf 32557.873 } & 54465.756 & {\bf 10.503 } & 236.312 \\ \hline
             \multirow{3}{*}{ Q10 } & 1 & {\bf 1574.498 } & 1681.446 & {\bf 0.724 } & 77.861 \\
                                    & 5 & 9602.298 & {\bf 8902.297 } & {\bf 3.635 } & 390.698 \\
                                    & 10 & {\bf 20372.195 } & 50133.868 & {\bf 7.262 } & 780.631 \\ \hline
             \multirow{3}{*}{ Q11 } & 1 & 117.053 & {\bf 80.026 } & {\bf 0.061 } & 5.928 \\
                                    & 5 & 552.067 & {\bf 541.54 } & {\bf 0.288 } & 28.146 \\
                                    & 10 & {\bf 998.988 } & 18808.77 & {\bf 0.575 } & 56.193 \\ \hline
             \multirow{3}{*}{ Q12 } & 1 & {\bf 1129.749 } & 1261.052 & {\bf 0.197 } & 10.507 \\
                                    & 5 & {\bf 5083.999 } & 14841.59 & {\bf 0.987 } & 52.627 \\
                                    & 10 & {\bf 10966.114 } & 41976.288 & {\bf 1.97 } & 105.057 \\ \hline
             \multirow{3}{*}{ Q13 } & 1 & 4588.313 & {\bf 1920.396 } & {\bf 1.798 } & 74.4 \\
                                    & 5 & 29358.318 & {\bf 12564.698 } & {\bf 8.59 } & 372.001 \\
                                    & 10 & 46773.33 & {\bf 22870.237 } & {\bf 17.98 } & 744.001 \\ \hline
             \multirow{3}{*}{ Q14 } & 1 & {\bf 285.983 } & 330.957 & {\bf 0.135 } & 18.571 \\
                                    & 5 & {\bf 1564.669 } & 16337.804 & {\bf 0.674 } & 92.793 \\
                                    & 10 & {\bf 3066.596 } & 42790.62 & {\bf 1.348 } & 185.705 \\ \hline
             \multirow{3}{*}{ Q15 } & 1 & 671.876 & {\bf 427.38 } & {\bf 0.541 } & 52.914 \\
                                    & 5 & {\bf 3607.563 } & 15686.577 & {\bf 2.712 } & 265.099 \\
                                    & 10 & {\bf 7351.63 } & 44726.829 & {\bf 5.424 } & 530.346 \\ \hline
             \multirow{3}{*}{ Q16 } & 1 & 358.35 & {\bf 301.054 } & {\bf 0.185 } & 16.455 \\
                                    & 5 & {\bf 1764.462 } & 2179.797 & {\bf 0.92 } & 81.74 \\
                                    & 10 & {\bf 3583.625 } & 4496.268 & {\bf 1.845 } & 163.946 \\ \hline
             \multirow{3}{*}{ Q17 } & 1 & 6743.266 & {\bf 6485.26 } & {\bf 0.36 } & 11.2 \\
                                    & 5 & {\bf 19005.803 } & 21509.807 & {\bf 1.8 } & 56.0 \\
                                    & 10 & {\bf 39388.936 } & 64307.077 & {\bf 3.6 } & 112.0 \\ \hline
             \multirow{3}{*}{ Q18 } & 1 & 15197.81 & {\bf 12058.411 } & {\bf 13.054 } & 1272.737 \\
                                    & 5 & {\bf 81764.098 } & 87866.796 & {\bf 65.215 } & 6358.46 \\
                                    & 10 & {\bf 170623.134 } & 188950.673 & {\bf 130.35 } & 12709.096 \\ \hline
             \multirow{3}{*}{ Q19 } & 1 & {\bf 123.344 } & 260.007 & {\bf 0.011 } & 1.501 \\
                                    & 5 & {\bf 751.385 } & 24049.184 & {\bf 0.054 } & 7.47 \\
                                    & 10 & {\bf 1224.462 } & 81688.742 & {\bf 0.105 } & 14.506 \\ \hline
             \multirow{3}{*}{ Q20 } & 1 & 2040.119 & {\bf 1499.107 } & {\bf 1.644 } & 1056.214 \\
                                    & 5 & {\bf 10370.233 } & 24150.297 & {\bf 8.221 } & 5279.964 \\
                                    & 10 & {\bf 22286.322 } & 67750.294 & {\bf 16.423 } & 10557.545 \\ \hline
             \multirow{3}{*}{ Q21 } & 1 & {\bf 904.569 } & 1049.034 & {\bf 0.013 } & 1.332 \\
                                    & 5 & {\bf 10135.259 } & 23102.873 & {\bf 0.063 } & 6.538 \\
                                    & 10 & {\bf 25376.696 } & 115486.433 & {\bf 0.128 } & 13.417 \\ \hline
             \multirow{3}{*}{ Q22 } & 1 & 110.451 & {\bf 79.542 } & {\bf 0.012 } & 0.405 \\
                                    & 5 & {\bf 476.884 } & 494.633 & {\bf 0.057 } & 2.035 \\
                                    & 10 & {\bf 928.369 } & 1868.3 & {\bf 0.114 } & 4.069 \\ \hline
        \end{tabular}
        \caption{Milestone Query vs. Table Query}
        \label{tab:milestone-vs-table-pg-200}
    \end{subtable}
    \caption{Experiment results for Page Size 200. Shorter time and smaller sizes are boldfaced.}
    \label{tab:general-result-pg-200}
\end{table*}

\clearpage

%% file: sections/appendix/user-study.tex
\section{User Study Appendix}
\label{appendix:user-study}

\subsection{Debugging Questions in Quiz}

The user study is conducted with an instance of beer database with the following schema (keys are underlined):

\begin{itemize}[leftmargin=*]
    \item Drinker (\underline{name}, address)
    \item Bar (\underline{name}, address)
    \item Beer (\underline{name}, brewery)
    \item Frequents (\underline{drinker}, \underline{bar}, times\_a\_week)
    \item Serves (\underline{bar}, \underline{beer}, price)
    \item Likes (\underline{drinker}, \underline{beer})
\end{itemize}

\subsubsection{Debugging Question 1}
For each bar Ben visits, find price of the most expensive and cheapest drink at that bar. Format the output as (bar, price), no duplicates.

The wrong query presented to the students are as follows:
\begin{lstlisting}[language=SQL, mathescape=true, basicstyle=\footnotesize\ttfamily]
WITH t1 AS (
  SELECT bar, price
  FROM serves
  WHERE price = (
      SELECT MAX(S1.price)
      FROM serves S1
      WHERE S1.bar = bar
    )
  UNION ALL
  SELECT bar, price
  FROM serves
  WHERE price = (
      SELECT MIN(S1.price)
      FROM serves S1
      WHERE S1.bar = bar
    )
)
SELECT t1.bar, t1.price
FROM t1, frequents
WHERE t1.bar = frequents.bar
  AND frequents.drinker = 'Ben';
\end{lstlisting}

The above query has two mistakes:
\begin{enumerate}[leftmargin=*]
    \item \sql{UNION ALL} creates duplicates when the most expensive and cheapest drinks share the same price.
    \item The \sql{bar} in both scalar subqueries is referencing the wrong column. Without correct aliasing, both \sql{bar} refer to the \sql{bar} in \sql{S1}, making the \WHERE\ condition a tautology.
\end{enumerate}

\subsubsection{Debugging Question 2}
Suppose every time a drinker frequents a bar, he buys all his favorite beers at that bar. Find the expected weekly revenue of each bar and rank them by the revenue from high to low. The output should be in the format of (bar, revenue). If a bar is not frequented by any drinker, or it does not serve any beer, or none of its beer is liked by any drinker, output (bar, NULL).

The wrong query presented to the students is as follows:
\begin{lstlisting}[language=SQL, mathescape=true, basicstyle=\footnotesize\ttfamily]
SELECT S.bar,
  SUM(F.times_a_week) * SUM(S.price) AS revenue
FROM serves S,
  frequents F,
  likes L
WHERE S.bar = F.bar
  AND S.beer = L.beer
GROUP BY S.bar
ORDER BY revenue DESC;
\end{lstlisting}

The above query has three mistakes:
\begin{enumerate}[leftmargin=*]
    \item The join predicate \sql{F.drinker = L.drinker} is missing.
    \item The expression for the sum is incorrect as it will blow up the result. The correct expression is \sql{SUM(F.times\_a\_week * S.price)}.
    \item There will be no ``\sql{NULL}'' tuple produced by the query, i.e., bars which do not serve any beer / serve no beer liked by anyone will not be included in the result. 
\end{enumerate}

\subsection{Brief Case Study for Debugging with LLM}
To examine how large language models (LLMs) perform on query debugging tasks, we fed both debugging questions in the quiz to an LLM and verified the correctness of its response. We perform the test with Gemini-3-pro and GPT-5.4, which represented the state-of-the-art LLMs at the time this paper was written. 

Since feeding the entire database instance to LLMs for debugging is usually not feasible, we used the following prompt for debugging and asked LLM to debug only based on the database schema and the query, with the assumption that the query is known to be incorrect: 

\begin{tcolorbox}[title={Prompt 1}, boxrule=0.8pt, colframe=gray, boxsep=2pt, left=2pt, right=2pt, top=2pt, bottom=2pt]
            \begin{lstlisting}[style=TinyJSON]
You are a SQL expert who can debug semantically 
incorrect query. 

Consider the following database schema:
{{ Database Schema }}

Consider the following question:
{{ Question }}

The following query is proven to be semantically 
incorrect:
{{ Query }}

What are the mistakes in the query?
\end{lstlisting}
\end{tcolorbox}

To ensure consistency, for each debugging question, we fed the LLM with the same prompt simultaneously in five separate conversations. We obtained the following overall result:
\begin{itemize}[leftmargin=*]
    \item For debugging question 1: Gemini-3-pro caught both bugs in 4 conversations, and it missed bug (2) in the remaining conversation; GPT-5.4 caught both bugs in all 5 conversations.
    \item For debugging question 2, both models caught all bugs in all 5 conversations.
    \item No hallucination was spotted.
\end{itemize}

While the above results looked promising, we further tested LLM's capability by removing the assumption that the query is wrong with the same setting in 5 isolated conversations: 
\begin{tcolorbox}[title={Prompt 2}, boxrule=0.8pt, colframe=gray, boxsep=2pt, left=2pt, right=2pt, top=2pt, bottom=2pt]
            \begin{lstlisting}[style=TinyJSON]
You are a SQL expert who can debug semantically 
incorrect query. 

Consider the following database schema:
{{ Database Schema }}

Consider the following question:
{{ Question }}

Is the following query semantically correct? 
If not, what are the mistakes? 
{{ Query }}
\end{lstlisting}
\end{tcolorbox}

As a result, a weaker assumption degrades the accuracy of the models:
\begin{itemize}[leftmargin=*]
    \item For debugging question 1: Gemini-3-pro caught both bugs in 3 conversations, and it missed bug (2) in one conversation and marked the query correct in one conversation; GPT-5.4 caught both bugs in 2 conversations and missed bug 2 in the remaining 3 conversations.
    \item For debugging question 2: Gemini-3-pro caught both bugs in all conversations while GPT-5.4 missed bug 3 in one conversation (it correctly identified the query would not produce \sql{NULL} but gave the wrong reason).
    \item No hallucination was spotted.
\end{itemize}

\subsubsection{Bar Raiser}
We further challenged LLMs with a harder question. Given the following schema: Friends(uid1, uid2, start\_ts, end\_ts). Returns for each pair of users, the maximal time periods during which they were friends.

The semantically incorrect query is the following:
\begin{lstlisting}[language=SQL]
WITH F(uid, start_ts, end_ts) AS (
  SELECT 
    DISTINCT uid1, 
    start_ts, 
    COALESCE(
      end_ts, CURRENT_TIME + INTERVAL '1 day'
    ) 
  FROM 
    Friends
), 
Mystery(uid, start_ts, end_ts) AS (
  SELECT f1.uid, f1.start_ts, f2.end_ts 
  FROM F f1, F f2 
  WHERE 
    f1.uid = f2.uid 
    AND f1.start_ts < f2.end_ts 
    AND NOT EXISTS (
      SELECT * 
      FROM F fi 
      WHERE 
        fi.uid = f1.uid 
        AND f1.start_ts < fi.end_ts 
        AND fi.end_ts < f2.end_ts 
        AND NOT EXISTS (
          SELECT * 
          FROM F 
          WHERE 
            uid = fi.uid 
            AND start_ts <= fi.end_ts 
            AND fi.end_ts < end_ts
        )
    )
) 
SELECT 
  uid, 
  start_ts, 
  NULLIF(
    end_ts, CURRENT_TIME + INTERVAL '1 day'
  ) 
FROM Mystery m 
WHERE 
  NOT EXISTS (
    SELECT * 
    FROM Mystery 
    WHERE 
      uid = m.uid 
      AND start_ts <= m.start_ts 
      AND m.end_ts <= end_ts
  );
\end{lstlisting}

In the above query, the \sql{NOT EXISTS} condition in the outer \sql{SELECT} block can never be satisfied, because both inequalities are non-strict (\sql{<=}), any row sql{m} will match itself (e.g., \sql{m.start\_ts <= m.start\_ts} is always true). As a result, \sql{EXISTS} is universally true, \sql{NOT EXISTS} is universally false, and every single row eliminates itself.

For each model (i.e., Gemini-3-pro and GPT-5.4) and each prompt (Prompt 1 and 2 with different assumptions), we create 5 separate conversations, and the results are the following:
\begin{itemize}[leftmargin=*]
    \item \textbf{Gemimi-3-pro:} With prompt 1, it correctly identified the bug in 4 conversations and missed it in one conversation; with prompt 2, it correctly identified the bugs in 3 conversations and missed it in 2 conversations. 
    \item \textbf{GPT-5.4:} With prompt 1, it correctly identified the bug in 3 conversations but missed it in the other 2; with prompt 2, it only correctly identified it in one conversation.
    \item Both models hallucinated and mentioned other irrelevant bugs in almost all conversations.
\end{itemize}

On the other hand, this bug can be easy to spot with \oursys: from the main query block, users can step into the execution of the \sql{NOT EXISTS} subquery with any external reference row for \sql{m.uid}, \sql{m.start\_ts}, \sql{m.end\_ts}, and quickly find out that the inequalities will always be evaluated to true as any external row will always ``semi-join'' with itself in the inner \sql{Mystery} table.

\subsubsection{Summary}
While Large Language Models (LLMs) are highly effective at debugging relatively simple SQL queries—particularly when operating under the strict assumption that a query is definitively flawed—their capabilities degrade significantly when faced with complex, multi-nested logic or when a user is merely uncertain about a query's correctness. Because LLMs act as probabilistic oracles rather than execution engines, they can easily hallucinate fixes or misinterpret the underlying data state. Consequently, there is a critical need for complementary tools to both ground the LLM's reasoning and allow users to rigorously verify the accuracy of its proposed explanations. In both capacities, \oursys\ provides a vital solution. By supplying a deterministic, step-by-step execution state, \oursys\ serves as both an authoritative verifier for human users and a reliable, objective execution environment that future LLM agents can query to validate their own hypotheses.

%% file: main.bib
@misc{tpch,
	author = {TPC Benchmark},
	howpublished = "\url{http://www.tpc.org/tpch}"
}

@misc{postgres,
	author = {PostgreSQL},
	howpublished = "\url{https://www.postgresql.org/}"
}

@inproceedings{buneman2001and,
  title={Why and where: A characterization of data provenance},
  author={Buneman, Peter and Khanna, Sanjeev and Wang-Chiew, Tan},
  booktitle={Database Theory—ICDT 2001: 8th International Conference London, UK, January 4--6, 2001 Proceedings 8},
  pages={316--330},
  year={2001},
  organization={Springer}
}

@inproceedings{green2007provenance,
  title={Provenance semirings},
  author={Green, Todd J and Karvounarakis, Grigoris and Tannen, Val},
  booktitle={Proceedings of the twenty-sixth ACM SIGMOD-SIGACT-SIGART symposium on Principles of database systems},
  pages={31--40},
  year={2007}
}

@inproceedings{bidoit2014query,
  title={Query-based why-not provenance with nedexplain},
  author={Bidoit, Nicole and Herschel, Melanie and Tzompanaki, Katerina},
  booktitle={Extending database technology (EDBT)},
  year={2014}
}

@inproceedings{chapman2009not,
  title={Why not?},
  author={Chapman, Adriane and Jagadish, HV},
  booktitle={Proceedings of the 2009 ACM SIGMOD International Conference on Management of data},
  pages={523--534},
  year={2009}
}

@article{huang2008provenance,
  title={On the provenance of non-answers to queries over extracted data},
  author={Huang, Jiansheng and Chen, Ting and Doan, AnHai and Naughton, Jeffrey F},
  journal={Proceedings of the VLDB Endowment},
  volume={1},
  number={1},
  pages={736--747},
  year={2008},
  publisher={VLDB Endowment}
}

@article{cui2000tracing,
  title={Tracing the lineage of view data in a warehousing environment},
  author={Cui, Yingwei and Widom, Jennifer and Wiener, Janet L},
  journal={ACM Transactions on Database Systems (TODS)},
  volume={25},
  number={2},
  pages={179--227},
  year={2000},
  publisher={ACM New York, NY, USA}
}

@inproceedings{agrawal2006trio,
  title={Trio: A system for data, uncertainty, and lineage},
  author={Agrawal, Parag and Benjelloun, Omar and Sarma, Anish Das and Hayworth, Chris and Nabar, Shubha and Sugihara, Tomoe and Widom, Jennifer},
  booktitle={VLDB},
  volume={6},
  pages={1151--1154},
  year={2006}
}

@article{karvounarakis2013collaborative,
  title={Collaborative data sharing via update exchange and provenance},
  author={Karvounarakis, Grigoris and Green, Todd J and Ives, Zachary G and Tannen, Val},
  journal={ACM Transactions on Database Systems (TODS)},
  volume={38},
  number={3},
  pages={1--42},
  year={2013},
  publisher={ACM New York, NY, USA}
}

@inproceedings{glavic2009perm,
  title={Perm: Processing provenance and data on the same data model through query rewriting},
  author={Glavic, Boris and Alonso, Gustavo},
  booktitle={2009 IEEE 25th International Conference on Data Engineering},
  pages={174--185},
  year={2009},
  organization={IEEE}
}

@article{arab2018gprom,
  title={GProM-a swiss army knife for your provenance needs},
  author={Arab, Bahareh Sadat and Feng, Su and Glavic, Boris and Lee, Seokki and Niu, Xing and Zeng, Qitian},
  journal={A Quarterly bulletin of the Computer Society of the IEEE Technical Committee on Data Engineering},
  volume={41},
  number={1},
  year={2018}
}

@article{lee2019pug,
  title={PUG: a framework and practical implementation for why and why-not provenance},
  author={Lee, Seokki and Lud{\"a}scher, Bertram and Glavic, Boris},
  journal={The VLDB Journal},
  volume={28},
  number={1},
  pages={47--71},
  year={2019},
  publisher={Springer}
}

@article{psallidas2018smoke,
  title={Smoke: Fine-grained lineage at interactive speed},
  author={Psallidas, Fotis and Wu, Eugene},
  journal={arXiv preprint arXiv:1801.07237},
  year={2018}
}

@inproceedings{interlandi2015titian,
  title={Titian: Data provenance support in spark},
  author={Interlandi, Matteo and Shah, Kshitij and Tetali, Sai Deep and Gulzar, Muhammad Ali and Yoo, Seunghyun and Kim, Miryung and Millstein, Todd and Condie, Tyson},
  booktitle={Proceedings of the VLDB Endowment International Conference on Very Large Data Bases},
  volume={9},
  number={3},
  pages={216},
  year={2015},
  organization={NIH Public Access}
}

@inproceedings{diestelkamper2020tracing,
  title={Tracing nested data with structural provenance for big data analytics.},
  author={Diestelk{\"a}mper, Ralf and Herschel, Melanie},
  booktitle={EDBT},
  pages={253--264},
  year={2020}
}

@article{amsterdamer2011putting,
  title={Putting lipstick on pig: Enabling database-style workflow provenance},
  author={Amsterdamer, Yael and Davidson, Susan B and Deutch, Daniel and Milo, Tova and Stoyanovich, Julia and Tannen, Val},
  journal={arXiv preprint arXiv:1201.0231},
  year={2011}
}

@inproceedings{jaakkola2003visual,
  title={Visual SQL--high-quality ER-based query treatment},
  author={Jaakkola, Hannu and Thalheim, Bernhard},
  booktitle={Conceptual Modeling for Novel Application Domains: ER 2003 Workshops ECOMO, IWCMQ, AOIS, and XSDM, Chicago, IL, USA, October 13, 2003. Proceedings 22},
  pages={129--139},
  year={2003},
  organization={Springer}
}

@inproceedings{cerullo2007system,
  title={A system for database visual querying and query visualization: Complementing text and graphics to increase expressiveness},
  author={Cerullo, Claudio and Porta, Marco},
  booktitle={18th International Workshop on Database and Expert Systems Applications (DEXA 2007)},
  pages={109--113},
  year={2007},
  organization={IEEE}
}

@inproceedings{haas1989extensible,
  title={Extensible query processing in Starburst},
  author={Haas, Laura M and Freytag, Johann Christoph and Lohman, Guy M and Pirahesh, Hamid},
  booktitle={Proceedings of the 1989 ACM SIGMOD international conference on Management of data},
  pages={377--388},
  year={1989}
}

@inproceedings{leventidis2020queryvis,
  title={QueryVis: Logic-based diagrams help users understand complicated SQL queries faster},
  author={Leventidis, Aristotelis and Zhang, Jiahui and Dunne, Cody and Gatterbauer, Wolfgang and Jagadish, HV and Riedewald, Mirek},
  booktitle={Proceedings of the 2020 ACM SIGMOD International Conference on Management of Data},
  pages={2303--2318},
  year={2020}
}

@misc{dbforge,
  title={dbForge},
  howpublished={\url{https://www.devart.com/dbforge/mysql/querybuilder/}},
  year={2023}
}

@misc{msaccess,
  title={Microsoft Access},
  howpublished={\url{https://www.microsoft.com/en-us/microsoft-365/access}},
  year={2023}
}

@misc{pgadmin,
  title={PgAdmin},
  howpublished={\url{https://www.pgadmin.org/}},
  year={2023}
}

@misc{rapidsql,
  title={Rapid SQL},
  howpublished={\url{https://www.idera.com/rapid-sql-ide/}},
  year={2023}
}

@inproceedings{miedema2021sqlvis,
  title={SQLVis: Visual query representations for supporting SQL learners},
  author={Miedema, Daphne and Fletcher, George},
  booktitle={2021 IEEE Symposium on Visual Languages and Human-Centric Computing (VL/HCC)},
  pages={1--9},
  year={2021},
  organization={IEEE}
}

@inproceedings{koutrika2010explaining,
  title={Explaining structured queries in natural language},
  author={Koutrika, Georgia and Simitsis, Alkis and Ioannidis, Yannis E},
  booktitle={2010 IEEE 26th International Conference on Data Engineering (ICDE 2010)},
  pages={333--344},
  year={2010},
  organization={IEEE}
}

@inproceedings{gehrmann2018end,
    title = "End-to-End Content and Plan Selection for Data-to-Text Generation",
    author = "Gehrmann, Sebastian  and
      Dai, Falcon  and
      Elder, Henry  and
      Rush, Alexander",
    booktitle = "Proceedings of the 11th International Conference on Natural Language Generation",
    month = nov,
    year = "2018",
    address = "Tilburg University, The Netherlands",
    publisher = "Association for Computational Linguistics",
    url = "https://aclanthology.org/W18-6505",
    doi = "10.18653/v1/W18-6505",
    pages = "46--56"
}

@inproceedings{xu2018sql,
    title = "{SQL}-to-Text Generation with Graph-to-Sequence Model",
    author = "Xu, Kun  and
      Wu, Lingfei  and
      Wang, Zhiguo  and
      Feng, Yansong  and
      Sheinin, Vadim",
    booktitle = "Proceedings of the 2018 Conference on Empirical Methods in Natural Language Processing",
    month = oct # "-" # nov,
    year = "2018",
    address = "Brussels, Belgium",
    publisher = "Association for Computational Linguistics",
    url = "https://aclanthology.org/D18-1112",
    doi = "10.18653/v1/D18-1112",
    pages = "931--936"
}

@inproceedings{iyer2016summarizing,
    title = "Summarizing Source Code using a Neural Attention Model",
    author = "Iyer, Srinivasan  and
      Konstas, Ioannis  and
      Cheung, Alvin  and
      Zettlemoyer, Luke",
    booktitle = "Proceedings of the 54th Annual Meeting of the Association for Computational Linguistics (Volume 1: Long Papers)",
    month = aug,
    year = "2016",
    address = "Berlin, Germany",
    publisher = "Association for Computational Linguistics",
    url = "https://aclanthology.org/P16-1195",
    doi = "10.18653/v1/P16-1195",
    pages = "2073--2083",
}

@inproceedings{shu2021logic,
    title = "Logic-Consistency Text Generation from Semantic Parses",
    author = "Shu, Chang  and
      Zhang, Yusen  and
      Dong, Xiangyu  and
      Shi, Peng  and
      Yu, Tao  and
      Zhang, Rui",
    booktitle = "Findings of the Association for Computational Linguistics: ACL-IJCNLP 2021",
    month = aug,
    year = "2021",
    address = "Online",
    publisher = "Association for Computational Linguistics",
    url = "https://aclanthology.org/2021.findings-acl.388",
    doi = "10.18653/v1/2021.findings-acl.388",
    pages = "4414--4426",
}

@article{brown2020language,
  title={Language models are few-shot learners},
  author={Brown, Tom and Mann, Benjamin and Ryder, Nick and Subbiah, Melanie and Kaplan, Jared D and Dhariwal, Prafulla and Neelakantan, Arvind and Shyam, Pranav and Sastry, Girish and Askell, Amanda and others},
  journal={Advances in neural information processing systems},
  volume={33},
  pages={1877--1901},
  year={2020}
}

@inproceedings{abouzied2012dataplay,
  title={Dataplay: interactive tweaking and example-driven correction of graphical database queries},
  author={Abouzied, Azza and Hellerstein, Joseph and Silberschatz, Avi},
  booktitle={Proceedings of the 25th annual ACM symposium on User interface software and technology},
  pages={207--218},
  year={2012}
}

@inproceedings{caballero2012declarative,
  title={Declarative debugging of wrong and missing answers for SQL views},
  author={Caballero, Rafael and Garc{\'\i}a-Ruiz, Yolanda and S{\'a}enz-P{\'e}rez, Fernando},
  booktitle={Functional and Logic Programming: 11th International Symposium, FLOPS 2012, Kobe, Japan, May 23-25, 2012. Proceedings 11},
  pages={73--87},
  year={2012},
  organization={Springer}
}

@inproceedings{caballero2012algorithmic,
  title={Algorithmic debugging of SQL views},
  author={Caballero, Rafael and Garc{\'\i}a-Ruiz, Yolanda and S{\'a}enz-P{\'e}rez, Fernando},
  booktitle={Perspectives of Systems Informatics: 8th International Andrei Ershov Memorial Conference, PSI 2011, Novosibirsk, Russia, June 27-July 1, 2011, Revised Selected Papers 8},
  pages={77--85},
  year={2012},
  organization={Springer}
}

@inproceedings{dimitriadou2014explore,
  title={Explore-by-example: An automatic query steering framework for interactive data exploration},
  author={Dimitriadou, Kyriaki and Papaemmanouil, Olga and Diao, Yanlei},
  booktitle={Proceedings of the 2014 ACM SIGMOD international conference on Management of data},
  pages={517--528},
  year={2014}
}

@inproceedings{le2019explique,
  title={Explique: Interactive databases exploration with SQL},
  author={Le Guilly, Marie and Petit, Jean-Marc and Scuturici, Vasile-Marian and Ilyas, Ihab F},
  booktitle={Proceedings of the 28th ACM International Conference on Information and Knowledge Management},
  pages={2877--2880},
  year={2019}
}

@article{akbarnejad2010sql,
  title={SQL QueRIE recommendations},
  author={Akbarnejad, Javad and Chatzopoulou, Gloria and Eirinaki, Magdalini and Koshy, Suju and Mittal, Sarika and On, Duc and Polyzotis, Neoklis and Varman, Jothi S Vindhiya},
  journal={Proceedings of the VLDB Endowment},
  volume={3},
  number={1-2},
  pages={1597--1600},
  year={2010},
  publisher={VLDB Endowment}
}

@article{miao2020rex,
  title={I-Rex: an interactive relational query explainer for SQL},
  author={Miao, Zhengjie and Chen, Tiangang and Bendeck, Alexander and Day, Kevin and Roy, Sudeepa and Yang, Jun},
  journal={Proceedings of the VLDB Endowment},
  volume={13},
  number={12},
  pages={2997--3000},
  year={2020},
  publisher={VLDB Endowment}
}

@inproceedings{hu2022rex,
  title={I-Rex: An Interactive Relational Query Debugger for SQL},
  author={Hu, Yihao and Miao, Zhengjie and Leong, Zhiming and Lim, Haechan and Zheng, Zachary and Roy, Sudeepa and Stephens-Martinez, Kristin and Yang, Jun},
  booktitle={Proceedings of the 53rd ACM Technical Symposium on Computer Science Education V. 2},
  pages={1180--1180},
  year={2022}
}

@article{DBLP:journals/jss/BrassG06,
  author       = {Stefan Brass and
                  Christian Goldberg},
  title        = {Semantic errors in {SQL} queries: {A} quite complete list},
  journal      = {J. Syst. Softw.},
  volume       = {79},
  number       = {5},
  pages        = {630--644},
  year         = {2006},
  url          = {https://doi.org/10.1016/j.jss.2005.06.028},
  doi          = {10.1016/J.JSS.2005.06.028},
  bibsource    = {dblp computer science bibliography, https://dblp.org}
}

@inproceedings{DBLP:conf/qsic/BrassG05,
  author       = {Stefan Brass and
                  Christian Goldberg},
  title        = {Proving the Safety of {SQL} Queries},
  booktitle    = {Fifth International Conference on Quality Software {(QSIC} 2005),
                  19-20 September 2005, Melbourne, Australia},
  pages        = {197--204},
  publisher    = {{IEEE} Computer Society},
  year         = {2005},
  url          = {https://doi.org/10.1109/QSIC.2005.50},
  doi          = {10.1109/QSIC.2005.50},
  bibsource    = {dblp computer science bibliography, https://dblp.org}
}

@inproceedings{DBLP:conf/gvd/BrassG04,
  author       = {Stefan Brass and
                  Christian Goldberg},
  editor       = {Mireille Samia and
                  Stefan Conrad},
  title        = {Detecting Logical Errors in {SQL} Queries},
  booktitle    = {Tagungsband zum 16. GI-Workshop Grundlagen von Datenbanken, Mohnheim,
                  NRW, Deutschland, 1.-4. Juni 2004},
  pages        = {28--32},
  publisher    = {Universit{\"{a}}t D{\"{u}}sseldorf},
  year         = {2004},
  bibsource    = {dblp computer science bibliography, https://dblp.org}
}

@inproceedings{DBLP:conf/lpar/VeanesTH10,
  author       = {Margus Veanes and
                  Nikolai Tillmann and
                  Jonathan de Halleux},
  editor       = {Edmund M. Clarke and
                  Andrei Voronkov},
  title        = {Qex: Symbolic {SQL} Query Explorer},
  booktitle    = {Logic for Programming, Artificial Intelligence, and Reasoning - 16th
                  International Conference, LPAR-16, Dakar, Senegal, April 25-May 1,
                  2010, Revised Selected Papers},
  series       = {Lecture Notes in Computer Science},
  volume       = {6355},
  pages        = {425--446},
  publisher    = {Springer},
  year         = {2010},
  url          = {https://doi.org/10.1007/978-3-642-17511-4\_24},
  doi          = {10.1007/978-3-642-17511-4\_24},
  bibsource    = {dblp computer science bibliography, https://dblp.org}
}

@article{DBLP:journals/pacmse/Haroon0G24,
  author       = {Sabaat Haroon and
                  Chris Brown and
                  Muhammad Ali Gulzar},
  title        = {DeSQL: Interactive Debugging of {SQL} in Data-Intensive Scalable Computing},
  journal      = {Proc. {ACM} Softw. Eng.},
  volume       = {1},
  number       = {{FSE}},
  pages        = {767--788},
  year         = {2024},
  url          = {https://doi.org/10.1145/3643761},
  doi          = {10.1145/3643761},
  bibsource    = {dblp computer science bibliography, https://dblp.org}
}

@article{grust2013observing,
  title={Observing sql queries in their natural habitat},
  author={Grust, Torsten and Rittinger, Jan},
  journal={ACM Transactions on Database Systems (TODS)},
  volume={38},
  number={1},
  pages={1--33},
  year={2013},
  publisher={ACM New York, NY, USA}
}

@inproceedings{dietrich2015sql,
  title={A SQL debugger built from spare parts: Turning a SQL: 1999 database system into its own debugger},
  author={Dietrich, Benjamin and Grust, Torsten},
  booktitle={Proceedings of the 2015 ACM SIGMOD International Conference on Management of Data},
  pages={865--870},
  year={2015}
}

@article{DBLP:journals/vldb/ChandraCKRS015,
  author       = {Bikash Chandra and
                  Bhupesh Chawda and
                  Biplab Kar and
                  K. V. Maheshwara Reddy and
                  Shetal Shah and
                  S. Sudarshan},
  title        = {Data generation for testing and grading {SQL} queries},
  journal      = {{VLDB} J.},
  volume       = {24},
  number       = {6},
  pages        = {731--755},
  year         = {2015},
  url          = {https://doi.org/10.1007/s00778-015-0395-0},
  doi          = {10.1007/S00778-015-0395-0},
  bibsource    = {dblp computer science bibliography, https://dblp.org}
}

@article{DBLP:journals/pvldb/ChuMRCS18,
  author       = {Shumo Chu and
                  Brendan Murphy and
                  Jared Roesch and
                  Alvin Cheung and
                  Dan Suciu},
  title        = {Axiomatic Foundations and Algorithms for Deciding Semantic Equivalences
                  of {SQL} Queries},
  journal      = {Proc. {VLDB} Endow.},
  volume       = {11},
  number       = {11},
  pages        = {1482--1495},
  year         = {2018},
  url          = {http://www.vldb.org/pvldb/vol11/p1482-chu.pdf},
  doi          = {10.14778/3236187.3236200},
  bibsource    = {dblp computer science bibliography, https://dblp.org}
}

@inproceedings{DBLP:conf/cidr/ChuWWC17,
  author       = {Shumo Chu and
                  Chenglong Wang and
                  Konstantin Weitz and
                  Alvin Cheung},
  title        = {Cosette: An Automated Prover for {SQL}},
  booktitle    = {8th Biennial Conference on Innovative Data Systems Research, {CIDR}
                  2017, Chaminade, CA, USA, January 8-11, 2017, Online Proceedings},
  publisher    = {www.cidrdb.org},
  year         = {2017},
  url          = {http://cidrdb.org/cidr2017/papers/p51-chu-cidr17.pdf},
  bibsource    = {dblp computer science bibliography, https://dblp.org}
}

@inproceedings{DBLP:conf/pldi/ChuWCS17,
  author       = {Shumo Chu and
                  Konstantin Weitz and
                  Alvin Cheung and
                  Dan Suciu},
  editor       = {Albert Cohen and
                  Martin T. Vechev},
  title        = {HoTTSQL: proving query rewrites with univalent {SQL} semantics},
  booktitle    = {Proceedings of the 38th {ACM} {SIGPLAN} Conference on Programming
                  Language Design and Implementation, {PLDI} 2017, Barcelona, Spain,
                  June 18-23, 2017},
  pages        = {510--524},
  publisher    = {{ACM}},
  year         = {2017},
  url          = {https://doi.org/10.1145/3062341.3062348},
  doi          = {10.1145/3062341.3062348},
  bibsource    = {dblp computer science bibliography, https://dblp.org}
}

@article{DBLP:journals/pacmmod/DingWYZXCPL23,
  author       = {Haoran Ding and
                  Zhaoguo Wang and
                  Yicun Yang and
                  Dexin Zhang and
                  Zhenglin Xu and
                  Haibo Chen and
                  Ruzica Piskac and
                  Jinyang Li},
  title        = {Proving Query Equivalence Using Linear Integer Arithmetic},
  journal      = {Proc. {ACM} Manag. Data},
  volume       = {1},
  number       = {4},
  pages        = {227:1--227:26},
  year         = {2023},
  url          = {https://doi.org/10.1145/3626768},
  doi          = {10.1145/3626768},
  bibsource    = {dblp computer science bibliography, https://dblp.org}
}

@article{DBLP:journals/pvldb/WangPC24,
  author       = {Shuxian Wang and
                  Sicheng Pan and
                  Alvin Cheung},
  title        = {{QED:} {A} Powerful Query Equivalence Decider for {SQL}},
  journal      = {Proc. {VLDB} Endow.},
  volume       = {17},
  number       = {11},
  pages        = {3602--3614},
  year         = {2024},
  url          = {https://www.vldb.org/pvldb/vol17/p3602-wang.pdf},
  bibsource    = {dblp computer science bibliography, https://dblp.org}
}

@inproceedings{DBLP:conf/sigmod/GiladMRY22,
  author       = {Amir Gilad and
                  Zhengjie Miao and
                  Sudeepa Roy and
                  Jun Yang},
  editor       = {Zachary G. Ives and
                  Angela Bonifati and
                  Amr El Abbadi},
  title        = {Understanding Queries by Conditional Instances},
  booktitle    = {{SIGMOD} '22: International Conference on Management of Data, Philadelphia,
                  PA, USA, June 12 - 17, 2022},
  pages        = {355--368},
  publisher    = {{ACM}},
  year         = {2022},
  url          = {https://doi.org/10.1145/3514221.3517898},
  doi          = {10.1145/3514221.3517898},
  bibsource    = {dblp computer science bibliography, https://dblp.org}
}

@article{DBLP:journals/pacmmod/HuGSRY24,
  author       = {Yihao Hu and
                  Amir Gilad and
                  Kristin Stephens{-}Martinez and
                  Sudeepa Roy and
                  Jun Yang},
  title        = {Qr-Hint: Actionable Hints Towards Correcting Wrong {SQL} Queries},
  journal      = {Proc. {ACM} Manag. Data},
  volume       = {2},
  number       = {3},
  pages        = {164},
  year         = {2024},
  url          = {https://doi.org/10.1145/3654995},
  doi          = {10.1145/3654995},
  bibsource    = {dblp computer science bibliography, https://dblp.org}
}

@inproceedings{DBLP:conf/sigmod/MiaoRY19,
  author       = {Zhengjie Miao and
                  Sudeepa Roy and
                  Jun Yang},
  editor       = {Peter A. Boncz and
                  Stefan Manegold and
                  Anastasia Ailamaki and
                  Amol Deshpande and
                  Tim Kraska},
  title        = {Explaining Wrong Queries Using Small Examples},
  booktitle    = {Proceedings of the 2019 International Conference on Management of
                  Data, {SIGMOD} Conference 2019, Amsterdam, The Netherlands, June 30
                  - July 5, 2019},
  pages        = {503--520},
  publisher    = {{ACM}},
  year         = {2019},
  url          = {https://doi.org/10.1145/3299869.3319866},
  doi          = {10.1145/3299869.3319866},
  bibsource    = {dblp computer science bibliography, https://dblp.org}
}

@inproceedings{DBLP:conf/icse/Presler-Marshall21,
  author       = {Kai Presler{-}Marshall and
                  Sarah Heckman and
                  Kathryn T. Stolee},
  title        = {SQLRepair: Identifying and Repairing Mistakes in Student-Authored
                  {SQL} Queries},
  booktitle    = {43rd {IEEE/ACM} International Conference on Software Engineering:
                  Software Engineering Education and Training, {ICSE} {(SEET)} 2021,
                  Madrid, Spain, May 25-28, 2021},
  pages        = {199--210},
  publisher    = {{IEEE}},
  year         = {2021},
  url          = {https://doi.org/10.1109/ICSE-SEET52601.2021.00030},
  doi          = {10.1109/ICSE-SEET52601.2021.00030},
  bibsource    = {dblp computer science bibliography, https://dblp.org}
}

@inproceedings{DBLP:conf/comad/ChandraBHJ021,
  author       = {Bikash Chandra and
                  Ananyo Banerjee and
                  Udbhas Hazra and
                  Mathew Joseph and
                  S. Sudarshan},
  editor       = {Jayant R. Haritsa and
                  Shourya Roy and
                  Manish Gupta and
                  Sharad Mehrotra and
                  Balaji Vasan Srinivasan and
                  Yogesh Simmhan},
  title        = {Edit Based Grading of {SQL} Queries},
  booktitle    = {{CODS-COMAD} 2021: 8th {ACM} {IKDD} {CODS} and 26th COMAD, Virtual
                  Event, Bangalore, India, January 2-4, 2021},
  pages        = {56--64},
  publisher    = {{ACM}},
  year         = {2021},
  url          = {https://doi.org/10.1145/3430984.3431012},
  doi          = {10.1145/3430984.3431012},
  bibsource    = {dblp computer science bibliography, https://dblp.org}
}

@inproceedings{DBLP:conf/pods/CarmeliZBKS20,
  author       = {Nofar Carmeli and
                  Shai Zeevi and
                  Christoph Berkholz and
                  Benny Kimelfeld and
                  Nicole Schweikardt},
  editor       = {Dan Suciu and
                  Yufei Tao and
                  Zhewei Wei},
  title        = {Answering (Unions of) Conjunctive Queries using Random Access and
                  Random-Order Enumeration},
  booktitle    = {Proceedings of the 39th {ACM} {SIGMOD-SIGACT-SIGAI} Symposium on Principles
                  of Database Systems, {PODS} 2020, Portland, OR, USA, June 14-19, 2020},
  pages        = {393--409},
  publisher    = {{ACM}},
  year         = {2020},
  url          = {https://doi.org/10.1145/3375395.3387662},
  doi          = {10.1145/3375395.3387662},
  bibsource    = {dblp computer science bibliography, https://dblp.org}
}

@article{DBLP:journals/pvldb/TziavelisGR21,
  author       = {Nikolaos Tziavelis and
                  Wolfgang Gatterbauer and
                  Mirek Riedewald},
  title        = {Beyond Equi-joins: Ranking, Enumeration and Factorization},
  journal      = {Proc. {VLDB} Endow.},
  volume       = {14},
  number       = {11},
  pages        = {2599--2612},
  year         = {2021},
  url          = {http://www.vldb.org/pvldb/vol14/p2599-tziavelis.pdf},
  doi          = {10.14778/3476249.3476306},
  bibsource    = {dblp computer science bibliography, https://dblp.org}
}

@article{DBLP:journals/pvldb/NiuGLLGKLP21,
  author       = {Xing Niu and
                  Boris Glavic and
                  Ziyu Liu and
                  Pengyuan Li and
                  Dieter Gawlick and
                  Vasudha Krishnaswamy and
                  Zhen Hua Liu and
                  Danica Porobic},
  title        = {Provenance-based Data Skipping},
  journal      = {Proc. {VLDB} Endow.},
  volume       = {15},
  number       = {3},
  pages        = {451--464},
  year         = {2021},
  url          = {http://www.vldb.org/pvldb/vol15/p451-niu.pdf},
  doi          = {10.14778/3494124.3494130},
  bibsource    = {dblp computer science bibliography, https://dblp.org}
}

@inproceedings{DBLP:conf/icdt/EldarCK24,
  author       = {Idan Eldar and
                  Nofar Carmeli and
                  Benny Kimelfeld},
  editor       = {Graham Cormode and
                  Michael Shekelyan},
  title        = {Direct Access for Answers to Conjunctive Queries with Aggregation},
  booktitle    = {27th International Conference on Database Theory, {ICDT} 2024, March
                  25-28, 2024, Paestum, Italy},
  series       = {LIPIcs},
  volume       = {290},
  pages        = {4:1--4:20},
  publisher    = {Schloss Dagstuhl - Leibniz-Zentrum f{\"{u}}r Informatik},
  year         = {2024},
  url          = {https://doi.org/10.4230/LIPIcs.ICDT.2024.4},
  doi          = {10.4230/LIPICS.ICDT.2024.4},
  bibsource    = {dblp computer science bibliography, https://dblp.org}
}

@article{DBLP:journals/tods/CarmeliTGKR23,
  author       = {Nofar Carmeli and
                  Nikolaos Tziavelis and
                  Wolfgang Gatterbauer and
                  Benny Kimelfeld and
                  Mirek Riedewald},
  title        = {Tractable Orders for Direct Access to Ranked Answers of Conjunctive
                  Queries},
  journal      = {{ACM} Trans. Database Syst.},
  volume       = {48},
  number       = {1},
  pages        = {1:1--1:45},
  year         = {2023},
  url          = {https://doi.org/10.1145/3578517},
  doi          = {10.1145/3578517},
  bibsource    = {dblp computer science bibliography, https://dblp.org}
}

@article{DBLP:journals/pvldb/SudhirTLHCM23,
  author       = {Sivaprasad Sudhir and
                  Wenbo Tao and
                  Nikolay Pavlovich Laptev and
                  Cyrille Habis and
                  Michael J. Cafarella and
                  Samuel Madden},
  title        = {Pando: Enhanced Data Skipping with Logical Data Partitioning},
  journal      = {Proc. {VLDB} Endow.},
  volume       = {16},
  number       = {9},
  pages        = {2316--2329},
  year         = {2023},
  url          = {https://www.vldb.org/pvldb/vol16/p2316-sudhir.pdf},
  doi          = {10.14778/3598581.3598601},
  bibsource    = {dblp computer science bibliography, https://dblp.org}
}

@article{DBLP:journals/pacmmod/YanLH23,
  author       = {Cong Yan and
                  Yin Lin and
                  Yeye He},
  title        = {Predicate Pushdown for Data Science Pipelines},
  journal      = {Proc. {ACM} Manag. Data},
  volume       = {1},
  number       = {2},
  pages        = {136:1--136:28},
  year         = {2023},
  url          = {https://doi.org/10.1145/3589281},
  doi          = {10.1145/3589281},
  bibsource    = {dblp computer science bibliography, https://dblp.org}
}

@article{DBLP:journals/pvldb/OrrKC19,
  author       = {Laurel J. Orr and
                  Srikanth Kandula and
                  Surajit Chaudhuri},
  title        = {Pushing Data-Induced Predicates Through Joins in Big-Data Clusters},
  journal      = {Proc. {VLDB} Endow.},
  volume       = {13},
  number       = {3},
  pages        = {252--265},
  year         = {2019},
  url          = {http://www.vldb.org/pvldb/vol13/p252-orr.pdf},
  doi          = {10.14778/3368289.3368292},
  bibsource    = {dblp computer science bibliography, https://dblp.org}
}

@inproceedings{DBLP:conf/icde/SeshadriPL96,
  author       = {Praveen Seshadri and
                  Hamid Pirahesh and
                  T. Y. Cliff Leung},
  editor       = {Stanley Y. W. Su},
  title        = {Complex Query Decorrelation},
  booktitle    = {Proceedings of the Twelfth International Conference on Data Engineering,
                  February 26 - March 1, 1996, New Orleans, Louisiana, {USA}},
  pages        = {450--458},
  publisher    = {{IEEE} Computer Society},
  year         = {1996},
  url          = {https://doi.org/10.1109/ICDE.1996.492194},
  doi          = {10.1109/ICDE.1996.492194},
  bibsource    = {dblp computer science bibliography, https://dblp.org}
}

@article{DBLP:journals/cacm/Bloom70,
  author       = {Burton H. Bloom},
  title        = {Space/Time Trade-offs in Hash Coding with Allowable Errors},
  journal      = {Commun. {ACM}},
  volume       = {13},
  number       = {7},
  pages        = {422--426},
  year         = {1970},
  url          = {https://doi.org/10.1145/362686.362692},
  doi          = {10.1145/362686.362692},
  bibsource    = {dblp computer science bibliography, https://dblp.org}
}

@inproceedings{DBLP:conf/sigmod/SelingerACLP79,
  author       = {Patricia G. Selinger and
                  Morton M. Astrahan and
                  Donald D. Chamberlin and
                  Raymond A. Lorie and
                  Thomas G. Price},
  editor       = {Philip A. Bernstein},
  title        = {Access Path Selection in a Relational Database Management System},
  booktitle    = {Proceedings of the 1979 {ACM} {SIGMOD} International Conference on
                  Management of Data, Boston, Massachusetts, USA, May 30 - June 1},
  pages        = {23--34},
  publisher    = {{ACM}},
  year         = {1979},
  url          = {https://doi.org/10.1145/582095.582099},
  doi          = {10.1145/582095.582099},
  bibsource    = {dblp computer science bibliography, https://dblp.org}
}

@article{DBLP:journals/pvldb/MullerDG18,
  author       = {Tobias M{\"{u}}ller and
                  Benjamin Dietrich and
                  Torsten Grust},
  title        = {You Say 'What', {I} Hear 'Where' and 'Why'? (Mis-)Interpreting {SQL}
                  to Derive Fine-Grained Provenance},
  journal      = {Proc. {VLDB} Endow.},
  volume       = {11},
  number       = {11},
  pages        = {1536--1549},
  year         = {2018},
  url          = {http://www.vldb.org/pvldb/vol11/p1536-muller.pdf},
  doi          = {10.14778/3236187.3236204},
  bibsource    = {dblp computer science bibliography, https://dblp.org}
}

@article{shapiro1965analysis,
  title={An analysis of variance test for normality (complete samples)},
  author={Shapiro, Samuel Sanford and Wilk, Martin B},
  journal={Biometrika},
  volume={52},
  number={3-4},
  pages={591--611},
  year={1965},
  publisher={Oxford University Press}
}

@article{mann1947test,
  title={On a test of whether one of two random variables is stochastically larger than the other},
  author={Mann, Henry B and Whitney, Donald R},
  journal={The annals of mathematical statistics},
  pages={50--60},
  year={1947},
  publisher={JSTOR}
}

@inproceedings{begoli2018apache,
  title={Apache calcite: A foundational framework for optimized query processing over heterogeneous data sources},
  author={Begoli, Edmon and Camacho-Rodr{\'\i}guez, Jes{\'u}s and Hyde, Julian and Mior, Michael J and Lemire, Daniel},
  booktitle={Proceedings of the 2018 International Conference on Management of Data},
  pages={221--230},
  year={2018}
}

@article{lei2024spider,
  title={Spider 2.0: Evaluating language models on real-world enterprise text-to-sql workflows},
  author={Lei, Fangyu and Chen, Jixuan and Ye, Yuxiao and Cao, Ruisheng and Shin, Dongchan and Su, Hongjin and Suo, Zhaoqing and Gao, Hongcheng and Hu, Wenjing and Yin, Pengcheng and others},
  journal={arXiv preprint arXiv:2411.07763},
  year={2024}
}

@article{li2024can,
  title={Can llm already serve as a database interface? a big bench for large-scale database grounded text-to-sqls},
  author={Li, Jinyang and Hui, Binyuan and Qu, Ge and Yang, Jiaxi and Li, Binhua and Li, Bowen and Wang, Bailin and Qin, Bowen and Geng, Ruiying and Huo, Nan and others},
  journal={Advances in Neural Information Processing Systems},
  volume={36},
  year={2024}
}

@inproceedings{ghosh-etal-2025-sqlgenie,
    title = "{SQLG}enie: A Practical {LLM} based System for Reliable and Efficient {SQL} Generation",
    author = "Ghosh, Pushpendu  and
      Jain, Aryan  and
      Yenigalla, Promod",
    editor = "Rehm, Georg  and
      Li, Yunyao",
    booktitle = "Proceedings of the 63rd Annual Meeting of the Association for Computational Linguistics (Volume 6: Industry Track)",
    month = jul,
    year = "2025",
    address = "Vienna, Austria",
    publisher = "Association for Computational Linguistics",
    url = "https://aclanthology.org/2025.acl-industry.71/",
    doi = "10.18653/v1/2025.acl-industry.71",
    pages = "1004--1012",
    ISBN = "979-8-89176-288-6"
}

@misc{ye2026texttosqlllmsreallydebug,
      title={Beyond Text-to-SQL: Can LLMs Really Debug Enterprise ETL SQL?}, 
      author={Jing Ye and Yiwen Duan and Yonghong Yu and Victor Ma and Yang Gao and Xing Chen},
      year={2026},
      eprint={2601.18119},
      archivePrefix={arXiv},
      primaryClass={cs.AI},
      url={https://arxiv.org/abs/2601.18119}, 
}

@inproceedings{
pourreza2025chasesql,
title={{CHASE}-{SQL}: Multi-Path Reasoning and Preference Optimized Candidate Selection in Text-to-{SQL}},
author={Mohammadreza Pourreza and Hailong Li and Ruoxi Sun and Yeounoh Chung and Shayan Talaei and Gaurav Tarlok Kakkar and Yu Gan and Amin Saberi and Fatma Ozcan and Sercan O Arik},
booktitle={The Thirteenth International Conference on Learning Representations},
year={2025},
url={https://openreview.net/forum?id=CvGqMD5OtX}
}

@article{gao2024text,
  title={Text-to-SQL Empowered by Large Language Models: A Benchmark Evaluation},
  author={Gao, Dawei and Wang, Haibin and Li, Yaliang and Sun, Xiuyu and Qian, Yichen and Ding, Bolin and Zhou, Jingren},
  journal={Proceedings of the VLDB Endowment},
  volume={17},
  number={5},
  pages={1132--1145},
  year={2024},
  publisher={VLDB Endowment}
}

@article{talaei2024chess,
  title={Chess: Contextual harnessing for efficient sql synthesis},
  author={Talaei, Shayan and Pourreza, Mohammadreza and Chang, Yu-Chen and Mirhoseini, Azalia and Saberi, Amin},
  journal={arXiv preprint arXiv:2405.16755},
  year={2024}
}

@inproceedings{li2023resdsql,
  title={Resdsql: Decoupling schema linking and skeleton parsing for text-to-sql},
  author={Li, Haoyang and Zhang, Jing and Li, Cuiping and Chen, Hong},
  booktitle={Proceedings of the AAAI Conference on Artificial Intelligence},
  volume={37},
  number={11},
  pages={13067--13075},
  year={2023}
}

@article{pourreza2023din,
  title={Din-sql: Decomposed in-context learning of text-to-sql with self-correction},
  author={Pourreza, Mohammadreza and Rafiei, Davood},
  journal={Advances in neural information processing systems},
  volume={36},
  pages={36339--36348},
  year={2023}
}
